\documentclass[12pt,english]{article}

\usepackage[latin9]{inputenc}
\usepackage{geometry}
\usepackage{color}
\usepackage{babel}
\usepackage{booktabs}
\usepackage{mathtools}
\usepackage{amsmath}
\usepackage{amsthm}
\usepackage{amssymb}
\usepackage[multiple]{footmisc}
\usepackage{graphicx}
\usepackage{comment}
\usepackage{setspace}
\usepackage{diagbox}
\usepackage{adjustbox}
\usepackage{booktabs}
\usepackage{multirow}
\usepackage{advdate}
\usepackage[authoryear,round,longnamesfirst]{natbib}
\usepackage[unicode=true,
 bookmarks=true,bookmarksnumbered=true,bookmarksopen=false,
 breaklinks=false,pdfborder={0 0 0},pdfborderstyle={},backref=false,colorlinks=true]
 {hyperref}
\hypersetup{
 linkcolor=purple, citecolor=purple, urlcolor=purple, filecolor=purple}

\usepackage{tikz}
\usetikzlibrary{shapes,trees}
\makeatletter

\theoremstyle{plain}
\newtheorem{assumption}{\protect\assumptionname}
\theoremstyle{definition}
\newtheorem{defn}{\protect\definitionname}
\theoremstyle{plain}
\newtheorem{lem}{\protect\lemmaname}
\theoremstyle{plain}
\newtheorem{thm}{\protect\theoremname}
\theoremstyle{plain}
\newtheorem{prop}{\protect\propositionname}
\theoremstyle{plain}
\newtheorem{property}{\protect\propertyname}
\theoremstyle{plain}
\newtheorem{cor}{\protect\corollaryname}
\theoremstyle{definition}
 \newtheorem{example}{\protect\examplename}
\DeclareMathOperator*{\argmax}{arg\,max}
\DeclareMathOperator*{\argmin}{arg\,min}

\usepackage{amssymb,amsfonts,amsmath,amsthm,mathrsfs}

\usepackage[capitalize,nameinlink]{cleveref} 

\theoremstyle{definition}
\newtheorem{examplecont}{Example}
\newenvironment{excont}[1]
  {\renewcommand{\theexamplecont}{\ref{#1}, cont}%
   \begin{examplecont}}
  {\end{examplecont}}

\makeatother

\providecommand{\assumptionname}{Assumption}
\providecommand{\propertyname}{Property}
\providecommand{\corollaryname}{Corollary}
\providecommand{\definitionname}{Definition}
\providecommand{\propositionname}{Proposition}
\providecommand{\examplename}{Example}
\providecommand{\lemmaname}{Lemma}
\providecommand{\theoremname}{Theorem}
\DeclareMathOperator\supp{supp}

\newcommand*{\dprime}{{\prime\prime}\mkern0mu}

\begin{document}
\title{Robust Misspecified Models and Paradigm Shifts}
\author{Cuimin Ba\thanks{University of Pennsylvania. Email: cuiminba@sas.upenn.edu. I am deeply indebted to George Mailath, Aislinn Bohren, and Kevin He for their guidance and support at every stage of this paper. I thank Nageeb Ali, Ben Brooks, Armin Falk, Hanming Fang, Mira Frick, Drew Fudenberg, Yuan Gao, Alice Gindin,  Marina Halac, Takuma Habu, Daniel Hauser, Ju Hu, Yuhta Ishii, Nawaaz Khalfan, Botond K\"{o}szegi, Changhwa Lee, Jonathan Libgober, Xiao Lin, Ce Liu, Steven Matthews, Guillermo Ordo\~{n}ez, Pietro Ortoleva, Wolfgang Pesendorfer, Andrew Postlewaite, Alvaro Sandroni, Pedro Solti, Juuso Toikka, Marcus Tomaino, Rakesh Vohra, Xi Weng, Ece Yegane and conference audiences for helpful comments and suggestions.}} %
\date{March 23, 2023\medskip{} 
\\
{\small{}\href{https://cuiminba.github.io/Job-Market/Job\%20Market\%20Paper_Cuimin\%20Ba.pdf}{-- Click here to see the most recent version --}}}
\maketitle
\begin{abstract}\begin{singlespace}
Individuals use models to guide decisions, but many models are wrong. This paper studies which misspecified models are likely to persist when individuals also entertain alternative models. Consider an agent who uses her model to learn the relationship between action choices and outcomes. The agent exhibits sticky model switching, captured by a threshold rule such that she switches to an alternative model when it is a sufficiently better fit for the data she observes. The main result provides a characterization of whether a model persists based on two key features that are straightforward to derive from the primitives of the learning environment, namely, the model's asymptotic accuracy in predicting the equilibrium pattern of observed outcomes and the `tightness' of the prior around this equilibrium. I show that misspecified models can be robust in that they persist against a wide range of competing models---including the correct model---despite individuals observing an infinite amount of data. Moreover, simple misspecified models with entrenched priors can be even more robust than correctly specified models. I use this characterization to provide a learning foundation for the persistence of systemic biases in two applications. First, in an effort-choice problem, I show that overconfidence in one's ability is more robust than underconfidence. Second, a simplistic binary view of politics is more robust than the more complex correct view when individuals consume media without fully recognizing the reporting bias. 
\end{singlespace}

\end{abstract}

\thispagestyle{empty}



\newpage

\section{\label{sec:Introduction}Introduction}

\setcounter{page}{1}

People use models to guide decision making, but the subjective nature of models suggests that model misspecification can be pervasive. Model misspecification can stem from the need to simplify the complex world as well as from behavioral biases such as overconfidence or correlation neglect. To explore how misspecified models impact beliefs and actions, the growing literature on misspecified learning focuses on the case of a dogmatic agent who uses a particular misspecified model and never considers changing this model.\footnote{\label{footnote:literature}Examples include: a monopolist trying to estimate the slope of the demand function when the true slope lies outside of the support of his prior \citep{nyarko1991learning,fudenberg2017active}; agents learning from private signals and other individuals' actions while neglecting the correlation between the observed actions \citep{eyster2010naive,ortoleva2015overconfidence,bohren2016informational} or overestimating how similar others' preferences are to their own \citep{gagnon2017taste}; overconfident agents falsely attributing low outcomes to an adverse environment \citep{heidhues2018unrealistic,heidhues2019overconfidence,bagindin2021overconfidence}; a decision maker imposing false causal interpretations on observed correlations \citep{spiegler2016bayesian,spiegler2019behavioral,spiegler2020can}; a gambler who flips a fair coin mistakenly believing that future tosses must exhibit systematic reversal \citep{rabin2010gambler,he2022mislearning}; individuals narrowly focusing their attention on only a few aspects rather than a complete state space \citep{mailath2019wisdom}.} While this simplifies the environment in a way that yields tractable characterizations of long-run beliefs, it leaves open the question of whether it is realistic to expect an agent to never abandon a wrong model. 

A plethora of evidence suggests that people often switch models when an alternative seems more compelling. For example, scientists adopt a new paradigm if it fits the observable data significantly better in terms of accuracy and simplicity (i.e., \citeauthor{kuhn1962}\textcolor{purple}{'s (}\citeyear{kuhn1962}\textcolor{purple}{)} theory of paradigm shifts). One classic example is the paradigm shift from the Ptolemaic model to the Copernican model in astronomy. 
Likewise, economists develop and adopt new models when evidence comes to light that important economic forces are missing from old models. Outside the realm of science, people alter their subjective assumptions about the world in daily life, such as changing thinking patterns in cognitive behavioral therapy or striving to overcome implicit biases through introspection \citep{ wegener1997flexible,di2015learning}. They are also influenced by and attracted to different political narratives as they receive more information \citep{fisher1985narrative,braungart1986life}.

If individuals consider switching to competing models, which (if any) misspecified models should we expect to persist and when? This paper proposes a novel learning framework to address this question. In this framework, an agent uses models to learn an unknown fixed data-generating process that governs the relationship between her action choices and random outcomes. Each model is a parametric theory of how actions affect the outcome distribution. Formally, this consists of a finite parameter space and a collection of possible data-generating processes indexed by the parameter values. 
For example, consider a monopolist who chooses production quantities based on a linear model of consumer demand. Here, for each pair of parameter values---fixing the slope and intercept of the demand curve---the model prescribes a mapping from production quantities to  distributions of consumer demand. While the \textit{dogmatic modeler} typically considered in the misspecified learning literature uses the same model throughout, in my framework the agent is a \textit{switcher} who entertains multiple models. For the main analysis, this agent starts with an initial model and entertains exactly one competing model. 
She subscribes to one model at any point in time to guide decisions but may switch to the other model depending on the history of outcome realizations.
Specifically, the agent has a prior over the parameters within each model, uses her current model to update her belief given the data she observes, and plays the optimal action given the associated posterior in the current model. To decide whether to switch to a competing model, the agent keeps track of the \emph{Bayes factor}---the likelihood ratio of the competing model relative to the current model given the observed data---and switches if it exceeds a fixed switching threshold. As the switching threshold increases, model switching requires more evidence and becomes stickier.  

One may wonder why the agent does not conduct Bayesian updating over the models and average out their predictions. As pointed out in \citeauthor{savage1972foundations}\textcolor{purple}{'s (}\citeyear{savage1972foundations}\textcolor{purple}{)} \textit{Foundations of Statistics}, Bayesianism is a reasonable description of human behavior only when decision makers focus on ``modest little worlds.''\footnote{\citet[p.~16]{savage1972foundations} describes it as ``utterly ridiculous'' to demand that ``one envisage every conceivable policy for the government of his whole life (at least from now on) in its most minute details, in the light of the vast number of unknown states of the world, and decide here and now on one policy.''} It is implausible to expect people to attach well-articulated probabilities to all models of the world given that models contain conceptually incoherent or even conflicting ideas, such as a geocentric model versus a heliocentric model, or a liberal worldview versus a conservative worldview. 
Therefore, this paper adopts a natural choice of non-Bayesian rule based on the Bayes factors. 

Within this framework, I next formalize notions of model persistence and robustness. A model is said to \emph{persist} against a competing model if, with positive probability, the agent eventually stops switching and sticks to this model forever. Further, a model is {robust} if it persists against a wide range of competing models. Critically, the scope of robustness may vary considerably with the set of admissible competing models in terms of their distance from the initial model (i.e. step size of switching) and the dimensions along which they differ from the initial model (i.e. direction of switching). 
I consider three different robustness notions with varying restrictions on the step size and the direction of switching. Fixing a prior, a model is \emph{globally robust} if it persists against any competing model associated with any prior over the parameter space, including the true model; it is \emph{locally robust} if it persists against local perturbations with similar predictions and similar priors; and finally, \emph{constrained locally robust} if it  persists against structured local perturbations that are constrained to a particular parametric family.  Which robustness notion is the most appropriate depends on the context of different applications. For example, global robustness is most appropriate when the agent can ``think outside the box'', potentially due to 
a vibrant culture of innovation, whereas local notions are more suitable when the agent is conservative and only takes small steps. The constrained notion is appropriate when the agent does not completely forgo structural assumptions of the model, such as linearity or normality, and only contemplates competing models that lie in restricted directions relative to the original one. These notions provide a language to compare the robustness properties of models across different environments, and their formalization is a central conceptual contribution of my framework. 

My main result characterizes when a model satisfies each robustness notion based on two properties that are easily derived from the primitives of the model: asymptotic accuracy and prior tightness, as summarized in Table \ref{tab:my_label}. A model has high asymptotic accuracy when it gives rise to a \textit{self-confirming equilibrium} that satisfies a stability condition called \textit{p-absorbingness}. In a self-confirming equilibrium, the agent holds a supporting belief over the model parameters such that the model prediction perfectly coincides with the objective outcome distribution. The stability condition requires that a dogmatic modeler who only uses this model eventually only plays actions in the support of the equilibrium with positive probability. With high asymptotic accuracy, the model has weakly higher explanatory power than any other competing model with positive probability in the limit. However, this limit condition alone does not imply persistence, because the learning dynamics may induce the agent to switch away before she comes close to the equilibrium belief. When the associated prior is tight in the sense of being concentrated around the set of p-absorbing self-confirming equilibria, the explanatory power of the model remains consistently high across all periods, leading to persistence. 

\begin{table}[]
    \centering
        \begin{tabular}{c c c c} \toprule \vspace{-0.2em}
         & \multicolumn{3}{c}{Notions of robustness}  \\ 
            \cmidrule(lr){2-4} \vspace{-1.5em}\\ 
        Properties  & \hspace{0.5em} global\hspace{0.5em} &\hspace{0.5em} local\hspace{0.5em} & constrained local\\ \hline 
        asymptotic accuracy &  {high} &{high} & {low}\\
        prior tightness  & {high} & {none} &  {none}\\\bottomrule
    \end{tabular}
    \caption{Summary of results.}
    \label{tab:my_label}
\end{table}

I start by characterizing \textit{which} models can be locally or globally robust for at least one full-support prior. It is natural to conjecture that local robustness is much weaker than global robustness. Surprisingly, Theorem \ref{thm:Global-Robustness} establishes that a model is globally robust for at least one prior if and only if it is locally robust for at least one prior, and both amount to a requirement for high asymptotic accuracy. In this regard, persisting against local perturbations is as demanding as persisting against fundamentally different paradigms. The intuition for this result is that, when high asymptotic accuracy fails, the model can be improved locally by perturbing all of its predictions towards the direction of the true data-generating process. Moreover, Theorem \ref{thm:Global-Robustness} holds for all levels of the switching threshold. Even if the agent is extremely reluctant to switch, the accumulation of evidence over time eventually leads to the abandonment of a less accurate model. Note that, however, high asymptotic accuracy does not equate to high efficiency---strictly suboptimal actions can be played in a self-confirming equilibrium if the model yields wrong predictions off-path. 

Next, I characterize \textit{when}, or under which priors, the models characterized earlier are locally or globally robust. Theorem \ref{thm:concentrated-prior} reveals the real distinction between global and local robustness: the former requires high prior tightness but the latter does not. I provide an exact closed-form quantification of the required level of tightness in terms of the switching threshold. Specifically, the prior probability assigned to the parameters associated with the relevant self-confirming equilibria must exceed the inverse of the switching threshold. Therefore, higher switching stickiness facilitates robustness not by expanding the set of qualified models, but by allowing for more diffuse priors.  

By contrast, neither high asymptotic accuracy nor prior tightness is necessary for constrained local robustness (Theorems \ref{thm:Local-lobustness-w-convergence} and \ref{thm:Local-robustness-sufficient}). Rather, constrained local robustness only requires the model to have asymptotic accuracy relatively higher than the admissible structured local perturbations, captured by the existence of a p-absorbing Berk-Nash equilibrium at which the model is \textit{locally dominant} with respect to the constrained family. As the family enlarges to the universe of all models, as expected, such a Berk-Nash equilibrium morphs into a self-confirming equilibrium. 

This general characterization provides a learning foundation for the persistence of certain forms of model misspecification. A misspecified model can be robust against \emph{arbitrary} competing models---including the true model---despite the agent having an infinite amount of data. When competing models are further constrained, misspecified models can be robust even with lower asymptotic accuracy. The results provide off-the-shelf tools to predict which underlying behavioral biases are more relevant in specific contexts. In an application, I show that overconfidence is globally robust while underconfidence is non-robust for a wide range of parameters. This characterization can be particularly useful in generating testable predictions for empirical research and suggesting behavioral policies to mitigate the consequences of misspecification. For example, the application suggests that underconfidence requires less intervention than overconfidence as it is self-correcting. 

This characterization also yields novel predictions about how qualitative features of the model and the switching environment affect persistence. First, an interesting contrast emerges between the robustness properties between misspecified and correctly specified models. On one hand, all correctly specified models have the high asymptotic accuracy that only a subset of misspecified models can achieve. On the other hand, a correctly specified model may be associated with a less tight prior as compared to a misspecified one, especially when the latter has a smaller parameter space or gives rise to a large number of self-confirming equilibria. Perhaps surprisingly, this implies that some misspecified models can be more robust than correctly specified models, precisely because they are sufficiently extreme and misleading. Second, lower switching stickiness can be a double-edged sword, since it makes global robustness harder to obtain for any model. While it is easier to switch away from a misspecified model, due to noisy information, it is also easier to abandon a correctly specified model and get stuck with a misspecified alternative. I apply these observations to a media consumption problem and demonstrate that an extreme binary worldview can outperform and even replace the correct worldview, leading to persistent polarization in political beliefs. 

The rest of the paper is organized as follows. The next subsection discusses related literature.  Section \ref{sec:Example} provides an illustrative example. Section \ref{sec:Framework} introduces the learning framework. Section \ref{sec:Robustness} characterizes local and global robustness, and Section \ref{sec:constrained local robustness} characterizes constrained local robustness. Next, Section \ref{sec:Applications} discusses two applications of the characterization results. Section \ref{sec:extension} discusses the extensions of the model switching framework and Section \ref{sec:Concluding-Remarks} concludes. Appendix \ref{sec:Auxiliary-Results} contains useful auxiliary results, Appendix \ref{sec:Proofs-Main-Results} includes all proofs of the main results, and Appendix \ref{sec:Online-Appendix} contains examples omitted from the main text.

\subsection*{Related Literature}

This paper builds on the growing literature on learning with misspecified models. Much of that literature mentioned in Footnote \ref{footnote:literature} focuses on case-by-case analyses of misspecified models. In contrast, this paper characterizes robust misspecified models in general environments with simple conditions on the primitives, providing a microfoundation for their persistence. Another strand of this literature studies equilibrium concepts when the decision maker is a dogmatic modeler. \citet{esponda2016berk} propose the concept of Berk-Nash equilibrium, generalizing the self-confirming equilibrium \citep{Battigalli1987,fudenberg1993self} by relaxing the requirement that the subjective prediction fully coincides with the objective reality.\footnote{Other related concepts include the analogy-based expectation equilibrium in \cite{jehiel2008revisiting} and the cursed equilibrium in \cite{eyster2005cursed}. As pointed out by \cite{esponda2016berk}, these two solution concepts coincide with Berk-Nash or self-confirming equilibrium under appropriately specified feedback structures.} This paper deepens our understanding of the distinction and the connection between the two solution concepts by relating self-confirming equilibrium to unconstrained robustness notions and Berk-Nash equilibrium to constrained robustness. 

Recent contributions to the misspecified learning literature focus on characterizing the asymptotic learning outcomes of dogmatic modelers in general environments. This paper faces many of the same technical challenges as the works in this area since model persistence partly hinges on the asymptotic behavior of a dogmatic modeler. \cite{bohren2021learning} provide criteria for local and global stability of strict Berk-Nash equilibria in settings with model heterogeneity. \citet{esponda2021asymptotic} find conditions for a single agent's action frequency to converge to a Berk-Nash equilibrium using tools from stochastic approximation. \citet{fudenberg2021limit} establish that a uniformly strict Berk-Nash equilibrium is uniformly stable in the sense that starting from any prior that is sufficiently close to the equilibrium belief, the dogmatic modeler's action converges to the equilibrium with arbitrarily high probability. In this paper, I show that robustness is related to p-absorbingness, a different stability notion that does not require the dogmatic modeler's action to converge, but her action to enter and eventually stay within the support of an equilibrium. \citet{frick2021belief} introduce a prediction accuracy ordering over parameters within a model and use martingale arguments to prove belief convergence. In Section \ref{sec:constrained local robustness}, I use a similar technique to compare prediction accuracy across models. This paper makes a technical contribution to the literature by integrating model switching into a misspecified learning framework. Given that the agent has access to multiple models, we need to keep track of multiple belief processes, all of which are generated by endogenous data. Most importantly, the Bayes factor that governs the model switching process interacts and correlates with all belief processes even when no switching has been made. In Section \ref{subsec:traps}, I show such interaction can cause the adoption of some models to be self-defeating; I resolve this difficulty by proposing a set of no-trap conditions.

This paper is most related to papers that explore why misspecified models persist.  \citet{cho2015learning} also study model switching with endogenous data but assume a different switching rule. They assume that the agent always compares the model prediction with the empirical realizations and characterize dominant models based on the speed of convergence to a self-confirming equilibrium. By contrast, a key conceptual innovation of my framework is to view model robustness as a relative concept. The agent in my framework decides whether to switch by comparing the prediction quality of her model relative to that of a competing model. \citet{olea2019competing} characterize the ``winner\textquotedblright{} model in a contest environment where agents use models to predict an exogenous data-generating process and make auction bids based on their subjective model prediction error. They identify a trade-off between model fit and model estimation uncertainty when the dataset is small. In contrast, I show that both asymptotic accuracy and prior tightness are important for robustness even with an infinite amount of data. \citet{gagnon2018channeled} study the stability of models when the agent entertains a correctly specified alternative model. They examine a setting where data is independent from actions but the agent only pays attention to the data they deem decision-relevant given the current model. This contrasts with my framework where data is endogenously generated by actions but the agent pays attention to all available data.

Other papers approach this problem from an evolutionary or welfare perspective. \citet{fudenberg2022misspecifications} study the evolutionary dynamics when a small share of a large population mutates to enlarge their models at a Berk-Nash equilibrium. They find that an equilibrium can resist  mutations that yield a better statistical fit but induce worse-performing actions. My results complement their findings by showing that a model can persist against competing models that have lower asymptotic accuracy but nevertheless induce better-performing actions. Moreover, similar to this paper, they show that a self-confirming equilibrium resists every mutation.\footnote{Despite the similarity, the underlying mechanisms of our results are different. Their result follows from the assumption that every mutation is an expansion of the original model. Since the same self-confirming equilibrium remains an equilibrium under the mutated model, it is possible for all individuals in the population to stick to the same behavior as before the mutation.} 
\citet{he2020evolutionarily} also evaluate competing models based on their expected objective payoffs but examine multi-agent strategic games where misspecification can lead to beneficial misinferences. \citet{frick2021welfare} study welfare comparisons of learning biases and find that some biases can outperform Bayesian updating because they may lead to correct learning at a faster speed. 

An extensive literature in decision theory studies the behavior of a decision maker who has access to multiple models or priors over states. \citet{ortoleva2012modeling} proposes and axiomatically characterizes an amendment to Bayes' rule, called the Hypothesis Testing model, where the agent switches to an alternative prior upon observing an event to which she assigns a probability below a threshold. This contrasts with my framework where the agent switches if the ratio of the probability of the observed outcomes under the current model \textit{relative to} the competing model is sufficiently low.\footnote{In my framework, even if a sequence of outcomes occur with low probability according to the current model, the agent may not switch if the probability of those outcomes under the competing model is also low.} \citet{karni2013reverse} provide a choice-based decision theory to model a self-correcting agent who can expand his universe of subjective states. A number of canonical decision criteria capture aversion to model uncertainty, e.g. the maxmin model \citep{GILBOA1989141}, the smooth ambiguity aversion model \citep{klibanoff2005smooth}, and the robust control model \citep{hansen2001robust}. \cite{battigalli2019learning} find that ambiguity aversion increases the likelihood that the decision maker gets stuck into a self-confirming equilibrium. The key distinction is that the agent in my framework does not have a prior over models; instead, she switches between models based on their statistical fit.

A few papers entertain the idea that people can change models and explore its implications in various settings. \citet{mullainathan2002thinking} presents a model of ``categorical thinking'' in which people switch between coarse categories and policies discontinuously, resulting in overreaction to news.  \citet{galperti2019persuasion} and \citet{schwartzstein2019using} extend the idea of alterable subjective models to a persuasion setting and study how a principal could persuade an agent to accept a different worldview.

Finally, this paper is also related to the statistics literature on model selection. Statisticians have developed a number of criteria that differ in their cost of computation and penalty for overfitting, such as the Bayes factor, Akaike information criterion (AIC), Bayesian information criterion (BIC), and likelihood-ratio test (LR test), and the machine learning community favors cross-validation due to its flexibility and ease of use \citep{chernoff1954distribution,akaike1974new,stone1977asymptotic,schwarz1978estimating,kass1995bayes,konishi2008information}. This work focuses on the Bayes factor rule and differs with the statistical literature by studying an endogenous data-generating process. I will come back to the comparison of different model selection rules in Section \ref{subsec:Discussion-model}.

\section{\label{sec:Example}Illustrative Example}
As a simple illustration of the learning framework and the main results, consider the following example.\footnote{I build this example on \cite{heidhues2018unrealistic}. In Section \ref{sec:Applications}, I extend this example to allow for more
general payoffs and outcome distributions.}
An artist chooses how much effort to exert in creating artwork in every period, $a_t\in\{0,1,2\}$. Upon exerting effort, he incurs a cost $a_t(a_t+0.5)$ but also obtains revenue from the sales of his work. The sales revenue takes a simple form: $y_t = (a_t+b)\omega+\epsilon_t$, where $b$ is the talent level of the artist, $\omega$ is an unknown market demand for arts, and $\epsilon_t$ captures a zero-mean random noise with a known distribution. Suppose that the true talent is $1$ and the true market demand is $2$. 
If the artist has a correct belief about talent, he is able to correctly infer the market demand from the sales data, allowing him to eventually choose the efficient effort level of $1$.\footnote{More precisely, if the artist is correct about his talent, a sufficient and necessary condition for correct learning is that the artist's prior assigns positive probability to the true market demand.} 
However, the artist is endowed with a misspecified model: while the artist knows the sales function, he is overconfident or underconfident about his talent. In such a model, the unknown market demand is a parameter to be estimated.\footnote{\cite{heidhues2018unrealistic} consider a dogmatic agent who never changes his model and show that both over- and underconfidence lead to wrong inferences about market demand and inefficient choices of effort in the long run.} This modeling approach captures the idea that individuals often commit fundamental attribution errors and thus are slower in changing self-perceptions than in updating beliefs about the outside environment \citep{miller1975self}.

Suppose the artist entertains competing models, are underconfidence and overconfidence equally likely to persist?  My results reveal an interesting asymmetry---overconfidence tends to be more robust than underconfidence. This prediction is consistent with a large amount of empirical evidence that people usually exhibit overconfidence instead of underconfidence \citep{langer1975heads}.

Let's first consider the case of an underconfident artist who believes in a talent level of $\hat b=0$ but entertains a correctly specified competing model that attaches positive probability to the true talent $b=1$. 
The artist is sticky in changing his model---he only switches to an alternative model if that explains the observed sales data significantly better. We can show that the underconfidence model does not persist because it has strictly lower asymptotic accuracy than the competing model. To see why, first note that underconfidence induces the artist to mistakenly attribute the higher-than-expected sales to a high market demand and thus incentivizes a high effort. What comes next is critical: a high effort choice then in turn induces a lower belief, correcting the artist's overestimation---due to complementarity, the marginal return to demand increases as effort rises, so a smaller inference-truth gap suffices to explain the sales data. However, repeating this logic, a lower belief in demand then decreases the effort and pushes up the belief, generating a negative feedback loop.  Specifically, an effort of $1$ pushes the artist's belief towards a demand of $4$ at which a higher effort is optimal, but an effort of $2$ pushes the belief towards a demand of $3$ at which a lower effort is optimal,
\begin{align*}
    {(\hat a^1+\hat b)\cdot\hat\omega^1} = (1+0)\cdot 4 = {(\hat a^1+ b)\cdot \omega}= (1+1)\cdot 2 = 4, \\
    {(\hat a^2+\hat b)\cdot\hat\omega^2} = (2+0)\cdot 3 = {(\hat a^2+ b)\cdot \omega}= (2+1)\cdot 2 = 6. 
\end{align*}
Consequently, the artist's effort cycles between $1$ and $2$ forever, and no single value of market demand can explain all data perfectly---the model admits no self-confirming equilibrium.
By contrast, a correctly specified competing model can always achieve high prediction accuracy in the limit. Therefore, no matter how reluctant the artist is to change his self-perception, he is going to abandon the initial model and correct his underconfidence.  

Next, let us turn to an overconfident artist who believes his talent level is instead given by $\hat b=2$ but also entertains a correctly specified competing model. In contrast to the previous case, the overconfidence model has high asymptotic accuracy. Overconfidence induces the artist to mistakenly attribute the disappointing sales to low demand and respond by exerting a low effort. Critically, a lower effort choice then induces an even lower belief---the marginal return to market demand decreases as effort drops, so a larger inference-truth gap is necessary to make sense of the disappointing sales. The positively reinforcing dynamics eventually lead the artist to reach the false conclusion that the market demand is $\hat \omega=1$ and play an inefficient choice of effort $\hat a=0$. Most importantly, they constitute a {self-confirming equilibrium}: zero effort is optimal against the misguided belief about the market demand, and this low belief perfectly matches the sales data given the misspecified model, since
\begin{equation*}
    {(\hat a+\hat b)\cdot\hat\omega} = (0+2)\cdot 1 = {(\hat a+ b)\cdot \omega}= (0+1)\cdot 2 = 2.
\end{equation*}
At this steady state, the initial model and the competing model generate equally accurate predictions; given the friction in model switching, this suggests that the artist sticks with the overconfidence model with positive probability. But this is not the end of the story yet---the equilibrium analysis suggests that overconfidence has the potential to persist in the limit but does not rule out switches in the course of converging to the steady state. The dynamic model switching framework introduced in Section \ref{sec:Framework} addresses this concern. 
My characterization in Section \ref{sec:Robustness} implies that the overconfidence model is globally robust (and thus persists against the correctly specified competing model) when the associated prior attaches sufficiently high probability to the low demand $\hat\omega=1$. Moreover, the required prior tightness is inversely related to the switching stickiness.


\section{\label{sec:Framework}Framework}

\subsection{Setup}

\paragraph{Objective Environment.}
Consider an infinitely repeated decision problem with a myopic agent.\footnote{I allow the agent to be non-myopic  Section \ref{subsec:patient agent} and show that all main results continue to hold.} In each period $t=0,1,2,...$, the agent chooses an action $a_{t}$ from a finite set $\mathcal{A}$ and then observes a random outcome $y_{t}$ from $\mathcal{Y}$. The set of possible outcomes $\mathcal Y$ is either an Euclidean space or a compact subset of an Euclidean space with at least two elements. The agent's action may affect the distribution of the outcome: conditional on $a_{t},$ outcome $y_{t}$ is independently drawn according to probability measure $Q^{\ast}\left(\cdot|a_{t}\right)\in\Delta\mathcal{Y}$. This true data generating process (henceforth true DGP) remains fixed throughout. At the end of period $t$, the agent obtains a flow payoff $u_{t}\coloneqq u\left(a_{t},y_{t}\right)\in\mathbb{R}$. The function $u$ is known to the agent. I denote the observable history in the beginning of period $t$ by $h_{t}\coloneqq\left(a_{\tau},y_{\tau}\right)_{\tau=0}^{t-1}$ and the set of all such histories by $H_{t}=\left(\mathcal{A}\times\mathcal{Y}\right)^{t}$. I make the following assumptions on the true DGP and the payoff function.
\begin{assumption} \label{assu:Objetive-DGP} (i) For all $a\in\mathcal{A}$, $Q^{\ast}\left(\cdot|a\right)$ is absolutely continuous w.r.t. a common measure $\nu$, and the Radon-Nikodym derivative $q^{\ast}\left(\cdot|a\right)$ is positive and continuous; (ii) For all $a\in\mathcal{A}$, $u\left(a,\cdot\right)\in L^{1}\left(\mathcal{\mathcal{Y}},\mathbb{R},Q^{\ast}\left(\cdot|a\right)\right)$.\footnote{$L^{p}\left(\mathcal{\mathcal{Y}},\mathbb{R},\nu\right)$ denotes the space of all functions $g:\mathcal{Y}\rightarrow\mathbb{R}$ s.t. $\int\left|g\left(y\right)\right|^{p}\nu\left(dy\right)<\infty$.}
\end{assumption}
These assumptions are standard in the literature. Assumption \ref{assu:Objetive-DGP}(i) means that the true DGP admits a positive and continuous density. When $\mathcal{Y}$ is discrete, $q^{\ast}\left(\cdot|a\right)$ is the probability mass function and $\nu$ is the counting measure; when $\mathcal{Y}$ is a continuum, $q^{\ast}\left(\cdot|a\right)$ is the probability density function and $\nu$ is the Lebesgue measure. Assumption \ref{assu:Objetive-DGP}(ii) ensures that the agent's objective expected period-$t$ payoff, $\overline{u}_{t}\coloneqq\int_{\mathcal{Y}}u\left(a_{t},y\right)q^{\ast}\left(y|a_{t}\right)\nu\left(dy\right)$, is well-defined and thus there exists an objectively optimal action. 

\paragraph{Subjective Models.}
The decision problem is trivial if the agent knows the true DGP---she simply plays a myopically optimal action every period. However, the agent does not know the true DGP and she uses subjective models to guide her decisions.

A model, indexed by $\theta$, is a theory of how actions affect the outcome distribution. Each model $\theta$ consists of two components: (1) a \emph{parameter set}, denoted by $\Omega^{\theta}$, and (2) \emph{predictions}, which is a collection of data generating processes denoted by $Q^{\theta}:\mathcal{A}\times\Omega^{\theta}\rightarrow\Delta\mathcal{Y}$. Conditional on action $a_t$, model $\theta$ predicts the outcome distribution $Q^\theta(\cdot|a_t,\omega)$, where $\omega$ can take any value in $\Omega^\theta$. Parameters are configuration variables specific to the model. For example, the parameters internal to a linear supply model include only the slope and the intercept, while the parameters of a more complicated model may include a variety of other factors. A model with a larger parameter set includes a larger number of DGPs. 
I restrict attention to models satisfying Assumption \ref{assu:Subjective-DGP}.

\begin{assumption}
\label{assu:Subjective-DGP} For all $a\in\mathcal{A}$: (i) $\Omega^{\theta}$ is a finite subset of some Euclidean space; (ii) for all $\omega\in\Omega^{\theta}$, $Q^{\theta}\left(\cdot|a,\omega\right)$ is absolutely continuous w.r.t. measure $\nu$, and the Radon-Nikodym derivative $q^{\theta}\left(\cdot|a,\omega\right)$ is positive and continuous; (iii) for all $\omega\in\Omega^{\theta}$, $u\left(a,\cdot\right)\in L^{1}\left(\mathcal{\mathcal{Y}},\mathbb{R},Q^{\theta}\left(\cdot|a,\omega\right)\right)$; (iv) for all $\omega\in\Omega^{\theta}$, there exists $r_{a}\in L^{2}\left(\mathcal{Y},\mathbb{R},\nu\right)$ such that $\left|\ln\frac{q^{\ast}\left(\cdot|a\right)}{q^{\theta}\left(\cdot|a,\omega\right)}\right|\leq r_{a}\left(\cdot\right)$ a.s.-$Q^{\ast}\left(\cdot|a\right)$.
\end{assumption}
Assumption \ref{assu:Subjective-DGP}(i) requires that the parameter space is finite. Assumptions \ref{assu:Subjective-DGP}(ii) and \ref{assu:Subjective-DGP}(iii) are analogous to Assumption \ref{assu:Objetive-DGP}. They ensure the existence of a density function  and that the expected payoffs predicted by any model are well-defined. Assumption \ref{assu:Subjective-DGP}(iv) ensures that the difference between the predictions of any model and the true DGP can be properly quantified, which also implies that no subjective models rule out events that occur with positive probability under the true DGP. 

Let $\Theta$ denote the universe of all models that satisfy Assumption \ref{assu:Subjective-DGP}. Since each model is a finite set of DGPs, we have $\Theta\subseteq\cup_{z=1}^{\infty}\left(\left(\Delta\mathcal{\mathcal{Y}}\right)^{|\mathcal{A}|}\right)^z$, where $z$ represents the size of the parameter set. Model $\theta$ is said to be \textit{correctly specified} if its predictions include the true DGP, i.e.  $q^{\ast}\left(\cdot|a\right)\equiv q^{\theta}\left(\cdot|a,\omega\right),\forall a\in\mathcal{A}$ for some $\omega\in\Omega^{\theta}$, and \textit{misspecified} otherwise. I use $\theta^\ast$ to denote the smallest correctly specified model, i.e. the model that solely consists of the true DGP, and refer to $\theta^\ast$ as the \textit{true model}.

\subsection{The Switcher's Problem}

The agent has access to a finite set of models, $\Theta^{\dagger}\subseteq\Theta$. It is often assumed in the literature that the decision maker is a \textit{dogmatic modeler} who uses a single model throughout. I call a dogmatic modeler with $\Theta^{\dagger}=\left\{ \theta\right\}$ a \emph{$\theta$-modeler}. The key departure I take here is to focus on a \textit{switcher} who entertains multiple models. A switcher adopts one model in any period but may switch between multiple models at different times. For the main analysis, I restrict attention to the two-model case where $\Theta^\dagger=\{\theta,\theta^\prime\}$.\footnote{This model switching framework can be extended to allow for three or more models in $\Theta^\dagger$. I analyze this extension in  \cref{sec:multiple-competing-models}.} 
Within a single period, a $\theta$-modeler and a switcher who adopts $\theta$ this period and shares the same belief over its parameters operate in an identical way. 
However, a $\theta$-modeler and a switcher may behave considerably differently across periods, because the decision rule of the switcher is specific to the current model and can be altered after a change of models.


I now describe a switcher's behavior in greater detail. I denote the agent's model choice in period $t$ by $m_{t}\in\Theta^{\dagger}$ and assume she starts with $\theta$, i.e. $m_0=\theta$. We can summarize a switcher's problem by a quadruple, $S=(\theta,\theta^\prime, \pi_0^\theta,\pi_0^{\theta^\prime})$, where the first two elements represent the \emph{initial model} and the \emph{competing model}, respectively, and the last two elements are the priors over the corresponding model's parameters,  $\pi_0^\theta\in\Delta\Omega^\theta$ and $\pi_0^{\theta'}\in\Delta\Omega^{\theta'}$. I assume that all priors are full-support.\footnote{This is without loss of generality. Using model $\theta$ with a non-full-support prior is equivalent to using model $\theta'$ with the smaller parameter space $\Omega^{\theta'}=\supp(\pi_0^\theta)$ and a full-support prior. } I now describe the events in period $t$ in chronological order.\medskip{}

\noindent \textbf{Model switching.} In period $0$, the agent adopts $\theta$ and proceeds immediately to choosing an action. In period $t\geq 1$, the agent employs a \textit{Bayes factor} rule to determine the model choice $m_{t}$. Fix a constant $\alpha>1$ that I call the \textit{switching threshold}. At the beginning of each period $t\geq1$, the agent calculates the Bayes factor, 
\begin{align}
\lambda_{t} & \coloneqq\frac{\ell_{t}(\theta')}{\ell_{t}(\theta)}\coloneqq \frac{\sum_{\omega\in\Omega^{\theta}}\pi_{0}^{\theta}(\omega)\ell_t(\theta,\omega)}{\sum_{\omega'\in\Omega^{\theta'}}\pi_{0}^{\theta'}(\omega')\ell_t(\theta',\omega')},\label{eq:likelihood-ratio}
\end{align}
where
\begin{align}
\ell_t(\theta,\omega)\coloneqq \prod_{\tau=0}^{t-1}q^{\theta}\left(y_{\tau}|a_{\tau},\omega\right),\label{eq:likelihood}
\end{align} and $\ell_{t}(\theta',\omega')$ is defined analogously.
The Bayes factor is the ratio of the likelihood of the data under model $\theta'$ to the likelihood of the data under model $\theta$; further, the likelihood under a model is the weighted sum of the likelihoods of the data under all DGPs included in the model, with the weights given by the prior. If $\theta$ is adopted in the past period, then the agent  makes a switch to $\theta'$ if and only if the Bayes factor exceeds the switching threshold, $\lambda_t> \alpha$. If $\theta'$ is adopted in the past period, the agent makes a switch back to $\theta$ if and only if $\lambda_t$ drops below the inverse of the switching threshold, $\lambda_t<1/\alpha$.\footnote{\noindent The symmetry in the switching threshold is not important for the main results.} If instead $1/\alpha\leq \lambda_t\leq\alpha$, the agent does not find the existing evidence adequate for a model switch. Given a larger $\alpha$, switching requires stronger evidence. Thus, $\alpha$ is a measure of switching \emph{stickiness}. I discuss the switching rule in the next subsection.\medskip{}

\noindent \textbf{Learning.} The agent then updates her belief over the parameters of both models using the full history. I recursively define two belief processes, $\pi_{t}^{\theta}=B^{\theta}(a_{t-1},y_{t-1},\pi_{t-1}^{\theta})$ and $\pi_{t}^{\theta'}=B^{\theta'}(a_{t-1},y_{t-1},\pi_{t-1}^{\theta'})$, where $B^{\theta}:\mathcal{A}\times\mathcal{Y}\times\Delta\Omega^{\theta}\rightarrow\Delta\Omega^{\theta}$ is the Bayesian operator for model $\theta$ and $B^{\theta'}$ is the Bayesian operator for $\theta'$. We can equivalently write the Bayes factor defined in \eqref{eq:likelihood-ratio} in terms of the posteriors,
\begin{equation}
    \lambda_t=\lambda_{t-1} \frac{\sum_{\omega'\in\Omega^{\theta'}}\pi_{t-1}^{\theta'}\left(\omega'\right)q^{\theta'}\left(y_{t-1}|a_{t-1},\omega'\right)}{\sum_{\omega\in\Omega^{\theta}}\pi_{t-1}^{\theta}\left(\omega\right)q^{\theta}\left(y_{t-1}|a_{t-1},\omega\right)}.\label{eq:likelihood-ratio-recursive}
\end{equation}
This expression has an intuitive interpretation: the first term on the right-hand side compares how well the models explain the data generated before the last period, and the second term compares how well they explain the most recent observation when the parameters are estimated using the past data. 

\medskip{}
\noindent \textbf{Actions.} The agent chooses an action according to a fixed pure policy that is optimal under the current model $m_{t}$. The policy under $\theta$, denoted by  $f^\theta:\Delta\Omega^\theta\rightarrow\mathcal A$, is a selection from the correspondence $A_m^{\theta}:\Delta\Omega^\theta\rightrightarrows\mathcal A$ that returns all myopically optimal actions at a given belief. Formally, for any belief $\pi^\theta_t\in\Delta\Omega^\theta$, \begin{equation}
    A^\theta_m(\pi_t^\theta)\coloneqq  \argmax_{a\in\mathcal A}\sum_{\Omega^\theta}\pi_t^\theta(\omega)\int_{\mathcal Y} u(a,y)q^\theta(y|a,\omega)\nu(dy).
\end{equation} Analogously, the policy under $\theta'$, denoted by $f^{\theta'}$, is a selection from $A_m^{\theta'}$. 

\medskip{}


Given $\theta$ and $\theta'$, the agent's model choice either eventually converges to one of the two models or oscillates forever. We say that model $\theta$ persists if the agent eventually stops switching and settles down with model $\theta$ with positive probability.\footnote{I construct the underlying probability space in Appendix \ref{sec:Auxiliary-Results}.} 

\begin{defn}
\label{def:Persistence}Model $\theta$ \textit{persist against $\theta'$} at priors $\pi_0^\theta$ and $\pi_0^{\theta'}$ if, given the switcher's problem $S=(\theta,\theta',\pi_0^{\theta'},\pi_0^{\theta'})$, the model choice $m_{t}$ eventually equals $\theta$ with positive probability, i.e. 
there exists $T>0$ such that with positive probability, $m_t=\theta$ for all $t\geq T$.
\end{defn}

The interpretation of persistence is as follows. If $\theta$ does not persist against $\theta'$, the competing model $\theta'$ will be adopted by the agent infinitely many times almost surely. This implies that the long-term beliefs and behavior of the switcher can be starkly different from the predictions of an analyst who only knows the initial model $\theta$. If an underlying bias gives rise to a model that does not persist, we expect it to be non-stable. By contrast, if $\theta$ persists against $\theta'$, then with positive probability, the switcher resembles a $\theta$-modeler in the long term, allowing us to use tools from the literature on misspecified learning to characterize her behavior. This also implies that we could expect to observe the stable existence of the underlying bias. I discuss alternative ways of defining persistence in Section \ref{sec:extension}.

Note that persistence is prior-sensitive---in principle, a model could persist against some competing model at some priors but not others. The prior potentially affects persistence in two ways. First, the prior plays a direct part in the computation of the Bayes factor. Second, the prior affects the agent's behavior and thus endogenously affects the outcome realizations and how well models fit the data. Therefore, we are interested in not only which models can persist but also how their persistence depends on the prior.

\subsection{\label{subsec:Discussion-model}Discussion}
 
 
I now briefly comment on the model switching rule before proceeding to the analysis. The discussion on other important assumptions is deferred to Section \ref{sec:extension}.\medskip{}

\noindent \textbf{Sticky switching. }A switch only occurs when an alternative model fits sufficiently better, as captured by the assumption that $\alpha>1$. 
Switching stickiness is well observed in reality and can stem from a variety of causes, such as conservatism, concerns about overreaction to noisy observations, or mental and physical costs associated with model switching. For example, universities base their promotion decisions on models of faculty performance that heavily rely on key statistics such as the rankings of the journals. While such an evaluation system can be deeply flawed, implementing a new system is highly costly, and thus a model switch only occurs when a meta-analysis of other potential evaluation systems points to a clear winner. In the statistics literature on Bayesian model selection, \cite{kass1995bayes} recommend using a threshold of 20 as the requirement of ``strong evidence'' in favor of the competing model. Sticky switching is important to the persistence of models. When $\alpha=1$, persistence is significantly harder to achieve since switching occurs too easily.\footnote{I analyze this case in Corollary \ref{cor:alpha equals 1}.}

\medskip{}

\noindent \textbf{Bayes factor rule.} The Bayes factor rule is a common model selection criterion in Bayesian statistics. It is a natural choice in this environment for the following reasons. First,  it is well known from the statistics literature that the Bayes factor automatically includes a penalty for including too much structure into the model and helps prevent overfitting \citep{kass1995bayes}. If model $\theta$ has a large parameter space and a diffuse prior, then each DGP in the model is assigned minimal weight, thinning out the likelihood $l_t(\theta)$.\footnote{The use of the prior in the calculation of the Bayes factor is crucial. Suppose, for example, we evaluate the likelihood of the data under model $\theta$ using the final posterior for all periods, i.e. $\hat l_t(\theta)\coloneqq \sum_{\omega\in\Omega^\theta} \pi_{t-1}^{\theta}(\omega)l_t(\theta,\omega)$, then the likelihood is less sensitive to the prior and the punishment on complexity vanishes over time if the posterior converges. 
If we evaluate the likelihood using the maximum likelihood estimates, i.e. $\tilde l_t(\theta)\coloneqq \prod_{\tau=0}^{t-1} l_t(\theta,\tilde\omega_t)$, where $\tilde\omega_t$ maximizes $l_t(\theta,\omega)$ among all $\omega\in\Omega^\theta$, then the likelihood ratio is completely independent from the prior. In this case, the likelihood under a larger model is always weakly higher than the likelihood under a nested smaller model. } This is consistent with the empirical observation that people in general favor simple models. As shown in the next section, this feature is the main driving force behind the necessity of high prior tightness for global robustness. Second, the Bayes factor rule is flexible in that it could be easily formulated for any model and any DGPs, without imposing assumptions about the underlying parametric structure.\footnote{Common alternative rules in statistics such as the BIC and the AIC are shown to approximate the Bayes factor under certain assumptions about the parametric family and the prior. } Finally, the Bayesian factor has a strong Bayesian flavor, so the agent maintains some form of consistency in belief updating and model switching.



\section{\label{sec:Robustness}Global and Local Robustness}

This section defines and provides a full characterization of global and local robustness. I characterize which models can be globally or locally robust for at least one prior in Section \ref{subsec:global robustness}. In Section \ref{subsec:traps}, I demonstrate a new technical challenge that arises due to the interaction between within-model learning and model switching. Finally, I characterize global and local robustness at a fixed prior in Section \ref{subsection:fixed priors}.


\subsection{Definition}
Persistence is defined relative to a competing model and a given set of priors. But which competing models would arise and what priors they would be assigned may not be predictable. This motivates a robustness approach. In this section, I consider two robustness notions: global robustness requires persistence against both local perturbations and paradigm shifts, while local robustness requires persistence against local perturbations. Global robustness is more applicable if the agent has a deep understanding of the environment---coming up with novel competing models requires knowledge and imagination---or if she works in an innovation-friendly culture (such as the technology industry). On the contrary, local robustness is more applicable when the environment is complicated or the agent is conservative. 

Formally, if model $\theta$ is globally robust, then provided a proper prior, it persists no matter what alternative model it is compared against and what prior is assigned to the alternative model. Conversely, if $\theta$ is not globally robust at any full-support prior, one can find a competing model associated with some prior that replaces $\theta$ infinitely often almost surely.
\begin{defn}[Global robustness] Model $\theta\in\Theta$ is \textit{globally robust at prior $\pi_0^\theta$} if $\theta$ persists against every competing model $\theta'\in\Theta$ at priors $\pi_0^\theta$ and $\pi_0^{\theta'}$ for every  ${\pi}_{0}^{\theta'}\in\Delta \Omega^{\theta'}$.
\label{def:Robustness} 
\end{defn}

Next, to define local robustness, we first need to properly quantify the distance between two arbitrary models $\theta$ and $\theta'$. Since every model consists of a collection of data-generating processes, a natural approach is to measure the distance between the corresponding sets of DGPs using the Prokhorov metric $d_P$ and the Hausdorff metric $d_H$.\footnote{The Prokhorov metric measures the distance between any two probability distributions on the same metric space. For any two probability measures $\mu$ and $\mu'$ over $\mathcal Y$, the Prokhorov distance is given by $d_P(\mu,\mu')\coloneqq \inf\left\{ \epsilon>0|\mu\left(Y\right)\leq \mu'\left(B_{\epsilon}\left(Y\right)\right)+\epsilon\text{ and }\mu'\left(Y\right)\leq \mu\left(B_{\epsilon}\left(Y\right)\right)+\epsilon\text{ for all }Y\subseteq\mathcal{Y}\right\}.$ The results in this paper also hold for alternative metrics such as Kullback-Leibler divergence or total variation. }\textsuperscript{,}\footnote{The Hausdorff metric measures how far two subsets of the same metric space are from each other. For any two sets $X$ and $Z$, their Hausdorff distance is $d_H(X,Z) = \max\{\sup_{x\in Z} d(x,Z), \sup_{z\in Z} d(X,z)\}$.} To begin with, denote the DGP to which model $\theta$ and parameter $\omega$ correspond by $Q^{\theta,\omega}$. The distance between any two DGPs $Q^{\theta,\omega}$ and $Q^{\theta',\omega'}$ is the maximum Prokhorov distance between the outcome distributions across all actions, 
\begin{equation}d(Q^{\theta,\omega},Q^{\theta',\omega'}) \coloneqq \max_{a\in\mathcal{A}}d_P(Q^{\theta,\omega}_a,Q^{\theta',\omega'}_a).\end{equation}
The distance between model $\theta$ and $\theta'$ is then given by the Hausdorff metric, \begin{equation}d(\theta,\theta') \coloneqq  d_{H}\left(\{Q^{\theta,\omega}\}_{\omega\in\Omega^\theta},\{Q^{\theta',\omega'}\}_{\omega'\in\Omega^{\theta'}}\right).\end{equation}
This metric allows us to define an \emph{$\epsilon$-neighborhood} of $\theta$,  \begin{equation}\label{defn:neighbor-models-local}N_{\epsilon}\left(\theta\right)\coloneqq\left\{ \theta^{\prime}\in\Theta:d(\theta,\theta')<\epsilon\right\}.\end{equation}
In order for the models to have similar initial predictions, the notion of local robustness also restricts the distance between the associated priors. Note that prior $\pi_0^\theta$ and prior $\pi_0^{\theta'}$ are defined on potentially different parameter spaces, but each of them corresponds to a belief over the set of all DGPs. Let this belief mapping be $\Gamma^\theta:\Delta\Omega^\theta\rightarrow \Delta (\Delta\mathcal Y)^\mathcal{A}$. An $\epsilon$-neighborhood of prior $\pi_0^\theta$ within the set of possible priors for $\theta'$ is given by 
\begin{equation}
N_\epsilon(\pi^\theta_0; \theta,\theta')\coloneqq \left\{\pi^{\theta'}_0\in\Delta \Omega^{\theta'}:d_P\left(\Gamma^\theta(\pi^\theta),\Gamma^{\theta'}(\pi^{\theta'})\right)<\epsilon\right\}.\label{defn:neighbor-priors-local} 
\end{equation}
Now we are ready to state the definition of local robustness. 
\begin{defn}[Local robustness]\label{def:local robustness}
 Model $\theta\in\Theta$ is \textit{locally robust at prior $\pi_0^\theta$} if there exists $\epsilon>0$ such that model $\theta$ persists against every competing model $\theta'\in N_\epsilon(\theta)$ at priors $\pi_0^\theta$ and $\pi_0^{\theta'}$ for every ${\pi}_{0}^{\theta'}\in N_\epsilon(\pi_0^\theta;\theta,\theta')$.
\label{def:Local Robustness informal} 
\end{defn}
Local robustness requires that there exists some positive $\epsilon$ such that a model persists against nearby models at nearby priors within the relevant $\epsilon$-neighborhoods. Hence, if model $\theta$ is locally robust, it persists as long as the switcher only considers sufficiently close perturbations. By contrast, if $\theta$ is not locally robust, then there is no chance that $\theta$ would be adopted forever even if the agent only considers small changes.

\subsection{\label{subsec:global robustness}Robustness at Some Prior}

I first characterize models that are locally or globally robust for at least one full-support prior. This characterization is useful because it directly speaks to the question of which models \textit{can} be robust---failing to be robust at some full-support prior implies non-robustness under every initial condition. Since all priors are assumed to be full-support, I sometimes drop this adjective for convenience. 

It is instructive to start our analysis with a special case: which models can persist against a correctly specified model? It is a well known fact that under a correctly specified model, a learner assigns probability 1 to the true outcome distribution in the limit \citep{easley1988controlling}. It follows that such a model perfectly matches the data in the long term, and thus any model that persists in its presence must also have perfect prediction accuracy in the limit. Since outcomes are endogenously generated by actions, this suggests that the agent---potentially after a lot of back-and-forth switching and belief updating---ends up playing a \emph{self-confirming equilibrium}. 


\begin{defn}\label{def:self-confirming eqm}
A strategy $\sigma\in\Delta\mathcal A$ is a \emph{self-confirming equilibrium} (SCE) under model $\theta$ if there exists a supporting belief $\pi\in\Delta\Omega^\theta$ such that the following conditions hold.
\begin{enumerate}
    \item[(1)] {Optimality}:  $\sigma$ is myopically optimal against  $\pi$, i.e. $\sigma \in\Delta A^\theta_m(\pi)$.
    \item[(2)] {Consistency}:  $\pi$ is consistent at  $\sigma$ in that $q^\theta(\cdot|a,\omega)\equiv q^\ast(\cdot|a)$ for all $a\in\supp(\sigma)$ and all $\omega\in\supp(\pi)$.
\end{enumerate}
\end{defn}

In such an equilibrium, the agent plays actions that are myopically optimal given a consistent supporting belief such that the model prediction fully coincides with the objective outcome distribution. Note that a SCE may well be inefficient---while consistency implies correct predictions about the payoff obtained in equilibrium, the model could have completely wrong predictions about the payoffs off path.  

But persistence against a correct model implies more than the existence of a SCE---
the agent should, with positive probability, end up playing \textit{only} the equilibrium actions under model $\theta$. If non-SCE actions are played infinitely often, the Bayes factor would still diverge to infinity and result in an abandonment of model $\theta$. Note that on paths where $\theta$ is adopted forever, a switcher eventually behaves no differently than a $\theta$-modeler. Thus, a necessary condition is that a $\theta$-modeler eventually only plays the equilibrium actions with positive probability. I refer to this stability notion as p-absorbingness, where ``p'' means that the equilibrium is absorbing \textit{with positive probability}. Since a $\theta$-modeler's problem is independent from the model switching process, we can further characterize p-absorbingness with conditions based on the primitives of the model.

\begin{defn}
\label{def:p-absorbing}Strategy $\sigma\in\Delta\mathcal{A}$ is \textit{p-absorbing} under $\theta$ if there exists a full-support prior $\pi_{0}^{\theta}$ and some positive integer $T$ such that, with positive probability, a $\theta$-modeler only plays actions in $\supp \left(\sigma\right)$ after period $T$.
\end{defn}

P-absorbingness differs from existing stability notions of self-confirming equilibria in the literature. In particular, it does not imply that the $\theta$-modeler's or the switcher's action sequence converges to a single action in the support of $\sigma$ or her action frequency converges to $\sigma$.\footnote{For example, p-absorbingness is weaker than the stability notion proposed by \citet{fudenberg2021limit}. By their definition, a pure equilibrium $a^{\ast}$ under $\theta$ is stable if for every $\kappa\in\left(0,1\right)$, there exists a belief $\pi\in\Delta\Omega^{\theta}$ such that for any prior $\pi_{0}^{\theta}$ sufficiently close to $\pi$, the dogmatic modeler's action sequence $a_{t}$ converges to $a^{\ast}$ with probability larger than $\kappa$. They do not define a stability notion for a mixed equilibrium.} Rather, it allows for non-convergent behavior within the support of $\sigma$, but rules out the scenario where the modeler almost surely plays actions outside the support of $\sigma$ infinitely often.\footnote{Indeed, a dogmatic modeler's action may never converge even when she eventually only plays the actions in the support of a p-absorbing SCE (see Example \ref{exa:p-absorbing-mixed-SCE} in  Appendix \ref{sec:Online-Appendix}).} Although p-absorbingness is a relatively weak requirement, not all self-confirming equilibria are p-absorbing.\footnote{If there exists an action that is outside the support of the SCE but optimal against the equilibrium belief, then playing that action can drive the agent away from the SCE (see Example \ref{exa:not-p-absorbing-SCE} in  Appendix \ref{sec:Online-Appendix}).} Note that while p-absorbingness implies the existence of some full-support prior that leads to the SCE being played by a $\theta$-modeler, in principle it is not necessary that the switcher also starts with such a prior.\footnote{This is because a switcher may go through a number of switches before she eventually settles down with a model, and those switches may happen to push her belief into a trajectory to the SCE even when her prior is not on the trajectory.} 
I conclude our analysis of a correctly specified competing model with Lemma \ref{lem:Correctly-full-characterization}.
\begin{lem} 
\label{lem:Correctly-full-characterization}If model $\theta$ persists against a correctly specified model $\theta'$ at some priors $\pi_0^\theta$ and $\pi_0^{\theta'}$, then there exists a p-absorbing SCE under $\theta$. 
\end{lem}

On its face, this lemma provides only a necessary condition for global robustness. On one hand, the condition in Lemma \ref{lem:Correctly-full-characterization} appears too strong for local robustness because this notion only requires persistence against local perturbations and any local perturbation of any misspecified model is necessarily misspecified. On the other hand, it is unclear whether the existence of a p-absorbing SCE would be sufficient for global robustness, even if the agent can start from any arbitrary full-support prior. Critically, p-absorbingness only ensures that  a $\theta$-modeler eventually plays the SCE with positive probability, but it remains to be shown that our agent can also arrive at the equilibrium with positive probability when she has access to multiple models---the presence of the second model may interfere with the learning process under the first model. Perhaps surprisingly, as I show in Theorem \ref{thm:Global-Robustness}, the existence of a p-absorbing SCE is both necessary for local robustness and sufficient for global robustness, which equates the two robustness notions provided  flexibility in the prior.

\begin{thm} 
\label{thm:Global-Robustness} The following statements are equivalent:
\begin{itemize}
    \item[(i)] Model $\theta$ is locally robust for at least one full-support prior.
    \item[(ii)] Model $\theta$ is globally robust for at least one full-support prior.
    \item[(iii)] There exists a p-absorbing SCE under model $\theta$.
\end{itemize}
\end{thm}

First, Theorem \ref{thm:Global-Robustness} provides a microfoundation for the persistence of certain types of misspecified models. 
A model can persist against any arbitrary competing model as long as it leads to a p-absorbing SCE at which the model predictions are indistinguishable from the truth. It is worth noting that Theorem \ref{thm:Global-Robustness} does not depend on switching threshold $\alpha$ (as long as $\alpha>1$), meaning that the sets of models that can be locally and globally robust are the same for any level of switching stickiness. Thus, the switching threshold may only affect model robustness through changing the set of supporting priors. 

Second, the equivalence between (i) and (ii) offers a new perspective for understanding local and global robustness. If a model fails to be globally robust, the switcher need not go far to find an attractive alternative---models that do not persist against paradigm shifts are also vulnerable to local changes. The intuition lying behind this result is quite simple. Although the agent is restricted to contemplate only local perturbations, such perturbations are unconstrained and can be made towards the direction of higher asymptotic accuracy. Specifically, for any model $\theta$, we could construct a neighbor competing model $\theta'$ within its $\epsilon$-neighborhood such that $\theta'$ outperforms $\theta$ consistently unless $\theta$ induces a self-confirming equilibrium. To do this, we simply take the predictions of $\theta'$ to be a convex combination between the predictions of $\theta$ and the true DGP for every action in $\mathcal A$. Since the Kullback-Leibler (KL henceforth) divergence between any two probability distributions is convex, the mixture model $\theta'$ yields a strictly lower KL divergence than model $\theta$ unless the agent plays a self-confirming equilibrium under model $\theta$.\footnote{The Kullback-Leibler divergence of a density $q$ from another density $q^{\prime}$ is given by $D_{KL}\left(q\parallel q^{\prime}\right)=\int_{\mathcal{Y}}q\ln\left(q/q^{\prime}\right)\nu\left(dy\right)$. 
The KL divergence is an asymmetric non-negative distance measure between $q$ and $q^{\prime}$, which is minimized to zero if and only if $q$ and $q^{\prime}$ coincide almost everywhere. It is convex in the following sense: for any two pairs of densities $(q_1,q_2)$ and $(q_1',q_2')$ and any $\gamma\in [0,1]$, we have $D_{KL}\left(\lambda q_1+(1-\lambda)q_2\parallel \lambda q_1'+(1-\lambda)q_2'\right)\leq \lambda D_{KL}\left( q_1\parallel q_1'\right) + (1-\lambda) D_{KL}\left(q_2\parallel q_2'\right)$. }

Third, comparing Theorem \ref{thm:Global-Robustness} with Lemma \ref{lem:Correctly-full-characterization}, we learn that the demanding notion of global robustness amounts to the requirement that $\theta$ persists against one correctly specified model at some prior. Provided that $\theta$ can persist against a competing model that assigns a tiny probability to the true DGP, it also has the potential to persist against the true model, or any model with arbitrarily complex parameter space. Conversely, models that fail to be globally robust do not persist in the long term as long as the agent considers any correctly specified model. An immediate corollary is that any correctly specified model is globally robust since a model must persist against itself.\footnote{Notice that the Bayes factor between one model and itself is always 1.}

More importantly, misspecified models can be globally or locally robust as well. Corollary \ref{corr:Suff-robustness} provides a sufficient condition for the existence of a p-absorbing SCE, which can be easily verified by computing the set of all SCEs induced by the model. Say that a self-confirming equilibrium $\sigma$ is \emph{quasi-strict} if there exists a supporting belief $\pi$ such that any action outside the support of $\sigma$ is strictly suboptimal given $\pi$, i.e. $\supp{(\sigma)}=A_m^\theta(\pi)$, then it can be shown that any quasi-strict SCE is p-absorbing.\footnote{\citet{fudenberg2022misspecifications} show that quasi-strictness can play a different role in facilitating the persistence of a misspecified model. In an evolutionary framework, they find that a model that admits a quasi-strict Berk-Nash equilibrium can resist local mutations provided that they induce weakly worse-performing actions. This is because quasi-strictness ensures that local mutants do not play actions outside the support of the equilibrium and thus no new evidence is generated after mutation.}

\begin{cor}
\label{corr:Suff-robustness}Model $\theta$ is locally or globally robust for at least one prior if $\theta$ is correctly specified or there exists a quasi-strict SCE under $\theta$.
\end{cor}
Quasi-strictness ensures that at any belief sufficiently close to the equilibrium supporting belief, it is strictly optimal to play the actions prescribed by the equilibrium. Moreover, since the equilibrium is self-confirming, the $\theta$-modeler's belief stays within a small neighborhood of the supporting belief with positive probability. Taken together, this implies that starting from a prior sufficiently close to the equilibrium supporting belief, the $\theta$-modeler plays the SCE forever with positive probability, and thus the SCE is p-absorbing.
 

\subsection{\label{subsec:traps}Self-Defeating Models and Traps}
The characterization in Theorem \ref{thm:Global-Robustness} is clean and intuitive. From a high level, p-absorbingness ensures that the SCE is reachable from some full-support prior, and the self-confirming property ensures that model $\theta$ makes perfect predictions in the limit and thus persists against any competing model. However, this rough intuition leaves out new interesting challenges that arise due to the interaction between within-model learning and model switching. Such challenges are inherent to the multiple-model learning framework and thus may be of independent interest to future research pursuits on problems other than persistence and robustness.       
               
In general, for a dogmatic modeler, both his behavior and beliefs are endogenous and may mutually influence each other; for a switcher, in additional to behavior and beliefs within each model, the model choice is also endogenous and all three endogenous objects can influence one another. This interaction can cause problems that prevent a model from being locally or globally robust at a given prior even if the model admits a p-absorbing SCE. In particular, the outcome realizations that lead a dogmatic modeler to the SCE may in fact trigger a switch away from model $\theta$, rendering its adoption to be \emph{self-defeating}. I illustrate this issue in Example \ref{exa:endogenous-model-choice}. For simplicity, in Example \ref{exa:endogenous-model-choice} I take the competing model to be the true model, but a similar phenomenon can occur with a competing model arbitrarily close to the initial model.  

\begin{example}[Self-defeating Models]\label{exa:endogenous-model-choice}
In each period, an agent chooses from two tasks $a_t\in\{a^1,a^2\}$ and observes the outcome/payoff of the chosen task $y_t\in \{0,1\}$, where $0$ represents failure and $1$ represents success. The true DGP prescribes that successes and failures happen with equal probability $0.5$ for either task. Hence, the agent would be indifferent between the tasks if the true DGP was known. 

The agent holds a subjective model $\theta$ that presumes the success rate may depend on both the task type and his luck $\omega\in\Omega^\theta=\{\omega^1,\omega^2\}$, where $\omega^1$ represents good luck and $\omega_2$ represents bad luck (see Table \ref{tab:endogenous-model-choice}). Under model $\theta$, the agent believes Task 1 is risky and success occurs more often if he has good luck, while Task 2 is safe and its outcome is independent from his luck. Besides, the agent is overall \emph{pessimistic} under $\theta$ because the assumed success rate is always (weakly) lower than its true level. The agent believes that his luck is fixed and has a uniform prior over his luck, i.e. $\pi_0^\theta(\omega_1)=0.5$. His policy under $\theta$ prescribes Task 1 iff good luck is more likely than bad luck, i.e. $\pi_t^\theta(\omega_1)\geq 0.5$.\footnote{The uniform prior is assumed for simple exposition. The mechanism in this example does not depend on the fact that the agent starts off being exactly indifferent between the tasks.}  In addition, the agent entertains the competing model $\theta^\ast$ that correctly predicts the true success rate. Under model $\theta^\ast$, the agent is indifferent and always chooses Task 1 (see Table \ref{tab:endogenous-model-choice}). We consider the case where his switching threshold is given by $\alpha=1.1$. 

\begin{table}[]
    \centering
    \begin{tabular}{c|c c} \label{table:endogenous-model-choice}
         $q^\theta(1|a,\omega)$ & $\omega^1$ & $\omega^2$ \\ 
         \hline 
         $a^1$ & 0.5 & 0.3\\ 
         $a^2$ & 0.4 & 0.4 
    \end{tabular} \quad \quad
    \begin{tabular}{c|c } 
         $q^{\theta^\ast}(1|a,\omega)$ & $\omega^\ast$ \\
         \hline 
         $a^1$ & 0.5 \\
         $a^2$ & 0.5
    \end{tabular}
    \caption{Initial model $\theta$ and competing model $\theta'$ in Example \ref{exa:endogenous-model-choice}. }
    \label{tab:endogenous-model-choice}
\end{table}

Choosing Task 1 is a strict self-confirming equilibrium under $\theta$, supported by a degenerate belief at $\omega^1$. To see why, note that the risky task is strictly optimal when the agent believes his has good luck; meanwhile, the superstitious belief of good luck offsets the overall pessimism and thus the model correctly predicts the success rate. 

However, it turns out that model $\theta$ does not persist against model $\theta^\ast$ at the given uniform prior, because the outcome realization that leads to choosing the risky task $a^1$ also triggers a model switch. As illustrated in Figure \ref{fig:endogenous-model-choice}, if the first realized outcome is a failure, the agent believes that he is more likely to have bad luck and thus switches his task choice to the safe task $a^2$ (Scenario 1); if the first realized outcome is a success, the agent switches his model choice to the more optimistic model $\theta^\ast$ in the next period (Scenario 2). In Scenario 1, the safe task choice causes the agent to stop updating on his luck. As a result, the agent never switches back to the risky task $a^1$ as long as he remains under model $\theta$. Since $\theta$ is incorrectly pessimistic when $a^2$ is chosen, the agent eventually switches to the correct model $\theta^\ast$ and enters Scenario 2. Once Scenario 2 occurs, the agent switches back to the overall pessimistic model $\theta$ only under the circumstance that he observes more failures than successes. But if so, the resulted posterior $\pi_t^\theta$ assigns higher probability to bad luck than good luck, which again induces the agent to choose the safe task $a^2$, bringing the agent back to Scenario 1. Eventually, the agent must abandon model $\theta$ and adopt the competing model $\theta^\ast$ forever. Therefore, $\theta$ does not persist against $\theta^\ast$ under the given priors.

\begin{figure}
    \centering
    \tikzstyle{level 1}=[level distance=6.5cm, sibling distance=4cm]
    \tikzstyle{level 2}=[level distance=4cm, sibling distance=3.5cm]
    
    \tikzstyle{bag} = [text centered, draw=black]
    \tikzstyle{bag2} = [text centered]
    \tikzstyle{end} = [circle, minimum width=3pt,fill, inner sep=0pt]

    \begin{tikzpicture}[grow=right,sloped]
    \node[bag] {\begin{tabular}{c} belief $\pi_0^\theta=(\frac{1}{2},\frac{1}{2})$ \\ model $m_0=\theta$ \\ action $a_0=a^1$ \end{tabular} }
    child {
        node[bag] {\begin{tabular}{c} belief $\pi_1^\theta=(\frac{5}{8},\frac{3}{8})$ \\ Bayes factor $\lambda_1=\frac{5}{4}>\alpha $ \\ \textbf{model $m_1=\theta^\ast$} \\ action $a_1=a^1$ \end{tabular}}        
            child {
                node[bag2]
                    {Scenario 2}
                    edge from parent[->,white]  
            }
            edge from parent[->]  
            node[above] {success}
            node[below]  {$y_0=1$}
    }
    child {
        node[bag]{\begin{tabular}{c} belief $\pi_1^\theta=(\frac{5}{12},\frac{7}{12})$ \\ Bayes factor $\lambda_1=\frac{5}{6}<\alpha$ \\ model $m_1=\theta$ \\ \textbf{action $a_1=a^2$} \end{tabular}}  
        child {
                node[bag2]
                    {Scenario 1} 
                    edge from parent[->,white]  
            }
        edge from parent[->]   
            node[above] {failure}
            node[below]  {$y_0=0$}
    };
\end{tikzpicture}
    \caption{Scenario analysis in Example \ref{exa:endogenous-model-choice}.}
    \label{fig:endogenous-model-choice}
\end{figure}

\end{example}

What are the root causes of the self-defeating result in Example \ref{exa:endogenous-model-choice}? First, the agent's model choice and his belief on luck are tightly correlated given the particular structure of the models in Example \ref{exa:endogenous-model-choice}. In order for the agent to choose the risky task, he must believe in good luck more than bad luck, but the successes needed to induce this belief inevitably lead to a switch to the more optimistic competing model. Second, the agent's the task choice and model choice are too sensitive to the early outcome realizations. Since the agent's prior $\pi_0^\theta$ is relatively far away from the SCE supporting belief and the switching threshold $\alpha$ is relatively low, a single observation is powerful enough to sway the agent's task choice or the model choice. Last but not least, the safe choice constitutes an absorbing \emph{trap} in model $\theta$ because it causes the agent to stop updating his belief within $\theta$. Indeed, the agent can never choose the risky option under model $\theta$ again once the agent enters the trap.

The above diagnosis points out two different directions for breaking the self-defeating result. First, the agent may not fall into the trap if he starts with a prior sufficiently close to the SCE supporting belief. Specifically, if the agent starts out assigning more probability to good luck, then his task choice will be less sensitive to early failures; meanwhile, his initial model becomes overall less pessimistic and thus his model choice will be less sensitive to early successes. The proof of Theorem \ref{thm:Global-Robustness} uses precisely this idea to show that there exists a prior $\pi_0^\theta$ sufficiently close to the SCE supporting belief, so that the self-defeating behavior does not arise with probability 1.

The second direction is to remove any traps from the model while holding the prior fixed. If model $\theta$ predicts that the success rate still depends on the agent's luck even slightly for the safe task, then the agent does not stop updating on his luck upon choosing the safe task. If so, there is no longer a trap that locks in the agent's task choice. Consequently, model $\theta$ is not self-defeating---we can always construct a finite sequence of outcomes to make the agent eventually confident in model $\theta$ as well as attach high probability to having good luck (this is non-trivial and requires proof). To rule out  traps of the sort described in Example \ref{exa:endogenous-model-choice}, we can assume the model is \textit{identifiable} as defined below.
\begin{defn}
Model $\theta$ is \emph{identifiable} if the predictions of different parameters in $\theta$ are different for all actions, i.e. $q^\theta(\cdot|a,\omega)\not=q^\theta(\cdot|a,\omega')$ for all distinct $\omega,\omega'\in\Omega^\theta$ and all $a\in\mathcal A$.
\end{defn}
It turns out non-identifiability is not the only cause of traps. The other type of traps are more technical and arise when the p-absorbing SCE under model $\theta$ is not quasi-strict. In this case, there exists some action that is optimal given the equilibrium supporting belief but not self-confirming. Under certain policies $f^\theta$, these actions can also act like traps---once they are played, the agent can never go back to play the SCE actions under the same model.
Definition \ref{def:no-traps} collects the two no-trap conditions. In the next subsection, I characterize local and global robustness at a fixed prior under the no-trap conditions.

\begin{defn}\label{def:no-traps}
Model $\theta$ has \emph{no traps} if $\theta$ is {identifiable} and all p-absorbing SCEs (if exists) under $\theta$ are quasi-strict.
\end{defn}

\subsection{Robustness at a Fixed Prior\label{subsection:fixed priors}}

Theorem \ref{thm:Global-Robustness} characterizes \emph{which models} can be locally and globally robust based on their prediction accuracy at the induced equilibrium, but it remains silent about \emph{under which priors} these models are locally or globally robust. The analysis about self-defeating models in the preceding subsection suggests that there is no good answer to this problem when there are traps in the model, except that we know such priors do exist. When there are no traps, a natural conjecture is that a model that admits a p-absorbing SCE is locally and globally robust at all priors. In this subsection, I show that while this conjecture is correct for local robustness, global robustness requires the prior to be tight and concentrated around the p-absorbing SCEs.

Let us start by analyzing the simplest case where the outcomes are exogenously generated, or equivalently, $\mathcal A$ contains a single action $\bar a$. As illustrated by Example \ref{exa:single-action-prior}, in this special case, a model is locally robust at any prior if and only if its predictions contain the true outcome distribution; by contrast, a model is globally robust at a given prior if and only if the prior assigns probability weakly higher than $1/\alpha$ to the true outcome distribution.
\begin{example}[Exogenous data]\label{exa:single-action-prior}
Suppose the agent works on a single task $\mathcal A=\{\bar a\}$ and observes the failure/success of the task, $\mathcal Y = \{0,1\}.$ As in Example \ref{exa:endogenous-model-choice}, the true DGP prescribes that the success rate is 0.5. The agent holds a subjective model $\theta$, under which he presumes the success rate depends on his luck $\omega\in\Omega^\theta=\{\omega^1,\omega^2\}$, where $\omega^1$ represents good luck and $\omega^2$ represents bad luck (see Table \ref{table:exogenous-data}). Note that model $\theta$ is correctly specified because it predicts the true success rate of $0.5$ when the agent has good luck. Corollary \ref{corr:Suff-robustness} immediately implies that model $\theta$ is both locally and globally robust for at least one prior. 

\begin{table}[]
    \centering
    \begin{tabular}{c|c c} \label{table:exogenous-data}
         $q^\theta(1|a,\omega)$ & $\omega^1$ & $\omega^2$  \\ 
         \hline 
         $\bar a$ & 0.5  & 0.3
           
    \end{tabular} \quad \quad 
    \begin{tabular}{c|c c} \label{table:exogenous-data}
         $q^{\theta'}(1|a,\omega)$ & $\omega^1$ & $\omega^2$  \\ 
         \hline 
         $\bar a$ & 0.5  & 0.3$+\epsilon$
           
    \end{tabular} \quad \quad 
    \begin{tabular}{c|c c} 
         $q^{\theta^\ast}(1|a,\omega)$ & $\omega^\ast$ \\
         \hline 
         $\bar a$ & 0.5 
    \end{tabular}
    \caption{Initial model $\theta$ and competing model $\theta'$ in Example \ref{exa:single-action-prior}. }
    
\end{table}

Model $\theta$ is locally robust at all full-support priors. To see why, let us assume the agent entertains a nearby competing model $\theta'$ that predicts a success rate of $0.3+\epsilon$ under bad luck. When $\epsilon>0$ is small, model $\theta'$ is slightly more optimistic than model $\theta$. Since there is a single action, it is straightforward to calculate the Bayes factor, 
$$\frac{\ell_t(\theta')}{\ell_t(\theta)}= \frac{\pi_0^{\theta'}(\omega^1) 0.5^t + \pi_0^{\theta'}(\omega^2) (0.3+\epsilon)^{S_t}(0.7-\epsilon)^{F_t}}{\pi_0^\theta(\omega^1) 0.5^t + \pi_0^\theta(\omega^2) 0.3^{S_t}0.7^{F_t}},$$
where $S_t$ and $F_t$ are the number of successes and failures observed before period $t$ and $S_t+F_t=t$. As the agent accumulates more evidence, the likelihood of bad luck eventually vanishes as compared that of good luck in both models. In the limit, the Bayes factor is completely determined by the ratio of the prior odds of good luck, $\pi_0^{\theta'}(\omega^1)/\pi_0^{\theta}(\omega^1)$, which is close to $1$ when the priors are sufficiently close and thus is bounded above by $\alpha>1$. Therefore, model $\theta$ persists against model $\theta'$. This argument can be generalized to other nearby competing models.

By contrast, model $\theta$ is globally robust at a given prior if and only if the prior assigns probability weakly higher than $1/\alpha$ to good luck, i.e. $\pi_0^\theta(\omega^1)\geq 1/\alpha$. To illustrate, suppose the agent entertains the true DGP $\theta^\ast$ as the competing model. The Bayes factor is given by $$\lambda_t = \frac{\ell_t(\theta^\ast)}{\ell_t(\theta)}= \frac{0.5^t}{\pi_0^\theta(\omega^1) 0.5^t + \pi_0^\theta(\omega^2) 0.3^{S_t}0.7^{F_t}}.$$ When $\pi_0^\theta(\omega^1)\geq 1/\alpha$, the Bayes factor $\lambda_t$ is bounded above by $\alpha$ for any history, so the agent never switches to the competing model. Intuitively, this is because the explanatory power of model $\theta$ is at least $1/\alpha$-times as large as the true model but a switch is only triggered when the Bayes factor is strictly larger than $\alpha$. On the other hand, if $\pi_0^\theta(\omega^1)<1/\alpha$, the agent must abandon model $\theta$ at some point. Since the likelihood of bad luck eventually vanishes, the Bayes factor $\lambda_t$ converges to $1/\pi_0^{\theta}(\omega^1)$, which is strictly larger than $\alpha$.
\end{example}
The key driving force behind  Example \ref{exa:single-action-prior} is that the Bayes factor rule acts like the \emph{Occam's Razor}---it favors parsimonious models with tight priors and punishes complex models with diffuse priors. When the agent compares a model with high asymptotic accuracy and its local perturbation, their priors are similarly tight around similar DGPs. Sticky switching ($\alpha>1$) then implies that the agent remains under the same model with positive probability. However, model $\theta$ with a diffuse prior can fit significantly worse than a model with a prior concentrated around the true DGP. In this case, a forever switch to the competing model is inevitable. Notably, as switching becomes stickier ($\alpha$ grows larger), the agent's tolerance for diffuse priors also increases. 

When outcomes are endogenously generated by actions, an analogous condition of prior tightness is that the prior assigns probability weakly higher than $1/\alpha$ to the parameters that are in the support of a SCE supporting belief. But the validity of this generalization is \textit{a priori} unclear because of two complications. First, contrasting the case with exogenous data, here the parameter(s) in the support of a SCE supporting belief may not predict the true outcome distribution when non-SCE actions are being played. So we cannot directly bound the Bayes factor from above using the prior ratio as in Example \ref{exa:single-action-prior} for all action histories. Second, when there are multiple p-absorbing SCEs, the supports of different SCE supporting beliefs may not overlap. Hence, it is unclear whether the prior should be concentrated around one of the supporting beliefs or the whole set of the supporting beliefs.  

Despite the complications, I show in Theorem \ref{thm:concentrated-prior} that when model $\theta$ has no traps, prior tightness is sufficient and necessary for global robustness at a given prior. To state the result, I denote by $C^\theta$ the set of parameters in $\theta$ that support at least one p-absorbing SCE and refer to $C^\theta$ as the set of \emph{consistent} parameters in $\theta$. Formally, for each $\omega\in C^\theta$, there exists a p-absorbing SCE under $\theta$ with supporting belief $\delta_\omega$.\footnote{Note that when $\theta$ is identifiable, no parameters predict the same outcome distribution, and thus the supporting belief of any SCE must be pure.}  Notice that model $\theta$ admits a p-absorbing SCE if and only if $ C^\theta$ is not empty.

\begin{thm}
\label{thm:concentrated-prior} Suppose model $\theta$ has no traps, then the following are true: 
\begin{itemize}
    \item[(i)] Model $\theta$ is locally robust at all priors if and only if $ C^\theta\not=\emptyset$.
    \item[(ii)] Model $\theta$ is globally robust at prior $\pi_0^\theta$ if and only if $ C^\theta\not=\emptyset$ and $\pi_0^\theta(C^\theta)\geq 1/\alpha$. 
    \item[(iii)] Model $\theta$ is globally robust at all priors if and only if $ C^\theta=\Omega^\theta$. 
\end{itemize}

\end{thm}

While Theorem \ref{thm:Global-Robustness} unifies local and global robustness in terms of {which} models can be robust, Theorem \ref{thm:concentrated-prior} clarifies the distinction between these notions in terms of {when} these models are robust---local robustness is prior-free, but global robustness is prior-sensitive. Theorem \ref{thm:concentrated-prior} provides an exact quantification of how concentrated the prior must be on $C^\theta$ in order to support global robustness. In particular, the tightness of the prior (measured by the probability assigned to $C^\theta$) multiplied by the switching stickiness (measured by the threshold $\alpha$) must be weakly larger than 1. Therefore, if model $\theta$ admits a p-absorbing SCE, then for any fixed full-support prior, $\theta$ is globally robust as long as switching is sufficiently sticky ($\alpha$ sufficiently high). Alternatively, if we fix $\alpha$, global robustness holds at all priors if and only if $C^\theta=\Omega^\theta$. Conversely, when  $C^\theta\not=\Omega^\theta$, global robustness fails at any given full-support prior if switching is sufficiently easy ($\alpha$ sufficiently close to $1$). 

Taken together, Theorems \ref{thm:Global-Robustness} and \ref{thm:concentrated-prior} draw an interesting comparison between misspecified models and correctly specified models in terms of their robustness properties. On one hand, all correctly specified models are locally robust at all priors and globally robust at some prior since they have high asymptotic accuracy, which is achieved only by a subset of misspecified models. On the other hand, some misspecified models can be globally robust at more priors if they have a simple structure or induce a large set of SCEs.\footnote{Under identifiability, misspecified models can easily satisfy $C^\theta=\Omega^\theta$ if there exist multiple p-absorbing SCEs, but the only correctly specified model that satisfies this condition is the true model $\theta^\ast$. If $\theta$ is correctly specified, then there exists $\omega^\ast\in\Omega^\theta$ such that $q^\theta(\cdot|a,\omega^\ast)=q^\ast(\cdot|a)$ for all $a\in\mathcal A$. Identifiability then implies that for any $\omega\in\Omega^\theta$ and $\omega\not=\omega^\ast$, we must have $q^\theta(\cdot|a,\omega)\not=q^\ast(\cdot|a)$ for all $a\in\mathcal A$. Therefore, at most one parameter, namely $\omega^\ast$, can support a SCE, and thus $\Omega^\theta=C^\theta$ implies $\Omega^\theta=\{\omega^\ast\}$.}  I illustrate this interesting wedge in Application \ref{app:media bias}.

For completeness, I also consider the case where switching is perfectly non-sticky and $\alpha$ is exactly $1$. In this extreme case, the existence of a p-absorbing SCE is no longer sufficient for either local robustness or global robustness for at least one full-support prior. Corollary \ref{cor:alpha equals 1} provides an alternative characterization in place of Theorems \ref{thm:Global-Robustness} and \ref{thm:concentrated-prior}.

\begin{cor}\label{cor:alpha equals 1}
Suppose model $\theta$ has no traps and $\alpha=1$, then model $\theta$ is locally or globally robust if and only if $C^\theta=\Omega^\theta$.

\end{cor}
Corollary \ref{cor:alpha equals 1} has three implications. First, the set of models that can be locally robust or globally robust for at least one prior stays unchanged when the threshold $\alpha$ is strictly larger than $1$ but shrinks discontinuously at $\alpha=1$. Second, the gap between local and global robustness has closed when switching is perfectly non-sticky, since both require the prior to be fully concentrated around p-absorbing SCEs. Finally, this result, combined with Theorem \ref{thm:concentrated-prior}, uncovers the equivalence between two strong notions of robustness---global robustness when switching is non-sticky and global robustness at all priors---both of them characterized by a simple condition, $C^\theta=\Omega^\theta$.

\section{Constrained Local Robustness\label{sec:constrained local robustness}}
This section relaxes the requirement of local robustness by restricting the agent to nearby competing models within a constrained family. The motivation for the constraint comes from the observation that decision makers often do not abandon their basic framework of assumptions and principals when coming up with nearby competing theories.\footnote{For example, Kuhn argues that science experiences alternating phases of normal science and revolutions. While revolutions entail paradigm shifts, he argues that ``normal science, ...,  often suppresses fundamental novelties because they are necessarily subversive of its basic commitments" (Kuhn, 1996, p. 5).} Rather, they maintain fundamental assumptions and consider structured and directed changes to their original model. With this constraint, I show that a model can be locally robust even if it does not give rise to a p-absorbing self-confirming equilibrium. Rather, the model only needs to induce a p-absorbing Berk-Nash equilibrium that satisfies a local dominance property. As the family of models grows larger, however, the local dominance property becomes more demanding and eventually morphs into the self-confirming property.  

\subsection{Definition}
I define a parametric family of models with a \emph{meta-model}. A meta-model $\bar\theta$ consists of a profile of data-generating processes, $\left\{ q\left(\cdot|a,\omega\right)\right\}_{a\in\mathcal{A},\omega\in\Omega^{\bar\theta}}$, where $q$ is uniformly continuous over the \emph{meta-parameter set} $\Omega^{\bar\theta}$ for all $a\in\mathcal{A}$. The set $\Omega^{\bar\theta}$ can be any subset of a Euclidean space without the restriction of being finite or bounded. The meta-model $\bar\theta$ then generates an affiliated family of models $\Theta^{\bar\theta}$. 
\begin{defn}
The \textit{$\bar\theta$-family} of models is the set $\Theta^{\bar\theta}\subseteq\Theta$ such that 
\begin{equation}
\Theta^{\bar\theta}\coloneqq\{\theta\in\Theta:q^{\theta}\left(\cdot|a,\omega\right)\equiv q\left(\cdot|a,\omega\right)\text{ for all }\omega\in\Omega^{\theta}\subseteq\Omega^{\bar\theta}\text{ and all }a\in\mathcal{A}\}.
\end{equation}
\end{defn}
\noindent Two models $\theta$ and $\theta^{\prime}$ belong to the same family generated by $\bar\theta$ if they share the same mapping from parameters to DGPs as $\bar\theta$ but differ in their parameter sets.  We can conveniently measure the distance between any $\theta$ and $\theta'$ by the Hausdorff distance between their parameter sets, $\Omega^{\theta}$ and $\Omega^{\theta^{\prime}}$. I define an $\epsilon$-neighborhood of $\theta$ as follows, \begin{align}
N_{\epsilon}^{\bar\theta}\left(\theta\right) & \coloneqq\left\{ \theta^{\prime}\in\Theta^{\bar\theta}:d_{H}\left(\Omega^{\theta},\Omega^{\theta^{\prime}}\right)<\epsilon\right\}.
\end{align}
Note that this definition of nearby models is conceptually different from our previous definition based on non-parametric distance measures over sets of DGPs (see \eqref{defn:neighbor-models-local}). To capture the idea that the agent contemplates models with similar underlying structures, a proper metric over models should measure not only the distance in model predictions but also the distance between the parameter values. Two parameters $\omega$ and $\omega'$ in the meta-parameter set $\Omega^{\bar\theta}$ can be distant from each other but still correspond to extremely similar or even identical predictions. However, if $\omega$ and $\omega'$ are close in $\Omega^{\bar\theta}$, it follows from the continuity of $q$ that they must predict similar data-generating processes. The distance measure over priors is also significantly simpler than that in \eqref{defn:neighbor-priors-local}. The $\epsilon$-neighborhood of any $\pi\in\Delta\Omega^{\bar\theta}$ is given by
\begin{equation}
N^{\bar\theta}_\epsilon(\pi)\coloneqq \{\pi^\prime\in\Delta\Omega^{\bar\theta}:d_P (\pi,\pi^\prime)<\epsilon.\}
\end{equation}
A model is locally robust within the $\bar\theta$-family if it persists against every nearby model within that family under nearby priors. Its formal definition mirrors Definition \ref{def:local robustness}.
\begin{defn}[Constrained local robustness]
\label{def:local-robustness} Model $\theta\in\Theta^{\bar\theta}$ is \emph{$\bar\theta$-constrained locally robust} at prior $\pi_0^\theta$ if there exists $\epsilon>0$ such that $\theta$ persists against every competing model $\theta'\in N_{\epsilon}^{\bar\theta}\left(\theta\right)$ at priors $\pi_0^\theta$ and $\pi_0^{\theta'}$ for every ${\pi}_{0}^{\theta'}\in N^{\bar\theta}_\epsilon(\pi^\theta_0)$.
\end{defn}
    
It is worth noting that constrained local robustness is a flexible concept. One might think that since all models within the same family differ only in their parameter sets, the agent is limited to expanding, downsizing, or replacing the elements of the parameter set of her initial model, but this is not entirely true. By specifying different meta-models, we allow the agent to access vastly different sets of competing models. In addition, we can modify the parameter set of model $\theta$ to include additional dimensions without altering the predictions of model $\theta$. The addition of these new dimensions allows the agent to re-assess the validity of certain implicit assumptions and perspectives built into her initial model when considering competing models. Example \ref{exa:underestimation-risks} illustrates the flexibility of the approach. 
\begin{example}[Portfolio choice]

\label{exa:underestimation-risks} Consider an investor who chooses a portfolio from $N$ stocks. Her action $a_t=(a_t^1,...,a_t^N)\in\mathcal A=\{0,1\}^N$ specifies whether to invest in each of the $N$ stocks. At the end of each period, she obtains a CARA utility of $u\left(a_t,y_{t}\right)=1-e^{-\sum_{n=1}^N a_t^n y^n_{t}}$, where $y^n_t$ is the return of stock $n$. The true return of stock $n$ is $y_{t}^n=r^n+\xi^n_{t}$, where $r^n$ is the average return of stock $n$ and $\xi^n_t$ is the stock-specific random noise. Moreover, $\xi_t=\left(\xi_t^1,...,\xi^N_{t}\right)$ follows a multivariate normal distribution with zero mean and a covariance matrix $\Sigma=(v^{ij})$. 

The investor's initial model $\theta$ is very simple. She presumes that every stock has identically distributed and uncorrelated returns $y_t^n=r+\xi^n_{t}$, where $r$ is the average return of all stocks and $\xi_{t}\sim N\left(\boldsymbol{0},v\cdot \boldsymbol{I_n}\right)$ for some constant $v>0$. She conducts Bayesian updating to learn about $r$ and $v$ and has a finite parameter space $\Omega^\theta=\Omega^\theta_r\times \Omega^\theta_v\subseteq\mathbb R^2_+$.

Is the investor's model persistent? The answer to this question may depend on which simplifying assumption(s) she is willing to drop when considering competing models. For example, the investor may maintain her assumption of i.i.d. returns and investigate if the market is more volatile than what she initially assumes. In this case, we can specify the meta-model $\bar\theta$ to have an expanded parameter space along the dimension of the variance $v$, i.e. $\Omega^{\bar\theta}=\Omega^\theta_{r}\times \mathbb R_+$. Alternatively, the investor may question the validity of her assumption of identically distributed returns while maintaining the assumption that returns are uncorrelated. We can accommodate this consideration by rewriting the parameter space of $\theta$ as $ \Omega^\theta=\{((r^n),(v^{ij})):r^i=r^j,v^{ii}=v^{jj},v^{ij}=0\text{ and }(r^i,v^{ii})\in\Omega^\theta\text{ for all }i\not=j\}$. Note that this relabeling does not alter the predictions of model $\theta$, but allows us to embed $\theta$ into a larger parametric family with the following meta-parameter space, $$\Omega^{\bar\theta}=\underbrace{\mathbb R_+\times ...\times \mathbb R_+}_{r^1, r^2, ..., r^N}\times \underbrace{\mathbb R_+\times ...\times \mathbb R_+}_{v_{11},v_{22}, ..., v_{NN}}.$$
Of course, the investor may also be open to all possibilities within the normal distribution framework, in which case he further drops the no-correlation assumption and consider a meta-parameter space $\Omega^{\bar\theta}=\mathbb R_+^{2N}\times \mathbb R^{N(N-1)}$. 
\end{example}

\subsection{\label{subsec:Local-robustness-sufficient}Characterization}

In this subsection, I provide necessary conditions and sufficient conditions for constrained local robustness. New challenges emerge as a result of the constraint we impose over the competing models. First, models sufficiently close to a misspecified model must also be misspecified, which prevents us from using Lemma \ref{lem:Correctly-full-characterization}. Second, contrasting the case of local robustness, it is now infeasible to perturb model predictions unanimously towards the true DGP, since the perturbed model may not be in the constrained family. Therefore, there is no reason to believe that constrained locally robust models must induce a self-confirming equilibrium.

\paragraph{Necessary conditions.}We can still take inspiration from the characterization of local and global robustness. For model $\theta$ to be constrained locally robust, then at least in the limit, the agent should not find a nearby model within the same family consistently better fitting. That is, we can break down the characterization into two parts---the identification of a proper equilibrium concept and a good measure of prediction accuracy.

Characterizing the asymptotic behavior of a dogmatic modeler is one of the central questions in the misspecified learning literature. A major finding of the literature is that whenever the modeler's behavior stabilizes, the limit behavior must constitute a Berk-Nash equilibrium that I define below \citep{esponda2016berk,esponda2021asymptotic}. 
\begin{defn}
\label{def:Berk-Nash-Equilibria}Strategy $\sigma\in\Delta\mathcal{A}$ is a \textit{Berk-Nash equilibrium} (BN-E) under $\theta$ if there exists a supporting belief $\pi\in\Delta\Omega^{\theta}\left(\sigma\right)$ such that the following conditions hold.
\begin{itemize}
    \item[(i)] Optimality: $\sigma$ is myopically optimal against $\pi$, i.e. $\sigma\in\Delta A_m^\theta(\pi)$.
    \item[(ii)] KL-minimization: every $\omega$ in $\supp{\pi}$ is a minimizer of the $\sigma$-weighted KL divergence, 
    $$\omega\in\argmin_{\omega^{\prime}\in\Omega^{\theta}}\sum_{\mathcal{A}}\sigma\left(a\right)D_{KL}\left(q^{\ast}\left(\cdot|a\right)\parallel q^{\theta}\left(\cdot|a,\omega^{\prime}\right)\right).$$
\end{itemize}
\end{defn}
For convenience, denote by $\Omega^{\theta}\left(\sigma\right)$ the set of all KL-minimizers at $\sigma$ under $\theta$. Intuitively, these parameters yield the closest match to the true DGP among the parameters of model $\theta$ when the agent plays $\sigma$. Besides myopic optimality, a Berk-Nash equilibrium requires that the supporting belief $\pi$ takes support on the KL-minimizers.  Every model admits at least one Berk-Nash equilibrium \citep{esponda2016berk}. Note that a SCE is a special form of a BN-E---any consistent parameter must also minimize the KL divergence.

Next, I turn to the choice of the prediction accuracy order. If model $\theta$ is constrained locally robust, then the predictions of $\theta$ should be more accurate than any sufficiently close competing model in the same family after the agent settles down with a Berk-Nash equilibrium. A natural choice is to compare the KL divergence of the predicted outcome distributions. 

\begin{defn}
\label{def:locally-KL-minimizing}Model $\theta$ is \textit{locally
KL-minimizing} at strategy $\sigma$ within the $\bar\theta$-family if there exists 
$\epsilon>0$ such that $\sum_{\mathcal A} \sigma(a)\mathbb E \left(\ln\frac{q(\cdot|a,\omega^\prime)}{q(\cdot|a,\omega)}\right)\leq 0$ for all $\omega\in\Omega^{\theta}\left(\sigma\right)$, $\omega^{\prime}\in\Omega^{\bar\theta}\cap B_{\epsilon}\left(\Omega^{\theta}\left(\sigma\right)\right)$.\footnote{Note that this is equivalent to the requirement that the KL divergence associated with $\omega$ is lower than that of $\omega'$, i.e. $D_{KL}\left(q^\ast(\cdot|a)\parallel q(\cdot|a,\omega)\right)\leq D_{KL}\left(q^\ast(\cdot|a)\parallel q(\cdot|a,\omega')\right)$.}
\end{defn}
In words, this condition says that any KL-minimizer in the parameter set of model $\theta$ is also a local KL-minimizer within the meta-parameter space $\Omega^{\bar\theta}$. Theorem \ref{thm:Local-lobustness-w-convergence} shows that if the action of a $\theta$-modeler converges, constrained local robustness indeed requires the existence of a p-absorbing BN-E and local KL-minimization. 
\begin{thm}
\label{thm:Local-lobustness-w-convergence}Suppose that the action of a $\theta$-modeler almost surely converges under all full-support priors. If model $\theta\in\Theta^{\bar\theta}$ is $\bar\theta$-constrained locally robust, then it must admit a pure p-absorbing BN-E $\sigma$ at which $\theta$ is locally KL-minimizing within the $\bar\theta$-family.
\end{thm}

There is one caveat: Theorem \ref{thm:Local-lobustness-w-convergence} is established with the presumption that the action sequence of a $\theta$-modeler converges over time, but this may not the case in general. When the agent holds a misspecified model that persists against a correctly specified model, her action must converge to the support of a self-confirming equilibrium (Lemma \ref{lem:Correctly-full-characterization}). But this convergence property is lost once we shift our focus to misspecified models. When the behavior of a $\theta$-modeler does not converge to (the support of) any Berk-Nash equilibrium, model $\theta$ may persist even without being locally KL-minimizing at any BN-E. I show that Theorem \ref{thm:Local-lobustness-w-convergence} still holds without the convergence assumption as long as $\mathcal{A}$ is binary (see Theorem \ref{thm:Local-dim2-necessary} in Appendix \ref{sec:Proofs-Main-Results}).  

The majority of the misspecified models in the literature features convergence of behavior. In most cases, specific assumptions over the types of misspecification and the outcome distributions are imposed to ensure convergence of actions or action frequencies (\citet{nyarko1991learning,heidhues2018unrealistic,he2022mislearning,bagindin2021overconfidence}). In the Appendix, I extend Theorem \ref{thm:Local-lobustness-w-convergence} to accommodate situations where the agent's action sequence does not converge but her action frequency does.\footnote{Formally, given a finite action space $\mathcal{A}$ and an action sequence $\left(a_{1},a_{2},...\right)$, the action frequency sequence is given by $\left(\sigma_{t}\right)_{t}$, where $\sigma_{t}\left(a\right)=\frac{1}{t}\sum_{\tau=1}^{t}1_{\left\{ a_{t}=a\right\} }$. \citet{esponda2021asymptotic} establish global almost-sure convergence of a dogmatic modeler's action frequency to a  BN-E if it is ``globally attracting'', where global attractiveness is defined based on a differential equation that describes the evolution of the action frequency. }  In those environments, Theorem \ref{thm:Local-lobustness-w-convergence} provides a simple criterion to determine if a given model is constrained locally robust.

\paragraph{Sufficient conditions.}
It remains to be examined whether the condition in Theorem \ref{thm:Local-lobustness-w-convergence} is sufficient for constrained local robustness. It turns out that higher accuracy in terms of lower KL divergence is a little weak for this purpose. In particular, the predictions of model $\theta$ can be shown to be more accurate than those of neighbor models in the limit but not {consistently} so starting from any given period, but the latter is critical in ensuring the agent does not switch to the competing model before the Berk-Nash equilibrium is reached.\footnote{The reason is more technical than conceptual: as pointed out by \cite{frick2021belief}, when a distribution $q$ yields lower KL divergence than $q'$, their log-likelihood ratio $\ln (q'/q)$ constitutes a supermartingale but this supermartingale can be unbounded below, preventing us from invoking the relevant maximal inequalities. } \cite{frick2021belief} develops a slightly stronger order of prediction accuracy that partially restores the supermartingale argument. I build on their  order and define a local dominance property to compare the prediction accuracy across models.

\begin{defn}\label{def:local-dominance}
Model $\theta$ is \textit{locally dominant} at strategy $\sigma$ within the $\bar\theta$-family if there exists $\epsilon,d>0$ such that $\mathbb E \left(\frac{q(\cdot|a,\omega^\prime)}{q(\cdot|a,\omega)}\right)^d\leq 1$ for all $\omega\in\Omega^\theta(\sigma)$, $\omega^\prime\in\Omega^{\bar\theta}\cap B_\epsilon(\Omega^\theta(\sigma))$, and $a\in\supp (\sigma$). 
\end{defn} 

In words, this condition requires the existence of some positive $d$ such that for the $d$-th power of the likelihood ratio between any neighbor parameter $\omega'$ in the meta-parameter space and any KL-minimizer $\omega$ has an expectation weakly lower than 1, whenever the agent plays the BN-E actions. Local dominance comes for free when $\sigma$ is a self-confirming equilibrium, because the expected likelihood ratio between any $\omega'$ and any consistent $\omega$ is always 1. 
Local dominance strengthened local KL-minimization in two aspects. First, local KL-minimization only compares nearby parameters at the mixed action $\sigma$ while local dominance makes a comparison at each action in the support of $\sigma$. Second, fixing an action $a$, as \cite{frick2021belief} point out, if $\mathbb E \left(\frac{q(\cdot|a,\omega^\prime)}{q(\cdot|a,\omega)}\right)^d\leq 1$ for any $d>0$, then we immediately have $\mathbb E \left(\ln\frac{q(\cdot|a,\omega^\prime)}{q(\cdot|a,\omega)}\right)< 0$.\footnote{The converse is not true: even if the inequality based on KL divergence holds in a small neighborhood of $\omega$, this does not imply there exists $d>0$ such that the inequality based on likelihood ratio holds uniformly for all $\omega^\prime$ in the neighborhood.}

\begin{thm}
\label{thm:Local-robustness-sufficient}Suppose model $\theta\in\Theta^{\bar\theta}$ admits a pure p-absorbing BN-E $\sigma$ at which $\theta$ is locally dominant within the $\bar\theta$-family. The following are true: 
\begin{itemize}
    \item[(i)] Model $\theta$ is $\bar\theta$-constrained locally robust.
    \item[(ii)] If $\theta$ has no traps, then it is $\bar\theta$-constrained locally robust at all full-support priors.
\end{itemize}
\end{thm}

Theorem \ref{thm:Local-robustness-sufficient} confirms our conjecture that the existence of a pure p-absorbing Berk-Nash equilibrium and local dominance are sufficient for constrained local robustness. Similar to local robustness, the notion has no prior tightness requirement under the no-trap conditions.\footnote{Note that we have to replace the SCE in the no-trap conditions with a BN-E. } Theorem \ref{thm:Local-robustness-sufficient} can be generalized to the case with a mixed p-absorbsing BN-E, but the statement is much more involved.\footnote{Additional conditions are needed if $\sigma$ is not pure. Suppose $\sigma$ is a p-absorbing mixed BN-E. Unless $\sigma$ is self-confirming, when the agent only plays actions in the support of $\sigma$, the parameters that empirically best fit the observed data can change with the empirical action frequency and are not necessarily be given by $\Omega^\theta (\sigma)$. Hence, the likelihood ratio between a parameter in $\Omega^{\bar\theta}$ and $\theta$ cannot be bounded using the likelihood ratio between that parameter and a KL-minimizer in $\Omega^\theta(\sigma)$. To generalize Theorem \ref{thm:Local-robustness-sufficient}, we need the existence of a p-absorbing mixed BN-E $\sigma$ such that $\Omega^\theta(\sigma^\prime)=\Omega^\theta(\sigma)$ for each $\sigma^\prime\in\Delta\supp(\sigma)$ and $\theta$ is locally dominant at $\sigma$ within the $\bar\theta$-family. } 
Analogously, we can characterize p-absorbingness by quasi-strictness. Since we focus on a pure equilibrium, this further simplifies to strictness.

\begin{cor}
\label{cor:Local-robustness-sufficient}Model $\theta\in\Theta^{\bar\theta}$ is $\bar\theta$-constrained locally robust if $\theta$ admits a strict BN-E at which $\theta$ is locally dominant within the $\bar\theta$-family. 
\end{cor}

\section{\label{sec:Applications}Applications}

I present two applications to demonstrate how results in this paper uncover new insights about the persistence of misspecified models. The first application illustrates our earlier finding that simple misspecified models may have better robustness properties than some correctly specified models. The second application revisits the comparison between over- and underconfidence in more general environments. 

\subsection{Media Bias, Extremism, and Polarization\label{app:media bias}}
In this application, I consider a stylized model of media consumption and demonstrate how misconceptions about media bias \citep{mediabias2005} can lead to stable polarization in political views despite no individual partisan bias. The misspecified model with these misconceptions and an extremism bias is globally robust regardless of the initial conditions. Even more surprisingly, people may abandon a correctly specified model, switch to such a misspecified model and then get stuck forever.

The agent has access to three media outlets and in each period she chooses one to consume, $\mathcal A=\{a^L,a^M,a^R\}$. The media outlets are indexed by their political leanings, left-wing, neutral, or right-wing. Each media outlet delivers two types of news, $\mathcal Y=\{l,r\}$, where $l$ represents good stories for the leftists and $r$ represents good stories for the rightists. The unknown state of the world $\omega\in\Omega=\{\omega^L,\omega^M,\omega^R\}$ governs the fraction of $l$ and $r$ stories happened in the real world and it remains fixed throughout the life of the agent. In particular, 60\% of the stories are $l$ stories ($r$ stories) in state $\omega^L$ ($\omega^R$), while an equal share of $l$ and $r$ stories happen in state $\omega^M$. The three media outlets differ in their ways of news reporting: in each state of the world, media $a^M$ truthfully reports the stories without bias, media $a^L$ selectively reports $l$ more than media $a^M$, and media $a^R$ selectively reports $r$ more than media $a^M$. The left panel of Table \ref{table:media-stories-true} summarizes the true fraction of $l$ stories reported by the media in different states. We restrict attention on the world in state $M$, where the true fractions of $l$ stories reported by the three media are given by $(0.6,0.5,0.4)$. 

\begin{table}[t]\label{table:media-stories-true}
    \centering
\begin{tabular}{c|c c c}
    $q^\theta(l|a,\omega)$ & $\omega^L$ & $\omega^M$ & $\omega^R$  \\ \hline
     $a^L$ & 0.7 & 0.6 & 0.5 \\
     $a^M$ & 0.6 & 0.5 & 0.4 \\
     $a^R$ & 0.5 & 0.4 & 0.3 
\end{tabular}\quad\quad 
\begin{tabular}{c|c c}
    $q^{\hat\theta}(l|a,\omega)$ & $\omega^L$ & $\omega^R$  \\ \hline
    $a^L$ & 0.6 & 0.5 \\
    $a^M$ & 0.5 & 0.5 \\
    $a^R$ & 0.5 & 0.4 
\end{tabular}
\caption{The left panel summarizes the true fraction of $l$ stories reported by each media outlet in each state of the world. It is also a description of the correctly specified model $\theta$. The right panel describes the predictions of a misspecified model $\hat\theta$.  
}
\end{table}

In this exercise, we focus attention on the comparison between two different models $\theta$ and $\hat\theta$ that I describe in Table \ref{table:media-stories-true}. Model $\theta$ is correctly specified: a $\theta$-modeler realizes that $\omega^M$ is a possible state of the world and are fully aware of the bias of both the left-wing and the right-wing media outlets. By contrast, model $\hat\theta$ is misspecified in two aspects. First, $\hat\theta$ features \emph{extremism} because it only recognizes the possibility of the extreme states $\omega^L$ and $\omega^R$. Second, $\hat\theta$ features \emph{naivety} about media bias: a $\hat\theta$-modeler underestimates the selective reporting bias of the left-wing $a^L$ and right-wing media $a^R$, and also underestimates the informativeness of the neutral media. As a result, when a $\hat\theta$-modeler subscribes to the left-wing media and finds that 60\% of the stories are good stories for leftists, she does not interpret it as evidence for the middle state $\omega^M$ (which does not exist in her extreme worldview), but treats it as evidence for the left state $\omega^L$; a similar logic applies to the right-wing media. She also mistakenly thinks that the reporting of the neutral media is totally uninformative about the state.

To highlight the core mechanism, I abstract away from specifying the payoff structure and outline the minimal assumptions that allow us to apply the characterization theorems in Section \ref{sec:Robustness}. It is straightforward to verify using Table \ref{table:media-stories-true} that the SCE supporting beliefs mentioned in Assumption \ref{assu:media bias} are indeed consistent.

\begin{assumption}\label{assu:media bias}
When the true state is $\omega^M$, the following are true:
\begin{itemize}
    \item[(i)] Model $\theta$ admits a unique SCE $a^M$ and it is strict, supported by belief $\delta_{\omega^M}$. 
    \item[(ii)] Model $\hat\theta$ admits only two strict SCEs, $a^L$ and $a^R$, supported by $\delta_{\omega^L}$ and $\delta_{\omega^R}$, respectively.
\end{itemize} 
\end{assumption}

Assumption \ref{assu:media bias} is natural and intuitive. With the correctly specified model $\theta$, the agent infers the true state and subscribes to the neutral media. With the misspecified model $\hat\theta$, however, the agent develops partisan bias and only subscribes to the media biased towards her political belief. The choices of the agent can be justified as the result of maximizing the sum of emotional and informational value from news consumption.\footnote{A micro-foundation is provided in Appendix \ref{sec:Online-Appendix}.}

Model $\hat \theta$ has an advantage over model $\theta$ due to its extremeness. Since both models $\theta$ and $\hat\theta$ admit at least one SCE,  Theorem \ref{thm:Global-Robustness} tells us that both models are globally robust at some prior. Interestingly, Theorem \ref{thm:concentrated-prior} implies a counter-intuitive result (see Proposition \ref{prop: media-bias-1} below): model $\theta$ is globally robust only when the associated prior assigns high enough probability to the true state $\omega^M$, while model $\hat\theta$ is globally robust at \emph{all} priors. In other words, model $\hat\theta$ is globally robust in a \emph{robust} way.

\begin{prop}\label{prop: media-bias-1}
Fix any $\alpha>1$. Model $\theta$ is globally robust at prior $\pi_0^\theta$ if and only if $\pi_0^\theta(\omega^M)\geq 1/\alpha$, while model $\hat\theta$ is globally robust at all priors. 
\end{prop}

Despite being misspecified, model $\hat\theta$ has a stronger global robustness property, because its narrative is simple, coherent, and balanced. To see this, notice that all parameter values in model $\hat\theta$ are consistent, i.e. $C^{\hat\theta}=\{\omega^L,\omega^R\}=\Omega^{\hat\theta}$. This makes it possible for model $\hat\theta$ to persist against any competing model. For example, model $\hat\theta$ can outperform a left-biased competing model in explaining the data when the agent happens to read a series of $r$ stories, and similarly it can outperform a right-biased competing model when the agent happens to read a series of $l$ stories. If the competing model is unbiased and correctly specified such as model $\theta$, model $\hat\theta$ still persists because of its simplicity.
 
Proposition \ref{prop: media-bias-1} characterizes the robustness properties of $\theta$ and $\hat\theta$ separately with the implicit assumption that they are the initial model choice of a switcher. What if a switcher originally adopts $\theta$ and entertains $\hat\theta$ as the competing model? Whether she will abandon $\theta$ in favor of $\hat\theta$ is \textit{a priori} unclear. While $\theta$ may not be globally robust at a given prior, this only tells us that $\theta$ does not persist against some competing model, but this competing model may not be $\hat\theta$.  Surprisingly, as I show in Proposition \ref{prop: media-bias-2}, $\hat\theta$ indeed replaces $\theta$ with positive probability if the switching threshold is low.

\begin{prop}\label{prop: media-bias-2}
Fix any full-support priors $\pi_0^\theta,\pi_0^{\hat\theta}$ and any $\alpha< 1/\pi_0^\theta(\omega^M)$. In the switcher's problem $(\theta,\hat\theta,\pi_0^\theta,\pi_0^{\hat\theta})$, the model choice $m_t$ eventually equals $\hat\theta$ with positive probability.
\end{prop}

In summary, this application generates three novel insights about news consumption and political beliefs. First, extremism and naivety about media bias go hand in hand and their persistence is robust against arbitrary competing narratives. Second, individuals may abandon their correct models and switch to incorrect alternatives because of their extremeness/simplicity. Third, even though the extreme and naive model has no built-in political bias, individuals who hold such a model gradually develop a strong partisan bias over time. The direction of the partisan bias is random and path-dependent, leading to long-term political polarization.

\subsection{\label{subsec:Positive-Negative}Overconfidence and Underconfidence}

In this application, I compare the robustness properties of over- and underconfidence in more general environments. I restrict attention on the prior-free local robustness notions since the interesting difference between over- and underconfidence only concerns the induced equilibria. I show that under natural assumptions, any level of overconfidence is locally robust while underconfidence is locally robust only on a union of unconnected intervals. 

This result breaks the symmetry between overconfidence and unconfidence and provides a novel mechanism for why we might expect one bias to be more persistent than the other. A plethora of evidence in psychology and economics suggests that overconfidence is more prevalent than underconfidence, and many hold the view that this is because agents derive ego utility from holding overconfident beliefs about their own positive traits \citep{brunnermeier2005optimal,koszegi2006ego,oster2013optimal}. By contrast, I provide a reason rooted in the learning environment itself: overconfidence has better robustness properties than underconfidence when the agent can switch models. 

As in our motivating example, an agent chooses effort $a_t$ from a finite set $\mathcal A$ in each period. The agent has payoff $u(a_t,y_t)=y_t$, where $y_t$ is the output of his work, including possibly any cost of effort. The output takes the form of $y_t=g(a_t,b^\ast,\omega^\ast)+\eta_t$, where function $g$ is twice continuously differntiable and strictly increasing in $b$ and $\omega$, $b^\ast\in[\underline b,\overline b]$ represents the agent's ability, and $\omega^\ast\in[\underline\omega,\overline\omega]$ captures a fundamental of the outside environment, such as the market demand or the quality of the agent's organization, and $\eta_t$ follows a known zero-mean normal distribution. The output function is strictly concave in effort $a_t$. In addition, following \cite{heidhues2018unrealistic}, I assume that the optimal effort decreases in the fundamental and weakly decreases in one's ability, as captured by Assumption \ref{assu:application-2}.\footnote{The assumption that effort and the fundamental are complements are natural but not critical. If we alter the orientation of $a_t$, then effort and the fundamental become substitutes, but all results remain unchanged. The important assumption that leads to positive belief reinforcement for overconfidence and negative belief reinforcement for underconfidence is that $g_{a\omega}$ has a different sign or is much larger than $g_{ab}$.  }
\begin{assumption}\label{assu:application-2} Function $g$ satisfies  $g_{ab}\coloneqq \frac{\partial^2 g}{\partial a\partial b}\leq0$ and $g_{a\omega}\coloneqq \frac{\partial^2 g}{\partial a\partial \omega}>0$.
\end{assumption}

I consider misspecified models that assign probability $1$ to some $\hat b\in[\underline b,\overline b]$ which deviates from its correct value. The agent is dogmatically overconfident about his ability when $\hat b>b^\ast$ and dogmatically underconfident when $\hat b<b^\ast$. To avoid trivial cases of non-robustness, I focus on models whose parameter sets are \emph{complete}: if the model assigns probability 1 to $\hat b$, then  for all $a\in\mathcal A$, the set $\Omega^\theta$ contains the value $\hat\omega$ that satisfies $g(a,\hat b,\hat \omega)=g(a,b^\ast,\omega^\ast)$. 

Under Assumption \ref{assu:application-2}, beliefs about the fundamental are \emph{positively reinforcing} when the agent is overconfident and \emph{negatively reinforcing} when underconfident, as summarized below. 
    
\begin{property}[Belief reinforcement]\label{property:application-2}
When the agent is overconfident (underconfident), higher beliefs lead to higher optimal actions, i.e. $\max A^\theta(\delta_{\omega^\prime})\leq\min A^\theta(\delta_{\omega^\dprime})$ for all $\omega^\dprime>\omega^\prime$, and higher actions leads to higher (lower) beliefs, i.e. $\max \Omega^\theta(\delta_{a^\prime})\leq\min \Omega^\theta(\delta_{a^\dprime})$ ($\min \Omega^\theta(\delta_{a^\prime})\geq\max \Omega^\theta(\delta_{a^\dprime})$) for all $a^\dprime>a^\prime$.
\end{property}

The first part that higher beliefs induce higher actions follows from the assumption that one's effort and the fundamental are complements. The intuition for the second part is similar to the motivating example: as a result of overconfidence (underconfidence), he underestimates (overestimates) the fundamental; when a higher action is played, the return to the fundamental $\omega$ is higher because $g_{a\omega}>0$, and thus the positive (negative) gap between the true state $\omega^\ast$ and the inferred fundamental $\hat \omega$ should be smaller such that expectations meet the reality, implying that the inferred state $\hat \omega$ is larger (smaller). Proposition \ref{prop:application-2} characterizes the robustness properties for different self-perceptions. 


\begin{prop}\label{prop:application-2}
Suppose model $\theta$ has a dogmatic self-perception $\hat b$ and a complete parameter space, then the following are true:
\begin{itemize}
    \item[(i)] Model $\theta$ with any level of overconfidence $\hat b>b^\ast$ is locally robust.
    \item[(ii)] There exists a strictly decreasing sequence $\beta_N<...<\beta_1<\beta_0=b^\ast$ such that, model $\theta$ with underconfidence is locally robust if $\hat b\in(\beta_{2k+1},\beta_{2k})$ for any $k\in\mathbb N$ and not locally robust if $\hat b\in(\beta_{2k},\beta_{2k-1})$ for some $k\in\mathbb N_+$.
\end{itemize}
\end{prop}

Proposition \ref{prop:application-2} shows that overconfidence is locally robust, but underconfidence is only locally robust on unconnected intervals.
The rest of this section illustrates the mechanism behind this result. First note that in this environment, if any pure Berk-Nash equilibrium induced by model $\theta$ must be self-confirming because the agent can, by the completeness of the parameter set, perfectly justify his observations by forming an incorrect belief over the fundamental. By contrast, a mixed Berk-Nash equilibrium can never be self-confirming in this environment---the agent necessarily finds his self-perception $\hat b$ inconsistent with the observations because no single value of the fundamental can reconcile the outcome distributions at more than two effort levels when $\hat b\not=b^\ast$. It can be shown that any competing model with a self-perception that is slightly closer to the truth would fit better uniformly at any effort levels. Therefore, Proposition \ref{prop:application-2} follows from Theorem \ref{thm:Global-Robustness}, provided that we can show overconfidence ensures the existence of a pure and p-absorbing BN-E while underconfidence sometimes gives rise to mixed BN-Es only. 

To illustrate, suppose the agent chooses from three effort levels $a^\prime<a^\ast<a^\dprime$ where $a^\ast$ is the unique optimal effort level when the agent knows the true DGP. Now consider an overconfident agent with $\hat b>b^\ast$ and suppose there exists a mixed BN-E in which the agent finds both $a^\prime$ and $a^\ast$ optimal against a supporting belief $\delta_{\hat \omega}$. Then for any assessment about the fundamental that is lower than $\hat\omega$, he strictly prefers the lower effort $a^\prime$, which, due to positive reinforcement, implies that he indeed finds a lower assessment of the fundamental more accurate. Therefore, effort $a^\prime$ must also constitute a SCE and, since it is strict, Corollary \ref{corr:Suff-robustness} implies the model is locally robust. 

Next let us turn to an underconfident agent with $\hat b<b^\ast$. Suppose there exists a mixed BN-E in which the agent finds both $a^\ast$ and $a^\dprime$ optimal against belief $\delta_{\hat \omega}$. Then for any higher assessment of $\omega$, he strictly prefers the higher effort $a^\dprime$, which, due to negative reinforcement, induces his belief over $\omega$ to drift downwards; by contrast, for any lower assessment of $\omega$, he strictly prefers the lower effort $a^\ast$, which induces his belief to drift upwards. Therefore, the mixed BN-E is the only BN-E and thus the model is not locally robust. That being said, as $\hat b$ further decreases, the agent might find the high effort $a''$ strictly optimal, upon which the unique BN-E becomes pure and thus self-confirming. This then gives rise to unconnected intervals of non-robustness and robustness as described in Proposition \ref{prop:application-2}.

\section{\label{sec:extension}Extensions and Discussion}

\subsection{\label{sec:multiple-competing-models}Multiple Competing Models}
This extension explores the consequence of the agent entertaining multiple competing models. The framework described in Section \ref{sec:Framework} could be easily extended to accommodate more than one competing models. Let $\Theta'\subseteq \Theta$ denote the finite subset of competing models that the agent entertains in the beginning and $\Theta^\dagger\coloneqq \Theta' \cup \{\theta\}$ denote the set of all models entertained. In the beginning of each period $t$, she compares her current model against all alternatives and switches to the one with the highest likelihood ratio if it exceeds the switching threshold $\alpha$. Specifically, the agent calculates the likelihood ratios between each model in $\Theta^\dagger$ and the model she used in last period, $\boldsymbol\lambda_t\coloneqq (\lambda_t^{\theta'})_{\theta'\in\Theta^\dagger}$, where
\begin{equation}
    \lambda_t^{\theta'} = \ell_t(\theta') /\ell_t(m_{t-1}).
\end{equation}
The agent then makes a switch if $\max_{\theta'\in\Theta^\dagger} \lambda_t^{\theta'}>\alpha$ and switch to the model $\theta'$ with the highest likelihood ratio.

The definitions of persistence and all notions of robustness can be modified by simply replacing $\theta'$ by $\Theta'$.  Suppose the agent entertains at most $K\geq 1$ competing models, then global robustness would require $\theta$ to persist against every $\Theta'\subseteq\Theta$ of size no larger than $K$ at all priors assigned to models in $\Theta'$. Local robustness notions can be similarly extended.\footnote{\label{fn:persistence-against-each-model}Note that if $\Theta'$ is not a singleton, then persistence against $\Theta'$ is not equivalent to persistence against each model in $\Theta'$, and neither implies the other. See Appendix \ref{sec:Online-Appendix} for examples.} 


Interestingly, the consideration of multiple competing models introduces overfitting. In particular, when the switching threshold $\alpha$ is not adjusted as the number of competing models $K$ becomes larger, even the true DGP may fail to be globally robust. Relatedly, \cite{schwartzstein2019using} find in a static setting, a decision maker switches from the true DGP to a competing model when a persuader is allowed to propose one after the data is realized. By contrast, I show that the persuader can achieve the same goal even if he has to propose before the outcomes are drawn, provided that he can present multiple competing models. More importantly, not only must the switcher switch at least once to the competing models, due to sticky switching, she may eventually settle down with one of them despite that the true DGP fits the data perfectly on average. For concreteness, I now construct a decision problem and $K$ competing models such that the true DGP $\theta^\ast$ does not persist when $K>\alpha+1$. 

\begin{example}[Overfitting] \label{exa:multiple-competing-models}
The agent has two actions $\mathcal A=\{a^\prime,a^\dprime\}$. The true DGP is a uniform distribution over $K$ outcomes, $\mathcal Y=\{1,...,K\}$, regardless of the action. The agent incurs a loss of $-K$ when drawing the outcome $y=1$ and receives a payoff of $0$ if any other outcome is realized. The agent pays an additional cost $c>0$ for playing $a^\prime$ and no cost if she plays $a^\dprime$. The agent's initial model $\theta$ is the true DGP. Hence, the agent optimally plays $a^\dprime$ in the first period to avoid the cost. 

Suppose the agent evaluates $K$ competing models that I describe below. Each model $\theta^k\in\{\theta^1,...,\theta^K\}$ has a single parameter $\omega^k$. When $a'$ is played, model $\theta^k$ agrees with the true DGP, and its prediction  corresponds to a uniform distribution. When $a''$ is played, model $\theta^k$ disagrees with the true DGP. For any $k>1$, $\theta^k$ predicts

\begin{equation*}
  q^{\theta^k}(y|a^\dprime,\omega^k) =
    \begin{cases}
      1-\frac{1}{K}-(K-1)\eta & \text{if }{y=n},\\
      \frac{1}{K}+\eta & \text{if }{y=1},\\
      \eta & \text{if }y\in\mathcal Y \setminus \{1,n\},
    \end{cases}
\end{equation*}
where $\eta$ is a small positive constant. When $k=1$, $q^{\theta^k}(\cdot|a^\dprime,\omega^k)$ is given by 
\begin{equation*}
  q^{\theta^1}(y|a^\dprime,\omega^1) =
    \begin{cases}
      1-(K-1)\eta & \text{if }{y=1,}\\
      \eta & \text{if }y\in\mathcal Y \setminus\{1\}.
    \end{cases}
    \end{equation*}
Importantly, model $\theta^k$ predicts that when $a^\dprime$ is played, the outcome $k$ is realized with probability near $1$. Since there is one such model for every possible outcome, the agent must switch to one of these competing models after the first period. In particular, if the realized outcome is $k$, the agent immediately switches to model $\theta^k$ when $\eta$ is sufficiently small, $$\frac{\ell_1(\theta^k)}{\ell_1(\theta)}\geq \frac{1-\frac{1}{K}-(K-1)\eta}{\frac{1}{K}}>\alpha,$$ where such $\eta$ exists because $K>\alpha+1$. 

Next, since each competing model assigns a probability larger than $1/K$ to the outcome $1$, once the switch happens, the agent finds it optimal to play $a^\prime$ to avoid the loss associated with outcome $1$ as long as $c$ is sufficiently small. However, since all models have the same correct predictions under $a^\prime$, the likelihood ratios remain unchanged thereafter. Hence, despite initially having the true model, the agent becomes permanently trapped with a wrong model and inefficient play. 
\end{example}

The trap described in Example \ref{exa:multiple-competing-models} is different from the trap we constructed in Example \ref{exa:endogenous-model-choice}. To see that, note that $\theta^\ast$ satisfies the no-trap conditions in Definition \ref{def:no-traps} since it is identifiable and has a quasi-strict SCE. The agent in Example \ref{exa:multiple-competing-models} gets trapped because $a'$ is strictly dominant under the competing models and all models have the exact same predictions once $a'$ is being played, eliminating any possibility of future learning and switching. However, the driving force that leads the agent into the trap in the first place is indeed similar in the two examples. In Example \ref{exa:endogenous-model-choice}, the agent holds a diffuse prior, rendering his model choice to be sensitive to early outcome realizations. In Example \ref{exa:multiple-competing-models}, overfitting also arises in the short term and prompts an early switch to other models. The more competing models the agent evaluates, the more likely such a switch occurs. 

Therefore, a natural remedy is to make switching stickier so that agent is less responsive to early outcome realizations. Indeed, Theorem \ref{thm:multiple-competing-models} shows that if $\alpha>K$, the scope of local and global robustness does not change at all and Theorem \ref{thm:Global-Robustness}  fully generalizes to this environment. A larger bound may be needed for constrained locally robustness because local dominance is weaker than the self-confirming property, but when $\alpha$ is large enough, Theorem \ref{thm:Local-robustness-sufficient} also generalizes.

\begin{thm} \label{thm:multiple-competing-models}
Suppose the agent evaluates at most $K$ competing models. 
\begin{enumerate}
    \item Suppose $\alpha>K$. Model $\theta$ is locally and globally robust for at least one prior if and only if there exists a p-absorbing SCE under $\theta$. 
     \item Suppose $\alpha>K^{1/d}$ for some $d>0$. Model  $\theta\in\Theta^{\bar\theta}$ is $\bar\theta$-constrained locally robust if $\theta$ admits a pure p-absorbing BN-E at which $\theta$ is locally dominant within the $\bar\theta$-family and the local dominance condition in Definition \ref{def:local-dominance} holds at $d$.
\end{enumerate}
\end{thm}

This extension shows that the consideration of multiple competing models at the same time can make persistence more difficult due to overfitting. While one may expect  such considerations to work towards the direction of reducing robust misspecification, they turn out to have as much bearing on correctly specified models as on misspecified models---in particular, as demonstrated in Example \ref{exa:multiple-competing-models}, even the true model may not persist. After adjusting the switching threshold for the increasing number of competing models, our main characterization remains valid, continuing to allow both correctly specified and misspecified models to be robust.


\subsection{\label{subsec:patient agent}Non-myopic Agent}
Our baseline framework focuses on a myopic agent and rules out any experimentation motives. This assumption can be less substantial than one might think. In this subsection, I discuss two potential ways of relaxing this assumption. 

First, we may assume that the agent is non-myopic within each model but maintain that she is myopic across models. That is, when choosing an optimal action, the agent maximizes her expected discounted sum of payoffs assuming that she keeps her current model $m_{t}$ in the future. An optimal policy $f^\theta$ solves the following dynamic programming problem,
\begin{equation*}
U^{\theta}\left(\pi^{\theta}_t\right)=\max_{a\in\mathcal{A}}\sum_{\omega\in\Omega^{\theta}}\pi^\theta\left(\omega\right)\int_{y\in\mathcal{Y}}\left[u\left(a,y\right)+\delta U^{\theta}\left(B^{\theta}\left(a,y,\pi^{\theta}_t\right)\right)\right]q^{\theta}\left(y|a,\omega\right)v\left(dy\right).\label{eq:dynamic-programming}
\end{equation*}
How should we interpret the asymmetry between experimentation within models and no experimentation across models? This asymmetry again highlights the stickiness of switching models as opposed to the smoothness of Bayesian updating, and it is plausible when resources are constrained.  For instance, consider an applied data scientist who uses one single model to guide data collection and make policy recommendations. While he is aware of potential misspecification, he may choose not to spare valuable resources in additional experiments to find the best model. However, he may indeed switch to a different model if the data at hand happens to suggest its superiority.\footnote{This assumption is also natural in organizations where decision making and model estimation are handled by separate teams. For example, a manager (e.g. the chairman of a central bank) chooses policies based on the predictions made by the research team (e.g. a group of macroeconomists), while the research team focuses on estimating the models given the available data.  }

If we relax the myopicity assumption this way, Theorems \ref{thm:Global-Robustness} to \ref{thm:Local-robustness-sufficient} go through without changes. This claim may appear surprising at first, because  experimentation motives should make it harder to sustain a self-confirming equilibrium or a Berk-Nash equilibrium, and thus the set of robust misspecified models might be smaller if the agent is non-myopic. This intuition is correct---as the agent becomes more patient, p-absorbingness is harder to achieve. However, note that the theorems only establish the equivalence relationship between the existence of p-absorbing equilibria and the models' robustness properties, so whether p-absorbingness can be achieved is irrelevant. 
In the Appendix, I provide stronger sufficient conditions for p-absorbingness such that variants of Corollaries \ref{corr:Suff-robustness} and \ref{cor:Local-robustness-sufficient} continue to hold.\footnote{In particular, any uniformly quasi-strict SCE is p-absorbing, and so is any uniformly strict BN-E.}

Alternatively, we may assume the agent is forward-looking both within and across models.  If  the agent anticipates future model switches, she may intentionally take actions that allow her to distinguish different models, even if her current model predicts a different optimal action. Characterizing robust models in this environment is significantly more challenging and beyond the scope of this paper. I conjecture that the set of robust misspecified models will shrink as the agent becomes increasingly patient.

\subsection{\label{subsec:alternative def persistence}Alternative Definitions of Persistence}

This subsection discusses alternative definitions of model persistence. Our definition of persistence in Section \ref{sec:Framework} requires that if a switcher initially adopts this model, she eventually settles down with it with positive probability. This concept has a natural interpretation and can be used to predict whether a particular bias is likely to exist in a stable form. However, by relaxing or strengthening different parts of this definition, we can obtain a couple of variants that are also worth exploring. These alternative definitions are useful if one looks for models with stronger or weaker persistence properties. This investigation also helps us better understand the original definition since it sheds light on the importance of the different parts of the concept.

\paragraph{Almost sure eventual adoption.} The first natural extension is to strengthen persistence by requiring that the model is eventually adopted with probability 1. That is, any such model is guaranteed to win out in the competition. But almost-sure persistence turns out to make global and local robustness impossible. In fact, given any model $\theta$ (including the true model $\theta^\ast$), we can easily construct a nearby competing model $\theta'$ such that the competing model is eventually adopted with positive probability. The idea is that the agent can draw a sequence of outcome realizations that can be better explained by the competing model, and once a switch happens, the agent does not feel compelled to switch back since the predictions of the two models are identical in the limit. Therefore, almost sure eventual adoption can be too restrictive to be any useful.\footnote{To do this, let us construct $\theta'$ such that it contains all DGPs in $\theta$ and one additional DGP that differs from any other DGPs for all actions in $\theta$. That is, we have $\Omega^{\theta'}=\Omega^{\theta}\cup \{\hat\omega\}$, where $q^{\theta'}(\cdot|a,\omega)=q^{\theta}(a|\cdot,\omega)$ and $q^{\theta'}(\cdot|a,\hat\omega)\not=q^{\theta}(\cdot|a,\omega)$ for all $\omega\in\Omega^\theta$ and all actions $a\in\mathcal A$. In addition, let the prior $\pi_0^\theta$ be proportional to $\pi_0^{\theta'}$ for the all shared parameters. With this structure, the Bayes factor $\lambda_t$ is bounded below by $ \pi_0^{\theta'}(\Omega^\theta)$. Note that since $\hat \omega$ predicts differently from model $\theta$, it is always a positive-probability event that the agent finds model $\theta'$ particularly compelling and make a switch because of the existence of $\hat\omega$. But then the agent never switches back if the lower bound of the Bayes factor, $\pi_0^{\theta'}(\Omega^\theta)$, is higher than $1/\alpha$, which can be achieved if we make $\pi_0^{\theta'}(\Omega^\theta)$ sufficiently close to 1.} 

\paragraph{No switch.} The current definition of persistence allows for back-and-forth switching before the eventual adoption of the model. A more conservative  definition could have required the agent to adopt the same model throughout. It turns out that all theorems continue to hold even with this conservative definition. Indeed, model switching does not play any role in ensuring local and global robustness when there are no traps---in the proof of all main theorems, I show that there exists a sequence of outcome realizations that induce the agent to play a SCE while remaining under the same model. When there are traps in the model, however, a temporary switch to the competing model can be instrumental to the persistence of the initial model because switching to the competing model may happen to keep the agent away from the traps. A full characterization of this case is left for future research.

\section{\label{sec:Concluding-Remarks}Concluding Remarks}
In this paper, I propose a new theoretical framework to study the persistence of misspecified models when decision makers are aware of potential model misspecification. I introduce sticky switching to the standard model of individual active learning  and study the limit of the model choices. I explore three different robustness notions and use them to derive novel insights about which models persist and when they persist. I show that all three robustness notions can be characterized in terms of two properties, asymptotic accuracy and prior tightness.

The idea that the existence of self-confirming equilibria can explain why incorrect models persist can be found in the existing literature \citep{heidhues2019overconfidence}. Instead of assuming that the agent starts outright from an equilibrium, my framework incorporates full-fledged model switching dynamics into active learning processes. The characterization highlights the importance of this consideration. Robustness not only requires the existence of a self-confirming equilibrium but also needs it to be p-absorbing, which connects the notion of model persistence with the stability of equilibria. Furthermore, global robustness requires high prior tightness around the set of p-absorbsing self-confirming equilibria. This finding provides a theoretical justification for the empirical observation that simple narratives and entrenched worldviews tend to be more persistent. 

The model-switching framework has great application value. My characterization of robust models, stated in the form of simple criteria that can be easily verified from the primitives, provide a learning foundation for certain misspecified models, some of which are already studied in misspecified learning literature.  It can also be used to predict the persistence of given behavioral biases in specific contexts, which can be useful for guiding empirical work on behavioral economics and relevant policy making.

Within this general framework of model switching, there are many other interesting questions to pursue. For example, persistence requires a positive chance of eventual adoption, but this concept is silent about the size of this probability. New insights may emerge from studying how this probability is determined by key primitives of the model, such as whether it is correctly specified or misspecified, and features of the learning environment, such as the switching stickiness. Another potentially fruitful direction could be to restrict attention to a given small set of models and fully characterize the dynamics of model choices, i.e. how a decision maker oscillates between two or more competing models persistently. 

\bibliographystyle{ecta}
\bibliography{references}

\begin{thebibliography}{59}
\newcommand{\enquote}[1]{``#1''}
\expandafter\ifx\csname natexlab\endcsname\relax\def\natexlab#1{#1}\fi

\bibitem[\protect\citeauthoryear{Akaike}{Akaike}{1974}]{akaike1974new}
\textsc{Akaike, H.} (1974): \enquote{A New Look at the Statistical Model
  Identification,} \emph{IEEE Transactions on Automatic Control}, 19, 716--723.

\bibitem[\protect\citeauthoryear{Ba and Gindin}{Ba and
  Gindin}{2022}]{bagindin2021overconfidence}
\textsc{Ba, C. and A.~Gindin} (2022): \enquote{A Multi-Agent Model of
  Misspecified Learning with Overconfidence,} \emph{Available at SSRN}.

\bibitem[\protect\citeauthoryear{Battigalli}{Battigalli}{1987}]{Battigalli1987}
\textsc{Battigalli, P.} (1987): \enquote{Comportamento Razionale ed Equilibrio
  nei Giochi e nelle Situazioni Sociali,} \emph{Unpublished thesis}.

\bibitem[\protect\citeauthoryear{Battigalli, Francetich, Lanzani, and
  Marinacci}{Battigalli et~al.}{2019}]{battigalli2019learning}
\textsc{Battigalli, P., A.~Francetich, G.~Lanzani, and M.~Marinacci} (2019):
  \enquote{Learning and Self-Confirming Long-Run Biases,} \emph{Journal of
  Economic Theory}, 183, 740--785.

\bibitem[\protect\citeauthoryear{Blackwell}{Blackwell}{1965}]{blackwell1965discounted}
\textsc{Blackwell, D.} (1965): \enquote{Discounted Dynamic Programming,}
  \emph{The Annals of Mathematical Statistics}, 36, 226--235.

\bibitem[\protect\citeauthoryear{Bohren}{Bohren}{2016}]{bohren2016informational}
\textsc{Bohren, J.~A.} (2016): \enquote{Informational Herding with Model
  Misspecification,} \emph{Journal of Economic Theory}, 163, 222--247.

\bibitem[\protect\citeauthoryear{Bohren and Hauser}{Bohren and
  Hauser}{2021}]{bohren2021learning}
\textsc{Bohren, J.~A. and D.~N. Hauser} (2021): \enquote{Learning with
  Heterogeneous Misspecified Models: Characterization and Robustness,}
  \emph{Econometrica}, 89, 3025--3077.

\bibitem[\protect\citeauthoryear{Braungart and Braungart}{Braungart and
  Braungart}{1986}]{braungart1986life}
\textsc{Braungart, R.~G. and M.~M. Braungart} (1986): \enquote{Life-Course and
  Generational Politics,} \emph{Annual Review of Sociology}, 12, 205--231.

\bibitem[\protect\citeauthoryear{Brunnermeier and Parker}{Brunnermeier and
  Parker}{2005}]{brunnermeier2005optimal}
\textsc{Brunnermeier, M.~K. and J.~A. Parker} (2005): \enquote{Optimal
  Expectations,} \emph{American Economic Review}, 95, 1092--1118.

\bibitem[\protect\citeauthoryear{Chernoff}{Chernoff}{1954}]{chernoff1954distribution}
\textsc{Chernoff, H.} (1954): \enquote{On the Distribution of the Likelihood
  Ratio,} \emph{The Annals of Mathematical Statistics}, 573--578.

\bibitem[\protect\citeauthoryear{Cho and Kasa}{Cho and
  Kasa}{2015}]{cho2015learning}
\textsc{Cho, I.-K. and K.~Kasa} (2015): \enquote{Learning and Model
  Validation,} \emph{The Review of Economic Studies}, 82, 45--82.

\bibitem[\protect\citeauthoryear{Di~Stefano, Gino, Pisano, and
  Staats}{Di~Stefano et~al.}{2015}]{di2015learning}
\textsc{Di~Stefano, G., F.~Gino, G.~P. Pisano, and B.~Staats} (2015):
  \emph{Learning by Thinking: Overcoming the Bias for Action through
  Reflection}, Harvard Business School Cambridge, MA, USA.

\bibitem[\protect\citeauthoryear{Easley and Kiefer}{Easley and
  Kiefer}{1988}]{easley1988controlling}
\textsc{Easley, D. and N.~M. Kiefer} (1988): \enquote{Controlling a Stochastic
  Process with Unknown Parameters,} \emph{Econometrica: Journal of the
  Econometric Society}, 1045--1064.

\bibitem[\protect\citeauthoryear{Esponda and Pouzo}{Esponda and
  Pouzo}{2016}]{esponda2016berk}
\textsc{Esponda, I. and D.~Pouzo} (2016): \enquote{Berk--Nash Equilibrium: A
  Framework for Modeling Agents with Misspecified Models,} \emph{Econometrica},
  84, 1093--1130.

\bibitem[\protect\citeauthoryear{Esponda, Pouzo, and Yamamoto}{Esponda
  et~al.}{2021}]{esponda2021asymptotic}
\textsc{Esponda, I., D.~Pouzo, and Y.~Yamamoto} (2021): \enquote{Asymptotic
  Behavior of Bayesian Learners with Misspecified Models,} \emph{Journal of
  Economic Theory}, 195, 105260.

\bibitem[\protect\citeauthoryear{Eyster and Rabin}{Eyster and
  Rabin}{2005}]{eyster2005cursed}
\textsc{Eyster, E. and M.~Rabin} (2005): \enquote{Cursed Equilibrium,}
  \emph{Econometrica}, 73, 1623--1672.

\bibitem[\protect\citeauthoryear{Eyster and Rabin}{Eyster and
  Rabin}{2010}]{eyster2010naive}
---\hspace{-.1pt}---\hspace{-.1pt}--- (2010): \enquote{Naive Herding in
  Rich-Information Settings,} \emph{American Economic Journal: Microeconomics},
  2, 221--43.

\bibitem[\protect\citeauthoryear{Fisher}{Fisher}{1985}]{fisher1985narrative}
\textsc{Fisher, W.~R.} (1985): \enquote{The Narrative Paradigm: An
  Elaboration,} \emph{Communications Monographs}, 52, 347--367.

\bibitem[\protect\citeauthoryear{Frick, Iijima, and Ishii}{Frick
  et~al.}{2021}]{frick2021welfare}
\textsc{Frick, M., R.~Iijima, and Y.~Ishii} (2021): \enquote{Welfare
  Comparisons for Biased Learning,} \emph{Cowles Foundation Discussion Papers.
  2605.}

\bibitem[\protect\citeauthoryear{Fudenberg, Lanzani, and Strack}{Fudenberg
  et~al.}{2021}]{fudenberg2021limit}
\textsc{Fudenberg, D., G.~Lanzani, and P.~Strack} (2021): \enquote{Limit Points
  of Endogenous Misspecified Learning,} \emph{Econometrica}, 89, 1065--1098.

\bibitem[\protect\citeauthoryear{Fudenberg, Lanzani et~al.}{Fudenberg
  et~al.}{2022}]{fudenberg2022misspecifications}
\textsc{Fudenberg, D., G.~Lanzani, et~al.} (2022): \enquote{Which
  Misspecifications Persist?} \emph{Theoretical Economics}.

\bibitem[\protect\citeauthoryear{Fudenberg and Levine}{Fudenberg and
  Levine}{1993}]{fudenberg1993self}
\textsc{Fudenberg, D. and D.~K. Levine} (1993): \enquote{Self-Confirming
  Equilibrium,} \emph{Econometrica}, 523--545.

\bibitem[\protect\citeauthoryear{Fudenberg, Romanyuk, and Strack}{Fudenberg
  et~al.}{2017}]{fudenberg2017active}
\textsc{Fudenberg, D., G.~Romanyuk, and P.~Strack} (2017): \enquote{Active
  Learning with a Misspecified Prior,} \emph{Theoretical Economics}, 12,
  1155--1189.

\bibitem[\protect\citeauthoryear{Gagnon-Bartsch}{Gagnon-Bartsch}{2017}]{gagnon2017taste}
\textsc{Gagnon-Bartsch, T.} (2017): \enquote{Taste Projection in Models of
  Social Learning,} Tech. rep., Mimeo.

\bibitem[\protect\citeauthoryear{Gagnon-Bartsch, Rabin, and
  Schwartzstein}{Gagnon-Bartsch et~al.}{2020}]{gagnon2018channeled}
\textsc{Gagnon-Bartsch, T., M.~Rabin, and J.~Schwartzstein} (2020):
  \emph{Channeled Attention and Stable Errors}, Harvard Business School.

\bibitem[\protect\citeauthoryear{Galperti}{Galperti}{2019}]{galperti2019persuasion}
\textsc{Galperti, S.} (2019): \enquote{Persuasion: The Art of Changing
  Worldviews,} \emph{American Economic Review}, 109, 996--1031.

\bibitem[\protect\citeauthoryear{Gilboa and Schmeidler}{Gilboa and
  Schmeidler}{1989}]{GILBOA1989141}
\textsc{Gilboa, I. and D.~Schmeidler} (1989): \enquote{Maxmin Expected Utility
  with Non-Unique Prior,} \emph{Journal of Mathematical Economics}, 18,
  141--153.

\bibitem[\protect\citeauthoryear{Groseclose and Milyo}{Groseclose and
  Milyo}{2005}]{mediabias2005}
\textsc{Groseclose, T. and J.~Milyo} (2005): \enquote{A Measure of Media Bias,}
  \emph{The Quarterly Journal of Economics}, 120, 1191--1237.

\bibitem[\protect\citeauthoryear{Hansen and Sargent}{Hansen and
  Sargent}{2001}]{hansen2001robust}
\textsc{Hansen, L. and T.~J. Sargent} (2001): \enquote{Robust Control and Model
  Uncertainty,} \emph{American Economic Review}, 91, 60--66.

\bibitem[\protect\citeauthoryear{He}{He}{2022}]{he2022mislearning}
\textsc{He, K.} (2022): \enquote{Mislearning from Censored Data: The Gambler's
  Fallacy and Other Correlational Mistakes in Optimal-Stopping Problems,}
  \emph{Theoretical Economics}, 17, 1269--1312.

\bibitem[\protect\citeauthoryear{He and Libgober}{He and
  Libgober}{2020}]{he2020evolutionarily}
\textsc{He, K. and J.~Libgober} (2020): \enquote{Evolutionarily Stable (Mis)
  specifications: Theory and Applications,} \emph{arXiv preprint
  arXiv:2012.15007}.

\bibitem[\protect\citeauthoryear{Heidhues, K{\H{o}}szegi, and Strack}{Heidhues
  et~al.}{2018}]{heidhues2018unrealistic}
\textsc{Heidhues, P., B.~K{\H{o}}szegi, and P.~Strack} (2018):
  \enquote{Unrealistic Expectations and Misguided Learning,}
  \emph{Econometrica}, 86, 1159--1214.

\bibitem[\protect\citeauthoryear{Heidhues, K{\H{o}}szegi, and Strack}{Heidhues
  et~al.}{2019}]{heidhues2019overconfidence}
---\hspace{-.1pt}---\hspace{-.1pt}--- (2019): \enquote{Overconfidence and
  Prejudice,} \emph{arXiv preprint arXiv:1909.08497}.

\bibitem[\protect\citeauthoryear{Jehiel and Koessler}{Jehiel and
  Koessler}{2008}]{jehiel2008revisiting}
\textsc{Jehiel, P. and F.~Koessler} (2008): \enquote{Revisiting Games of
  Incomplete Information with Analogy-Based Expectations,} \emph{Games and
  Economic Behavior}, 62, 533--557.

\bibitem[\protect\citeauthoryear{Karni and Vier{\o}}{Karni and
  Vier{\o}}{2013}]{karni2013reverse}
\textsc{Karni, E. and M.-L. Vier{\o}} (2013): \enquote{"Reverse Bayesianism": A
  Choice-Based Theory Of Growing Awareness,} \emph{American Economic Review},
  103, 2790--2810.

\bibitem[\protect\citeauthoryear{Kass and Raftery}{Kass and
  Raftery}{1995}]{kass1995bayes}
\textsc{Kass, R.~E. and A.~E. Raftery} (1995): \enquote{Bayes Factors,}
  \emph{Journal of the American Statistical Association}, 90, 773--795.

\bibitem[\protect\citeauthoryear{Klibanoff, Marinacci, and Mukerji}{Klibanoff
  et~al.}{2005}]{klibanoff2005smooth}
\textsc{Klibanoff, P., M.~Marinacci, and S.~Mukerji} (2005): \enquote{A Smooth
  Model of Decision Making under Ambiguity,} \emph{Econometrica}, 73,
  1849--1892.

\bibitem[\protect\citeauthoryear{Konishi and Kitagawa}{Konishi and
  Kitagawa}{2008}]{konishi2008information}
\textsc{Konishi, S. and G.~Kitagawa} (2008): \emph{Information Criteria and
  Statistical Modeling}, Springer Science \& Business Media.

\bibitem[\protect\citeauthoryear{K{\"o}szegi}{K{\"o}szegi}{2006}]{koszegi2006ego}
\textsc{K{\"o}szegi, B.} (2006): \enquote{Ego Utility, Overconfidence, and Task
  Choice,} \emph{Journal of the European Economic Association}, 4, 673--707.

\bibitem[\protect\citeauthoryear{Kuhn}{Kuhn}{1962}]{kuhn1962}
\textsc{Kuhn, T.~S.} (1962): \emph{The Structure of Scientific Revolutions},
  University of Chicago Press.

\bibitem[\protect\citeauthoryear{Langer and Roth}{Langer and
  Roth}{1975}]{langer1975heads}
\textsc{Langer, E.~J. and J.~Roth} (1975): \enquote{Heads I Win, Tails It's
  Chance: The Illusion of Control as a Function of the Sequence of Outcomes in
  a Purely Chance Task,} \emph{Journal of Personality and Social Psychology},
  32, 951.

\bibitem[\protect\citeauthoryear{Mailath and Samuelson}{Mailath and
  Samuelson}{2020}]{mailath2019wisdom}
\textsc{Mailath, G.~J. and L.~Samuelson} (2020): \enquote{Learning under
  Diverse World Views: Model-Based Inference,} \emph{American Economic Review},
  110, 1464--1501.

\bibitem[\protect\citeauthoryear{Maitra}{Maitra}{1968}]{maitra1968discounted}
\textsc{Maitra, A.} (1968): \enquote{Discounted Dynamic Programming on Compact
  Metric Spaces,} \emph{Sankhy{\=a}: The Indian Journal of Statistics, Series
  A}, 211--216.

\bibitem[\protect\citeauthoryear{Miller and Ross}{Miller and
  Ross}{1975}]{miller1975self}
\textsc{Miller, D.~T. and M.~Ross} (1975): \enquote{Self-Serving Biases in the
  Attribution of Causality: Fact or fiction?} \emph{Psychological bulletin},
  82, 213.

\bibitem[\protect\citeauthoryear{Montiel~Olea, Ortoleva, Pai, and
  Prat}{Montiel~Olea et~al.}{2022}]{olea2019competing}
\textsc{Montiel~Olea, J.~L., P.~Ortoleva, M.~M. Pai, and A.~Prat} (2022):
  \enquote{Competing Models,} \emph{The Quarterly Journal of Economics}, 137,
  2419--2457.

\bibitem[\protect\citeauthoryear{Mullainathan}{Mullainathan}{2002}]{mullainathan2002thinking}
\textsc{Mullainathan, S.} (2002): \enquote{Thinking Through Categories,} Tech.
  rep., Working Paper, Harvard University.

\bibitem[\protect\citeauthoryear{Nyarko}{Nyarko}{1991}]{nyarko1991learning}
\textsc{Nyarko, Y.} (1991): \enquote{Learning in Mis-specified Models and the
  Possibility of Cycles,} \emph{Journal of Economic Theory}, 55, 416--427.

\bibitem[\protect\citeauthoryear{Ortoleva}{Ortoleva}{2012}]{ortoleva2012modeling}
\textsc{Ortoleva, P.} (2012): \enquote{Modeling the Change of Paradigm:
  Non-Bayesian Reactions to Unexpected News,} \emph{American Economic Review},
  102, 2410--36.

\bibitem[\protect\citeauthoryear{Ortoleva and Snowberg}{Ortoleva and
  Snowberg}{2015}]{ortoleva2015overconfidence}
\textsc{Ortoleva, P. and E.~Snowberg} (2015): \enquote{Overconfidence in
  Political Behavior,} \emph{American Economic Review}, 105, 504--35.

\bibitem[\protect\citeauthoryear{Oster, Shoulson, and Dorsey}{Oster
  et~al.}{2013}]{oster2013optimal}
\textsc{Oster, E., I.~Shoulson, and E.~Dorsey} (2013): \enquote{Optimal
  Expectations and Limited Medical Testing: evidence from Huntington disease,}
  \emph{American Economic Review}, 103, 804--30.

\bibitem[\protect\citeauthoryear{Rabin and Vayanos}{Rabin and
  Vayanos}{2010}]{rabin2010gambler}
\textsc{Rabin, M. and D.~Vayanos} (2010): \enquote{The Gambler's and Hot-Hand
  Fallacies: Theory and Applications,} \emph{Review of Economic Studies}, 77,
  730--778.

\bibitem[\protect\citeauthoryear{Savage}{Savage}{1972}]{savage1972foundations}
\textsc{Savage, L.~J.} (1972): \emph{The Foundations of Statistics}, Courier
  Corporation.

\bibitem[\protect\citeauthoryear{Schwartzstein and Sunderam}{Schwartzstein and
  Sunderam}{2021}]{schwartzstein2019using}
\textsc{Schwartzstein, J. and A.~Sunderam} (2021): \enquote{Using Models to
  Persuade,} \emph{American Economic Review}, 111, 276--323.

\bibitem[\protect\citeauthoryear{Schwarz et~al.}{Schwarz
  et~al.}{1978}]{schwarz1978estimating}
\textsc{Schwarz, G. et~al.} (1978): \enquote{Estimating the Dimension of a
  Model,} \emph{The Annals of Statistics}, 6, 461--464.

\bibitem[\protect\citeauthoryear{Spiegler}{Spiegler}{2016}]{spiegler2016bayesian}
\textsc{Spiegler, R.} (2016): \enquote{Bayesian Networks and Boundedly Rational
  Expectations,} \emph{The Quarterly Journal of Economics}, 131, 1243--1290.

\bibitem[\protect\citeauthoryear{Spiegler}{Spiegler}{2019}]{spiegler2019behavioral}
---\hspace{-.1pt}---\hspace{-.1pt}--- (2019): \enquote{Behavioral Implications
  of Causal Misperceptions,} \emph{Annual Review of Economics}, 12.

\bibitem[\protect\citeauthoryear{Spiegler}{Spiegler}{2020}]{spiegler2020can}
---\hspace{-.1pt}---\hspace{-.1pt}--- (2020): \enquote{Can Agents with Causal
  Misperceptions be Systematically Fooled?} \emph{Journal of the European
  Economic Association}, 18, 583--617.

\bibitem[\protect\citeauthoryear{Stone}{Stone}{1977}]{stone1977asymptotic}
\textsc{Stone, M.} (1977): \enquote{An Asymptotic Equivalence of Choice of
  Model by Cross-validation and Akaike's Criterion,} \emph{Journal of the Royal
  Statistical Society: Series B (Methodological)}, 39, 44--47.

\bibitem[\protect\citeauthoryear{Wegener and Petty}{Wegener and
  Petty}{1997}]{wegener1997flexible}
\textsc{Wegener, D.~T. and R.~E. Petty} (1997): \enquote{The Flexible
  Correction Model: The Role of Naive Theories of Bias in Bias Correction,} in
  \emph{Advances in experimental social psychology}, Elsevier, vol.~29,
  141--208.

\end{thebibliography}

\pagebreak{}

\appendix

\section{\label{sec:Auxiliary-Results}Auxiliary Lemmas}

\noindent The underlying probability space $\left(Y,\mathcal{F},\mathbb{P}\right)$ is constructed as follows. The sample space is $\mathscr{Y}\coloneqq\left(\mathcal{Y}^{\infty}\right)^{\mathcal{A}}$, each element of which consists of infinite sequences of outcome realizations $\left(y_{a,0},y_{a,1},...\right)$ for all actions $a\in\mathcal{A}$, where $y_{a,t}$ denotes the outcome when the agent takes $a$ in period $t$. Let us denote by $\mathbb{P}$ the probability measure over $\mathscr{Y}$ induced by independent draws from $q^{\ast}$ and denote by $\mathcal{F}$ the product sigma algebra. Let $h\coloneqq\left(a_{t},y_{t}\right)_{t=0}^{\infty}$ denote an infinite history and $H\coloneqq\left(\mathcal{A}\times\mathcal{Y}\right)^{\infty}$ be the set of infinite histories. Combined with the switching threshold $\alpha$, the switcher's problem $(\theta,\theta',\pi_0^\theta,\pi_0^{\theta'})$, and policies $(f^{\theta},f^{\theta'})$, $\mathbb{P}$ induces a probability measure over $H$ when the agent is a switcher, denoted by $\mathbb{P}_{S}$. Meanwhile, the measure $\mathbb{P}$, prior $\pi_{0}^{\theta}$, and policy $f^{\theta}$ induce a different probability measure over $H$ for a $\theta$-modeler who uses the same prior and policy, denoted by $\mathbb{P}_{D}$. All probabilistic statements about a switcher are made with respect to $\mathbb P_S$ and all those about a $\theta$-modeler are with respect to $\mathbb P_D$, unless indicated otherwise.  

\begin{lem}
\label{lem:RL-converge} Consider any switcher's problem $(\theta,\theta',\pi_0^\theta,\pi_0^{\theta'})$ in which $\theta,\theta'\in\Theta$ and $\theta'$ is correctly specified. The ratio $\ell_{t}(\theta)/\ell_{t}(\theta')$ a.s. converges to a non-negative random variable with finite expectation. 
\end{lem}
\begin{proof}
Let $\kappa_{t}=\ell_{t}(\theta)/\ell_{t}(\theta')$, then $\kappa_{0}=1$ and $\kappa_{t}\geq0,\forall t$. I now construct the probability space in which $\kappa_{t}$ is a martingale. Given prior $\pi_{0}^{\theta'}$, denote by $\mathbb{P}_{S}^{\theta'}$ the probability measure over the set of histories $H$ as implied by model $\theta'$. Formally, for any $\hat{H}\subseteq H$, we have $\mathbb{P}_{S}^{\theta'}\left(\hat H\right)=\sum_{\omega\in\Omega^{\theta'}}\pi_{0}^{\theta'}\left(\omega\right)\mathbb{P}_{S}^{\theta',\omega}\left(\hat H\right)$, where $\mathbb{P}_{S}^{\theta',\omega}$ is the probability measure over $H$ induced by the switcher if the true DGP is as described by $\theta'$ and $\omega$. Take the conditional expectation of $\kappa_t$ with respect to $\mathbb{P}_{S}^{\theta'}$, then we have 
\begin{align*}
 & \mathbb{E}^{\mathbb{P}_{S}^{\theta'}}\left(\kappa_{t}|h_{t}\right)\\
 & =\mathbb{E}^{\mathbb{P}_{S}^{\theta'}}\left[\frac{\sum_{\omega\in\Omega^{\theta}}q^{\theta}\left(y_{t-1}|a_{t-1},\omega\right)\pi_{t-1}^{\theta}\left(\omega\right)}{\sum_{\omega^{\prime}\in\Omega^{\theta'}}q^{\theta'}\left(y_{t-1}|a_{t-1},\omega^{\prime}\right)\pi_{t-1}^{\theta'}\left(\omega^{\prime}\right)}\kappa_{t-1}|h_{t}\right]\\
 & =\kappa_{t-1}\sum_{\tilde{\omega}\in\Omega^{\theta'}}\pi_{t-1}^{\theta'}\left(\tilde{\omega}\right)\left[\int_{\mathcal{Y}}\frac{\sum_{\omega\in\Omega^{\theta}}q^{\theta}\left(y_{t-1}|a_{t-1},\omega\right)\pi_{t-1}^{\theta}\left(\omega\right)}{\sum_{\omega^{\prime}\in\Omega^{\theta'}}q^{\theta'}\left(y_{t-1}|a_{t-1},\omega^{\prime}\right)\pi_{t-1}^{\theta'}\left(\omega^{\prime}\right)}q^{\theta'}\left(y_{t-1}|a_{t-1},\tilde{\omega}\right)\nu\left(dy_{t-1}\right)\right]\\
 & =\kappa_{t-1}\int_{\mathcal{Y}}\left[\frac{\sum_{\omega\in\Omega^{\theta}}q^{\theta}\left(y_{t-1}|a_{t-1},\omega\right)\pi_{t-1}^{\theta}\left(\omega\right)}{\sum_{\omega^{\prime}\in\Omega^{\theta'}}q^{\theta'}\left(y_{t-1}|a_{t-1},\omega^{\prime}\right)\pi_{t-1}^{\theta'}\left(\omega^{\prime}\right)}\left(\sum_{\tilde{\omega}\in\Omega^{\theta'}}q^{\theta'}\left(y_{t-1}|a_{t-1},\tilde{\omega}\right)\pi_{t-1}^{\theta'}\left(\tilde{\omega}\right)\right)\right]\nu\left(dy_{t-1}\right)\\
 & =\kappa_{t-1}\int_{\mathcal{Y}}\left[\sum_{\omega\in\Omega^{\theta}}q^{\theta}\left(y_{t-1}|a_{t-1},\omega\right)\pi_{t-1}^{\theta}\left(\omega\right)\right]\nu\left(dy_{t-1}\right)\\
 & =\kappa_{t-1}\sum_{\omega\in\Omega^{\theta}}\left[\int_{\mathcal{Y}}q^{\theta}\left(y_{t-1}|a_{t-1},\omega\right)\nu\left(dy_{t-1}\right)\right]\pi_{t-1}^{\theta}\left(\omega\right)=\kappa_{t-1}.
\end{align*}
Hence, $\kappa_{t}$ is a martingale w.r.t. $\mathbb{P}_{S}^{\theta'}$. Since $\kappa_{t}\geq0,\forall t$, the Martingale Convergence Theorem implies that $\kappa_{t}$ converges to $\kappa$ almost surely w.r.t. $\mathbb{P}_{S}^{\theta'}$, and $\mathbb{E}^{\mathbb{P}_{S}^{\theta'}}\kappa\leq\mathbb{E}^{\mathbb{P}_{S}^{\theta'}}\kappa_{0}=1$. Since $\theta'$ is correctly specified, there exists a parameter $\omega^{\ast}\in\Omega^{\theta'}$ such that $q^{\ast}\left(\cdot|a\right)\equiv q^{\theta'}\left(\cdot|a,\omega^{\ast}\right),\forall a\in\mathcal{A}$. It then follows from $\pi_{0}^{\theta'}\left(\omega^{\ast}\right)>0$ that $\kappa_{t}$ also converges to $\kappa$ almost surely w.r.t. $\mathbb{P}_{S}^{\theta',\omega^{\ast}}$, which is the same measure as $\mathbb{P}_{S}$. Moreover, $\mathbb{E}\kappa<\infty$ because otherwise it contradicts $\mathbb{E}^{\mathbb{P}_{S}^{\theta'}}\kappa\leq1.$
\end{proof}
\begin{lem}
\label{lem:Convergence-lambda--correct} Suppose $\theta\in\Theta$ persists against a correctly specified model $\theta'\in\Theta$ at some full-support priors $\pi_{0}^{\theta},\pi_{0}^{\theta'}$. Then on paths where $m_{t}$ eventually equals $\theta$, we have $\lambda_t\xrightarrow{\text{a.s.}} \lambda_\infty \leq \alpha$, $\pi_{t}^{\theta'}\xrightarrow{\text{a.s.}}\pi_{\infty}^{\theta'}$, and $\pi_{t}^{\theta}\xrightarrow{\text{a.s.}}\pi_{\infty}^{\theta}$.
\end{lem}
\begin{proof}
It immediately follows from Lemma \ref{lem:RL-converge} that $\ell_{t}(\theta')/\ell_{t}(\theta)\xrightarrow{\text{a.s.}}\iota\leq\alpha$ on paths where $m_{t}$ converges to $\theta$. I now show that $\pi_{t}^{\theta}$ and $\pi_{t}^{\theta'}$ also converge almost surely. Given any $\omega\in\Omega^{\theta}$, we can write
\begin{align*}
\frac{\pi_{t}^{\theta}\left(\omega\right)}{\pi_{0}^{\theta}\left(\omega\right)} & =\frac{\prod_{\tau=0}^{t-1}q^{\theta}\left(y_{\tau}|a_{\tau},\omega\right)}{\sum_{\omega^{\prime}\in\Omega^{\theta}}\prod_{\tau=0}^{t-1}q^{\theta}\left(y_{\tau}|a_{\tau},\omega^{\prime}\right)\pi_{0}^{\theta}\left(\omega^{\prime}\right)}\\
 & =\frac{\ell_{t}(\theta')}{\ell_{t}(\theta)}\cdot\frac{\prod_{\tau=0}^{t-1}q^{\theta}\left(y_{\tau}|a_{\tau},\omega\right)}{\sum_{\omega^{\prime\prime}\in\Omega^{\theta'}}\prod_{\tau=0}^{t-1}q^{\theta'}\left(y_{\tau}|a_{\tau},\omega^{\prime\prime}\right)\pi_{0}^{\theta'}\left(\omega^{\prime\prime}\right)}\\
 & \coloneqq\frac{\ell_{t}(\theta')}{\ell_{t}(\theta)}\cdot\frac{\ell_{t}(\theta,\omega)}{\ell_{t}(\theta')},
\end{align*}
where the second term $\ell_{t}(\theta,\omega)/\ell_{t}(\theta')$ can be seen as the likelihood ratio of a model that consists of a single parameter $\omega$ and the competing model $\theta'$. By Lemma \ref{lem:RL-converge}, $\ell_{t}(\theta,\omega)/\ell_{t}(\theta')$ a.s. converges to a random variable with finite expectation. Consider the paths on which $m_{t}$ converges to $\theta$. On these paths, both $\ell_{t}(\theta')/\ell_{t}(\theta)$ and $\ell_{t}(\theta,\omega)/\ell_{t}(\theta')$ converges a.s., which implies that $\pi_{t}^{\theta}\left(\omega\right)$ a.s. converges to a random variable with finite expectation as well. Since this is true for all $\omega\in\Omega^{\theta}$, $\pi_{t}^{\theta}$ a.s. converges to some limit $\pi_{\infty}^{\theta}$ on those paths. Analogously, for any $\omega^{\prime}\in\Omega^{\theta'}$, we can write
\[
\frac{\pi_{t}^{\theta'}\left(\omega^{\prime}\right)}{\pi_{0}^{\theta'}\left(\omega^{\prime}\right)}= \frac{\ell_{t}(\theta',\omega')}{\ell_{t}(\theta')},
\]
which, again by Lemma \ref{lem:RL-converge}, converges almost surely.
\end{proof}
\begin{lem}
\label{lem:kronecker-lemma} Fix any $\theta,\theta^{\prime}\in\Theta$, $\omega\in\Omega^{\theta},\omega^{\prime}\in\Omega^{\theta^{\prime}}$ and any sequence of actions $\left(a_{1},a_{2},...\right)$. For each infinite history $h\in\left(\mathcal{A}\times\mathcal{Y}\right)^{\infty}$ that is generated according to $\left(a_{1},a_{2},...\right)$ by the true DGP, let 
\[
\xi_{t}\left(h\right)=\ln\frac{q^{\theta}\left(y_{t}|a_{t},\omega\right)}{q^{\theta^{\prime}}\left(y_{t}|a_{t},\omega^{\prime}\right)}-\mathbb{E}\left(\ln\frac{q^{\theta}\left(y_{t}|a_{t},\omega\right)}{q^{\theta^{\prime}}\left(y_{t}|a_{t},\omega^{\prime}\right)}|h_{t}\right).
\]
Then for any fixed $t_{0}\geq1$, 
\[
\lim_{t\rightarrow\infty}\left(t-t_{0}+1\right)^{-1}\sum_{\tau=t_{0}}^{t}\xi_{\tau}\left(h\right)=0,\text{ a.s..}
\]
\end{lem}
\begin{proof}
Note that $\xi_{t}\left(h\right)$ is a martingale difference process since $E\left(\xi_{t}\left(h\right)|h_{t}\right)=0$. So for any $t_{0}$, $\xi_{t_{0}}^{t}\left(h\right)\coloneqq\sum_{\tau=t_{0}}^{t}\left(t-\tau+1\right)^{-1}\xi_{\tau}\left(h\right)$ is also a martingale difference process. To use the Martingale Convergence Theorem, I now show that $\sup_{t}\mathbb{E}\left(\left(\xi_{t_{0}}^{t}\right)^{2}\right)<\infty$. Notice that 
\begin{align*}
\mathbb{E}\left(\left(\xi_{t_{0}}^{t}\right)^{2}\right)= & \mathbb{E}\left[\left(\sum_{\tau=t_{0}}^{t}\left(t-\tau+1\right)^{-1}\xi_{\tau}\left(h\right)\right)^{2}\right]\\
\leq & \sum_{\tau=t_{0}}^{t}\left(t-\tau+1\right)^{-2}\mathbb{E}\left[\left(\xi_{\tau}\left(h\right)\right)^{2}\right]\\
\leq & \sum_{\tau=t_{0}}^{t}\left(t-\tau+1\right)^{-2}\mathbb{E}\left[\left(\ln\frac{q^{\theta}\left(y_{t}|a_{t},\omega\right)}{q^{\theta^{\prime}}\left(y_{t}|a_{t},\omega^{\prime}\right)}\right)^{2}\right]\\
\leq & \sum_{\tau=t_{0}}^{t}\left(t-\tau+1\right)^{-2}\mathbb{E}\left[\left(\ln\frac{q^{\ast}\left(y_{t}|a_{t}\right)}{q^{\theta}\left(y_{t}|a_{t},\omega\right)}\right)^{2}+\left(\ln\frac{q^{\ast}\left(y_{t}|a_{t}\right)}{q^{\theta^{\prime}}\left(y_{t}|a_{t},\omega^{\prime}\right)}\right)^{2}\right]\\
\leq & 2\sum_{\tau=t_{0}}^{t}\left(t-\tau+1\right)^{-2}\max_{a}\mathbb{E}\left[\left(r_{a}\left(y\right)\right)^{2}\right]<\infty,
\end{align*}
where the first inequality follows from the fact that, for any $\tau^{\prime}>\tau\geq t_{0}$, $\mathbb{E}\left(\xi_{\tau}\left(h\right)\xi_{\tau^{\prime}}\left(h\right)\right)=\mathbb{E}\left(\mathbb{E}\left(\xi_{\tau^{\prime}}\left(h\right)|h_{\tau}\right)\xi_{\tau}\left(h\right)\right)=0$ and the last inequality follows from Assumption \ref{assu:Subjective-DGP}. Now we can invoke the Martingale Convergence Theorem which implies that $\xi_{t_{0}}^{t}$ converges to a random variable $\xi_{t_{0}}^{\infty}$ almost surely with $\mathbb{E}\left(\left(\xi_{t_{0}}^{\infty}\right)^{2}\right)<\infty$. Since $\xi_{t_{0}}^{\infty}=\lim_{t\rightarrow\infty}\sum_{\tau=t_{0}}^{t}\left(t-\tau+1\right)^{-1}\xi_{\tau}\left(h\right)$ is finite a.s., it follows from the Kronecker Lemma that 
\[
\lim_{t\rightarrow\infty}\left(t-t_{0}+1\right)^{-1}\sum_{\tau=t_{0}}^{t}\xi_{\tau}\left(h\right)=0,\text{a.s.}.
\]
\end{proof}

Let us define action frequency $\sigma_{t}:\mathcal{A}^{t+1}\rightarrow\Delta\mathcal{A}$ to measure how frequent each action has been played up to period $t$. In particular, given an action sequence $\left(a_{0},a_{1},...\right)$, let
\[
\sigma_{t}\left(a\right)=\frac{\sum_{\tau=0}^{t}\boldsymbol{1}\left(a_{t}=a\right)}{t}.
\]

\begin{lem}
Fix any $\theta\in\Theta$. \label{lem:belief-when-action-freq-converges}Suppose the action frequency of a $\theta$-modeler converges to $\sigma$, then her belief ${\pi}_{t}^{\theta}$ satisfies ${\pi}^{\theta}_t\left(\Omega^{\theta}\left(\sigma\right)\right)\xrightarrow{\text{a.s.}}1$. Similarly, if the action frequency of a switcher with $\Theta^{\dagger}\ni\theta$ converges to $\sigma$, then her belief $\pi_{t}^{\theta}$ also satisfies ${\pi}^{\theta}_t\left(\Omega^{\theta}\left(\sigma\right)\right)\xrightarrow{\text{a.s.}}1$.
\end{lem}
\begin{proof}
The proof is completely identical for either a $\theta$-modeler or a switcher. Since $\Omega^{\theta}$ is finite, for any given $\sigma$, there exists $\epsilon>0$ such that 
\begin{equation}
\sum_{\mathcal{A}}\sigma\left(a\right)\left[D_{KL}\left(q^{\ast}\left(\cdot|a\right)\parallel q^{\theta}\left(\cdot|a,\omega\right)\right)-D_{KL}\left(q^{\ast}\left(\cdot|a\right)\parallel q^{\theta}\left(\cdot|a,\omega^{\prime}\right)\right)\right]<-\epsilon \label{eq:minus-eps}
\end{equation}
for all $\omega\in\Omega^{\theta}\left(\sigma\right)$ and $\omega^{\prime}\in\Omega^{\theta}/\Omega^{\theta}\left(\sigma\right)$. For any $\omega\in\Omega^{\theta}\left(\sigma\right)$ and $\omega^{\prime}\in\Omega^{\theta}/\Omega^{\theta}\left(\sigma\right)$ at time $t$,
\begin{align*}
\frac{\pi_{t}^{\theta}\left(\omega^{\prime}\right)}{\pi_{t}^{\theta}\left(\omega\right)} & =\frac{\prod_{\tau=0}^{t-1}q^{\theta}\left(y_{\tau}|a_{\tau},\omega^{\prime}\right)\pi_{0}^{\theta}\left(\omega^{\prime}\right)}{\prod_{\tau=0}^{t-1}q^{\theta}\left(y_{\tau}|a_{\tau},\omega\right)\pi_{0}^{\theta}\left(\omega\right)}\\
 & =\exp\left(\sum_{\tau=0}^{t-1}\ln\frac{q^{\theta}\left(y_{\tau}|a_{\tau},\omega^{\prime}\right)}{q^{\theta}\left(y_{\tau}|a_{\tau},\omega\right)}+\ln\frac{\pi_{0}^{\theta}\left(\omega^{\prime}\right)}{\pi_{0}^{\theta}\left(\omega\right)}\right).
\end{align*}
We are done if this ratio converges to 0. Notice that
\begin{align*}
 & \frac{1}{t}\sum_{\tau=0}^{t-1}\mathbb{E}\left(\ln\frac{q^{\theta}\left(y_{\tau}|a_{\tau},\omega^{\prime}\right)}{q^{\theta}\left(y_{\tau}|a_{\tau},\omega\right)}|h_{t}\right)\\
= & -\sum_{\mathcal{A}}\sigma_{t}\left(a\right)\left[D_{KL}\left(q^{\ast}\left(\cdot|a\right)\parallel q^{\theta}\left(\cdot|a,\omega^{\prime}\right)\right)-D_{KL}\left(q^{\ast}\left(\cdot|a\right)\parallel q^{\theta}\left(\cdot|a,\omega\right)\right)\right],
\end{align*}
which converges to the left-hand side of Eq. \eqref{eq:minus-eps} as $\sigma_{t}$ converges to $\sigma$. Hence, there exists $T_{1}$ such that 
\[
\frac{1}{t}\sum_{\tau=0}^{t-1}\mathbb{E}\left(\ln\frac{q^{\theta}\left(y_{\tau}|a_{\tau},\omega^{\prime}\right)}{q^{\theta}\left(y_{\tau}|a_{\tau},\omega\right)}|h_{t}\right)<-\frac{\epsilon}{2},\forall t>T_{1}.
\]
By Lemma \ref{lem:kronecker-lemma}, there exists $T_{2}$ such that when $t>T_{2}$, 
\begin{align*}
\frac{1}{t}\sum_{\tau=0}^{t-1}\ln\frac{q^{\theta}\left(y_{\tau}|a_{\tau},\omega^{\prime}\right)}{q^{\theta}\left(y_{\tau}|a_{\tau},\omega\right)} & <\frac{1}{t}\sum_{\tau=0}^{t-1}\mathbb{E}\left(\ln\frac{q^{\theta}\left(y_{\tau}|a_{\tau},\omega^{\prime}\right)}{q^{\theta}\left(y_{\tau}|a_{\tau},\omega\right)}|h_{t}\right)+\frac{\epsilon}{3}
\end{align*}
It follows that when $t>\max\left\{ T_{1},T_{2}\right\} $,
\[
\sum_{\tau=0}^{t-1}\ln\frac{q^{\theta}\left(y_{\tau}|a_{\tau},\omega^{\prime}\right)}{q^{\theta}\left(y_{\tau}|a_{\tau},\omega\right)}<t\cdot\left(-\frac{\epsilon}{6}\right).
\]
Hence, $\frac{\pi_{t}^{\theta}\left(\omega^{\prime}\right)}{\pi_{t}^{\theta}\left(\omega\right)}$ converges to $0$ for all $\omega\in\Omega^{\theta}\left(\sigma\right)$ and $\omega^{\prime}\in\Omega^{\theta}/\Omega^{\theta}\left(\sigma\right)$.
\end{proof}
\begin{lem}
\label{lem:Uhc}For any $\theta\in\Theta$, the optimal action correspondence $A^{\theta}:\Delta\Omega^{\theta}\rightrightarrows\mathcal{A}$ is upper hemicontinuous in both the belief $\pi$ and the agent's discount factor $\delta$.
\end{lem}
\begin{proof}
This is a standard result directly following from \citet{blackwell1965discounted}
and \citet{maitra1968discounted}.
\end{proof}
\begin{lem}
\label{lem:compact-Berk-NE}For any $\theta\in\Theta$, the set of
all Berk-Nash equilibria under $\theta$ is compact.
\end{lem}
\begin{proof}
Denote the set of all Berk-Nash equilibria under model $\theta$ as $BN^{\theta}\subseteq\Delta\mathcal{A}$. Since $\Delta\mathcal{A}$ is bounded, we only need to show that $BN^{\theta}$ is closed. Suppose $\sigma$ is the limit of some sequence $\left(\sigma_{n}\right)_{n}$ of Berk-Nash equilibria, but $\sigma$ is not a Berk-Nash equilibrium, i.e. $\sigma\not\in BN^{\theta}$. Then for every belief $\pi\in\Delta\Omega^{\theta}\left(\sigma\right)$, we have that $\sigma\not\in\Delta A_{m}^{\theta}\left(\pi\right)$. Since $\Omega^{\theta}\left(\cdot\right)$ is upper hemicontinuous, it must be that $\Omega^{\theta}\left(\sigma_{n}\right)\subseteq\Omega^{\theta}\left(\sigma\right)$ for large enough $n$. Hence, we have $\sigma\not\in\Delta A_{m}^{\theta}\left(\pi\right)$ for every belief $\pi\in\Delta\Omega^{\theta}\left(\sigma_{n}\right)$ when $n$ is large enough. However, we know that $\supp \left(\sigma\right)\subseteq\supp \left(\sigma_{n}\right)$ for large enough $n$, which implies that $\sigma_{n}\not\in\Delta A_{m}^{\theta}\left(\pi\right)$ for large $n$. This is a contradiction.
\end{proof}

\section{\label{sec:Proofs-Main-Results}Proofs of Main Results}
\addtocontents{toc}{\protect\setcounter{tocdepth}{1}}

I prove all theorems under the assumption that the agent may be non-myopic within each model but is myopic across models (see Section \ref{subsec:patient agent}). This includes the special case where the agent is myopic everywhere.

\subsection{Proof of Theorem \ref{thm:Global-Robustness}: (ii)$\Leftrightarrow$(iii)} \label{proof:thm:Global-Robustness}

I first prove Lemma \ref{lem:Correctly-full-characterization}, which implies the necessity of a p-absorbing SCE for global robustness. I then prove Lemma \ref{lem:p-absorbingness}, which is then used to show sufficiency.
\begin{proof}[Proof of Lemma \ref{lem:Correctly-full-characterization}]

By Lemma \ref{lem:Convergence-lambda--correct}, on paths where $\theta$ is eventually forever adopted, beliefs $\pi_t^\theta$ and $\pi_t^{\theta'}$ both converge almost surely. Consider any $\hat{\omega}$ such that with positive probability, $m_{t}$ eventually equals $\theta$ and $\hat{\omega}\in\text{supp\ensuremath{\left(\pi_{\infty}^{\theta}\right)}}$. Let $A^{-}\left(\hat{\omega}\right)\equiv\left\{ a\in\mathcal{A}:q^{\theta}\left(\cdot|a,\hat{\omega}\right)\not=q^{\ast}\left(\cdot|a\right)\right\} $. I now show that every action in $A^{-}\left(\hat{\omega}\right)$ is played at most finite times a.s. on the paths where $m_{t}$ converges to $\theta$ and $\hat{\omega}\in\text{supp\ensuremath{\left(\pi_{\infty}^{\theta}\right)}}$. Suppose instead that actions in $A^{-}\left(\hat{\omega}\right)$ are played infinitely often. Then there must exist some $\gamma>0$ such that $\mathbb{E}\ln\frac{q^{\ast}\left(y|a_{t}\right)}{q^{\theta}\left(y|a_{t},\hat{\omega}\right)}>\gamma$ for infinitely many $t$. Since $\theta'$ is correctly specified, there exists a parameter $\omega^{\ast}\in\Omega^{\theta'}$ such that $q^{\ast}\left(\cdot|a\right)\equiv q^{\theta'}\left(\cdot|a,\omega^{\ast}\right),\forall a\in\mathcal{A}$. Hence, $\mathbb{E}\ln\frac{q^{\theta'}\left(y|a_{t},\omega^{\ast}\right)}{q^{\theta}\left(y|a_{t},\hat{\omega}\right)}>\gamma$ for infinitely many $t$. Notice that 
\begin{align*}
\frac{\ell_{t}(\theta')}{\ell_{t}(\theta)} & =\frac{\sum_{\omega^{\prime}\in\Omega^{\theta'}}\prod_{\tau=0}^{t-1}q^{\theta'}\left(y_{\tau}|a_{\tau},\omega^{\prime}\right)\pi_{0}^{\theta'}\left(\omega^{\prime}\right)}{\sum_{\omega\in\Omega^{\theta}}\prod_{\tau=0}^{t-1}q^{\theta}\left(y_{\tau}|a_{\tau},\omega\right)\pi_{0}^{\theta}\left(\omega\right)}\\
 & >\pi_{t}^{\theta}\left(\hat{\omega}\right)\frac{\pi_{0}^{\theta'}\left(\omega^{\ast}\right)}{\pi_{0}^{\theta}\left(\hat{\omega}\right)}\frac{\prod_{\tau=0}^{t-1}q^{\theta'}\left(y_{\tau}|a_{t},\omega^{\ast}\right)}{\prod_{\tau=0}^{t-1}q^{\theta}\left(y_{\tau}|a_{\tau},\hat{\omega}\right)}\\
 & =\pi_{t}^{\theta}\left(\hat{\omega}\right)\frac{\pi_{0}^{\theta'}\left(\omega^{\ast}\right)}{\pi_{0}^{\theta}\left(\hat{\omega}\right)}\exp\left[\sum_{\tau=0}^{t-1}1_{\left\{ a_{\tau}\in A^{-}\left(\hat{\omega}\right)\right\} }\ln\frac{q^{\theta'}\left(y_{\tau}|a_{t},\omega^{\ast}\right)}{q^{\theta}\left(y_{\tau}|a_{\tau},\hat{\omega}\right)}\right],
\end{align*}
which, by Lemma \ref{lem:kronecker-lemma}, a.s. increases to infinity as $t\rightarrow\infty$, contradicting the assumption that $m_{t}$ converges to $\theta$. Therefore, on the paths where $m_{t}\text{ eventually equals }\theta$, almost surely, there exists $T$ such that $a_{t}\in\mathcal{A}\backslash\cup_{\omega^{\prime}\in\text{supp\ensuremath{\left(\pi_{\infty}^{\theta}\right)}}}A^{-}\left(\hat{\omega}\right),\forall t>T$.

Since $q^{\theta}\left(\cdot|a,\omega^{\prime}\right)\equiv q^{\ast}\left(\cdot|a\right)$ for all $\omega^{\prime}\in\supp \left(\pi_{\infty}^{\theta}\right)$ and all $a\in\mathcal{A}\backslash\cup_{\omega^{\prime}\in\text{supp\ensuremath{\left(\pi_{\infty}^{\theta}\right)}}}A^{-}\left(\omega^{\prime}\right)$, the actions that are played in the limit have no experimentation value and are myopically optimal. Therefore, any strategy that takes support on the limit actions is a self-confirming equilibrium. Fixing a particular value of $\pi_{\infty}^{\theta}$ that is a limit belief for a positive measure of histories where $m_{t}\text{ eventually equals }\theta$, there exists a set of actions $\hat A\subseteq A_{m}^{\theta}\left(\pi_{\infty}^{\theta}\right)$ such that on those histories, the agent only plays actions from this set in the limit. Since $m_{t}$ eventually converges to $\theta$, it must be true that with positive probability, a $\theta$-modeler who inherits the switcher's prior and policy from the period when the last switch happens also only plays actions from $\hat A$ in the limit with positive probability. Therefore, take any strategy $\sigma$ with $\supp \left(\sigma\right)=\hat A$, it is a p-absorbing self-confirming equilibrium under $\theta$.
\end{proof}

\begin{lem}\label{lem:p-absorbingness}
If $\sigma$ is a p-absorbing SCE, then for any $\gamma \in (0,1)$ and $\epsilon>0$, there exists a full-support prior $\pi^\theta_0$ under which, with probability higher than $\gamma$, a $\theta$-modeler only plays actions in $\supp(\sigma)$ and her belief stays within $B_\epsilon(\Delta \Omega^\theta(\sigma))$ for all periods.\footnote{For any set of finite probability distributions $Z$ over sample space $S$, I use $B_\epsilon(Z)$ to denote the set of probability distributions whose minimum distance from any element in $Z$ is smaller than $\epsilon$, i.e. $B_\epsilon(Z)=\{z\in\Delta S:\min_{z^\prime\in Z}d_P(z,z^\prime)<\epsilon\}$, where $d_P$ represents the usual Prokhorov metric over $\Delta S$.}
\end{lem}
\begin{proof}[Proof of Lemma \ref{lem:p-absorbingness}]
 
Suppose there exists a p-absorbing SCE $\sigma$ under $\theta$. Consider the learning process of a $\theta$-modeler. By definition, there exists a full-support prior $\pi_{0}^{\theta}\in\Delta\Omega^{\theta}$ such that with positive probability, she eventually only plays actions in $\supp \left(\sigma\right)$ and each element of $\supp \left(\sigma\right)$ is played infinitely often (this is without loss of generality). Denote those paths by $\tilde{H}$. Then by a similar argument as in the proof of Lemma \ref{lem:Correctly-full-characterization}, $\pi_{t}^{\theta}$ a.s. converges to a limit $\pi_{\infty}^{\theta}$ on $\tilde{H}$, with $\supp \left(\pi_{\infty}^{\theta}\right)\subseteq \Omega^\theta(\sigma)=\{\omega\in\Omega^{\theta}:q^{\ast}\left(\cdot|a\right)=q^{\theta}\left(\cdot|a,\omega\right),\forall a\in\supp \left(\sigma\right)\}$. 

This implies the existence of an integer $T>0$ such that, with positive probability, we have (1) $a_t\in\supp (\sigma),\forall t\geq T$, (2) $\pi^\theta_t$ converges to a limit $\pi^\theta_\infty$ with $\supp (\pi^\theta_\infty)\subseteq \Omega^\theta(\sigma)$. Pick any $\epsilon>0$. Since the learning processes are Markov, we can find  a new prior $\tilde \pi_0^\theta \in B_\epsilon (\Delta \Omega^\theta(\sigma))$ under which, on a positive measure of histories, a $\theta$-modeler behaves such that (1$^\prime$) $a_t\in\supp (\sigma),\forall t\geq0$, and (2$^\prime$) the posterior $\tilde\pi^\theta_t$ almost surely converges to $\pi^\theta_\infty$ and never leaves $B_\epsilon (\Delta \Omega^\theta(\sigma))$ for all $t\geq 0$.

Denote the event described by (1$^\prime$) and (2$^\prime$) by $E$ . I now show for any constant $\gamma\in(0,1)$, there exists a full-support prior $\hat \pi^\theta_0$ under which $\mathbb P_D(E)>\gamma$. Suppose for contradiction that this is not true. Denote the probability of $E$ under any full-support prior by $\gamma(\pi^\theta)$ and let $\overline \gamma \coloneqq \sup_{\pi^\theta_0 \in\text{int}(\Delta\Omega^{\theta})} \gamma(\pi^\theta)$, where $\text{int}(\Delta\Omega^{\theta})$ denotes all full-support beliefs over $\Omega^\theta$, then it follows that $\overline\gamma<1$. By definition, for any $\psi>0$, there exists some prior $\pi^{\theta,\psi}_0$ such that $\gamma (\pi^{\theta,\psi}_0)>\overline \gamma-\psi$. But under this prior, with probability $1-\gamma (\pi^{\theta,\psi}_0)$, the dogmatic modeler  eventually  either arrives at some posterior $\pi^{\theta,\psi}_t$ that either leads her to play an action outside $\supp (\sigma)$ or leaves the neighborhood $B_\epsilon (\Delta \Omega^\theta(\sigma))$. Hence, there exists an integer $T>0$ such that $$\mathbb P_D\left(\gamma(\pi^{\theta,\psi}_T)=0\right)>\gamma (\pi^{\theta,\psi}_0)-\psi>\overline \gamma-2\psi.$$ Now, consider the supremum probability that $E$ is achieved if the agent starts with a prior that is equal to one of the possible posteriors $\pi^{\theta,\psi}_T$. Since 
$$\gamma(\pi^{\theta,\psi}_0)= \mathbb{E}_{h_T\in H_T}^{\mathbb P_D} \gamma(\pi^{\theta,\psi}_T),$$ 
we have
\begin{align*}
    \sup_{h_T\in H_T} \gamma(\pi^{\theta,\psi}_T) & > \frac{\gamma(\pi^{\theta,\psi}_0)}{1-\mathbb P_D\left(\gamma(\pi^{\theta,\psi}_T)=0\right)} \\
    & >  \frac{\overline \gamma -\psi}{1-\overline\gamma+2\psi}.
\end{align*}
But notice that when $\psi$ is sufficiently small, the term $\frac{\overline \gamma -\psi}{1-\overline\gamma+2\psi}$ is strictly larger than $\overline\gamma $, contradicting the assumption that $\overline \gamma$ is the supremum of $\gamma(\pi^\theta)$ over all full-support beliefs.  

\end{proof}

\begin{proof}[Proof of Theorem \ref{thm:Global-Robustness} (iii)$\Rightarrow$(ii)]
Pick any competing model $\theta^\prime\in\Theta$ and any full-support prior $\pi^{\theta^\prime}_0\in\Delta\Omega^{\theta'}$. Let 
$S_t \coloneqq \ell_t(\theta')/\ell_t(\theta^\ast),$ then $S_t$ is a martingale with respect to both $\mathbb P_D$ and $\mathbb P_S$ by Lemma \ref{lem:RL-converge}. By the Ville's maximal inequality for supermartingales, the probability that $S_n$ is bounded above by a positive constant larger than $1$ is bounded away from $0$. In particular, for any $\eta\in(1,\alpha)$,
$$\mathbb P_D(S_t\leq \eta,\forall t\geq 0)\geq 1- \frac{\mathbb E^{\mathbb P_D} S_0}{\eta}= 1-\frac{1}{\eta}.$$ Note that this inequality holds for any model $\theta^\prime$. 

Denote by $\sigma$ a p-absorbing SCE under $\theta$. By Lemma \ref{lem:p-absorbingness}, we know that for any $\eta\in(1,\alpha)$ and $\epsilon>0$, there exist a prior $\pi^\theta_0\in B_\epsilon(\Delta\Omega^\theta(\sigma))$ such that $\mathbb P_D(E)> 1/\eta$ (the event $E$ is defined in the proof of Lemma \ref{lem:p-absorbingness}). Therefore, 
\begin{align*}
    &\mathbb P_D (E\text{ occurs and }S_t\leq \eta,\forall t\geq 0) \\
    \geq \quad&\mathbb P_D (E)+\mathbb P_D(S_t\leq \eta,\forall t\geq 0)-1>0.
\end{align*}Denote the histories where $E\text{ occurs and }S_t\leq \eta,\forall t\geq 0$ by $\hat H$. When $\epsilon$ is small enough, we have that on $\hat H$,
\begin{align*}
\lambda_t=\frac{\ell_{t}(\theta^{\prime})}{\ell_{t}(\theta)} & =\frac{\sum_{\omega^{\prime}\in\Omega^{\theta^{\prime}}}\pi_{0}^{\theta^{\prime}}\left(\omega^{\prime}\right)\prod_{\tau=0}^{t-1}q^{\theta^{\prime}}\left(y_{\tau}|a_{\tau},\omega^{\prime}\right)}{\sum_{\omega\in\Omega^{\theta}}\pi_0^{\theta}\left(\omega\right)\prod_{\tau=0}^{t-1}q^{\theta}\left(y_{\tau}|a_{\tau},\omega\right)}\\
 & < \frac{\sum_{\omega^{\prime}\in\Omega^{\theta^{\prime}}}\pi_{0}^{\theta^{\prime}}\left(\omega^{\prime}\right)\prod_{\tau=0}^{t-1}q^{\theta^{\prime}}\left(y_{\tau}|a_{\tau},\omega^{\prime}\right)}{\pi_0^{\theta}\left(\Omega^\theta(\sigma)\right) \prod_{\tau=0}^{t-1}q^{\ast}\left(y_{\tau}|a_{\tau}\right)}  \\
 & \leq\frac{\eta}{1-\epsilon}<\alpha 
\end{align*}
where the first inequality follows from the fact that $\pi_0^{\theta}$ is full-support and the second inequality follows from the definition of $\hat H$. Thus, on $\hat{H}$, the switcher never makes any switch to the competing model $\theta^{\prime}$, i.e. $m_{t}=\theta,\forall t\geq0$, and her action choices would be identical to the $\theta$-modeler. Therefore, if we endow the switcher with the same prior $\pi_0^{\theta}$, event $\hat{H}$ also occurs with positive probability under $\mathbb{P}_{S}$.
\end{proof}

\subsection{Proof of Theorem \ref{thm:Global-Robustness}: (i)$\Rightarrow$(iii)}

I show that if $\theta$ is locally robust at some prior, then it must admit a p-absorbing SCE. Construct a competing model $\theta'$ as follows. Let $\theta'$ have the identical parameter space as $\theta$, i.e. $\Omega^{\theta'}=\Omega^{\theta}$, and let its predictions be given by $q^{\theta'}\left(\cdot|a,\omega\right)=\mu q^{\theta}\left(\cdot|a,\omega\right)+(1-\mu)  q^{\ast}\left(\cdot|a\right)$, for all $a\in\mathcal{A}$ and all $\omega\in\Omega^{\theta}$, where $\mu\in(0,1)$.  
For any $\epsilon>0$, when $\mu$ is close enough to $1$, we have $\theta'\in N_{\epsilon}\left(\theta\right)$. By the definition of local robustness, there exists $\epsilon>0$ such that $\theta$ persists against $\theta'$ under some full-support priors $\pi_0^\theta$ and $\pi_0^{\theta'}=\pi_0^\theta$. Consider any $\hat\omega\in\Omega^\theta$ such that 
$$\mathbb P_S \left(m_t\text{ eventually equals }\theta\text{ and }\liminf_{t\rightarrow\infty} \pi_t^\theta(\hat\omega)>0 \right)>0.$$
Let $A^-(\hat\omega)\coloneqq \{a\in\mathcal A:q^\theta(\cdot|a,\hat\omega)\not=q^\ast(\cdot|a)\}$. Then every action in $A^-(\hat\omega)$ is played at most finite times a.s. on the path where $m_t\text{ eventually equals }\theta\text{ and }\liminf_{t\rightarrow\infty} \pi_t^\theta(\hat\omega)>0$. Suppose instead that actions in $A^{-}\left(\hat{\omega}\right)$ are played infinitely often. Then there must exist some $\gamma>0$ such that $\mathbb{E}\ln\frac{q^{\ast}\left(y|a_{t}\right)}{q^{\theta}\left(y|a_{t},\hat{\omega}\right)}>\gamma$ for infinitely many $t$. So we have 
\begin{align*}
    \mathbb{E}\ln\frac{q^{\theta'}\left(y|a_{t},\hat{\omega}\right)}{q^{\theta}\left(y|a_{t},\hat{\omega}\right)} = \mathbb{E}\ln\left(\mu+(1-\mu)\frac{q^{\ast}\left(y|a_{t}\right)}{q^{\theta}\left(y|a_{t},\hat{\omega}\right)}\right) > (1-\mu)\gamma
\end{align*} where the inequality follows from the concavity of the logarithm function. Therefore, 
\begin{align*}
\lambda_t = & \frac{\sum_{\omega\in\Omega^{\theta}}\prod_{\tau=0}^{t-1}q^{\theta'}\left(y_{\tau}|a_{\tau},\omega\right)\pi_{0}^{\theta}\left(\omega\right)}{\sum_{\omega\in\Omega^{\theta}}\prod_{\tau=0}^{t-1}q^{\theta}\left(y_{\tau}|a_{\tau},\omega\right)\pi_{0}^{\theta}\left(\omega\right)}\\
> &  \pi_{t}^{\theta}\left(\hat{\omega}\right)\frac{\pi_{0}^{\theta}\left(\hat\omega\right)}{\pi_{0}^{\theta}\left(\hat{\omega}\right)}\frac{\prod_{\tau=0}^{t-1}q^{\theta'}\left(y_{\tau}|a_{t},\hat\omega\right)}{\prod_{\tau=0}^{t-1}q^{\theta}\left(y_{\tau}|a_{\tau},\hat{\omega}\right)}\\
= & \pi_{t}^{\theta}\left(\hat{\omega}\right) \exp \left[\sum_{\tau=0}^{t-1}1_{\left\{ a_{\tau}\in A^{-}\left(\hat{\omega}\right)\right\} }\ln\frac{q^{\theta'}\left(y_{\tau}|a_{t},\hat\omega\right)}{q^{\theta}\left(y_{\tau}|a_{\tau},\hat{\omega}\right)}\right],
\end{align*}
which, by Lemma \ref{lem:kronecker-lemma}, a.s. increases to infinity when $m_{t}$ converges to $\theta$ and $\liminf_{t\rightarrow\infty} \pi_t^\theta(\hat\omega)>0$.
This implies that, letting $\hat\Omega^\theta\coloneqq\{\omega\in\Omega^\theta:\liminf_{t\rightarrow\infty}\pi_t^\theta(\hat\omega)>0\}$, on the paths where $m_{t}\text{ eventually equals }\theta$, there almost surely exists $T$ such that $a_{t}\in\mathcal{A}\backslash\cup_{\hat\omega\in\hat\Omega^\theta}A^{-}\left(\hat{\omega}\right),\forall t>T$. Since $q^{\theta}\left(\cdot|a,\hat\omega\right)$ is equal to $q^{\ast}\left(\cdot|a\right)$ for all $\hat\omega\in\hat\Omega^\theta$ and all $a\in\mathcal{A}\backslash\cup_{\hat\omega\in\hat\Omega^\theta}A^{-}\left(\hat\omega\right)$, the posterior $\pi_t^\theta$ must converge to a limit $\pi_\infty^\theta$. The rest of the arguments are identical to those in the proof of Lemma \ref{lem:Correctly-full-characterization}; it follows that $\theta$ must admit a p-absorbing SCE.
\hfill $\Box$

\subsection{Proof of Corollary \ref{corr:Suff-robustness}}
I prove Corollary \ref{corr:Suff-robustness} assuming the agent is myopic. I show below that any quasi-strict SCE satisfies a stability property stronger than p-absorbingness, which implies Corollary \ref{corr:Suff-robustness}. Next, I show that if we   strengthen quasi-strictness with uniform quasi-strictness, then Corollary \ref{corr:Suff-robustness} holds for a non-myopic agent as well.

\begin{lem}\label{lem:quasi-strictness}
Suppose $\sigma$ is a quasi-strict SCE with supporting belief $\hat \pi$, then for any $\gamma\in(0,1)$, there exists $\epsilon>0$ such that starting from any prior $\pi_0^\theta\in B_\epsilon(\hat \pi)$, the probability that the $\theta$-modeler always plays actions in $\supp(\sigma)$ for all periods is strictly larger than $\gamma$.
\end{lem}

\begin{proof}

If $\sigma$ is quasi-strict, then $\supp \left(\sigma\right)=A^\theta_m(\hat \pi)$. Since $A^{\theta}$ is upper hemicontinuous (Lemma \ref{lem:Uhc}), there exists $\tilde\epsilon>0$ small enough such that $\supp \left(\sigma\right)\supset A^{\theta}\left({\pi}\right)$ for all ${\pi}\in B_{\tilde\epsilon}\left(\hat\pi\right)$.

Suppose $a_{t}\in\supp \left(\sigma\right),\forall t\geq0$, then for every $\omega\in\Omega^{\theta}\backslash \Omega^\theta(\sigma)$,
\begin{align*}
\mathbb{E}\left[\frac{\pi_{t}^{\theta}\left(\omega\right)}{\pi_{t}^{\theta}\left(\Omega^\theta(\sigma)\right)}|h_{t}\right] & =\mathbb{E}\left[\frac{\pi_{0}^{\theta}\left(\omega\right)\prod_{\tau=0}^{t-1}q^{\theta}\left(y_{\tau}|a_{\tau},\omega\right)}{\sum_{\omega^{\prime}\in \Omega^\theta(\sigma)}\pi_{0}^{\theta}\left(\omega^{\prime}\right)\prod_{\tau=0}^{t-1}q^{\theta}\left(y_{\tau}|a_{\tau},\omega^{\prime}\right)}|h_{t}\right]\\
 & =\mathbb{E}\left[\frac{\pi_{0}^{\theta}\left(\omega\right)}{\pi_{0}^{\theta}\left(\Omega^\theta(\sigma)\right)}\frac{\prod_{\tau=0}^{t-1}q^{\theta}\left(y_{\tau}|a_{\tau},\omega\right)}{\prod_{\tau=0}^{t-1}q^{\ast}\left(y_{\tau}|a_{\tau}\right)}|h_{t}\right]\\
 & =\frac{\pi_{0}^{\theta}\left(\omega\right)\prod_{\tau=0}^{t-2}q^{\theta}\left(y_{\tau}|a_{\tau},\omega\right)}{\pi_{0}^{\theta}\left(\Omega^\theta(\sigma)\right)\prod_{\tau=0}^{t-2}q^{\ast}\left(y_{\tau}|a_{\tau}\right)}=\frac{\pi_{t-1}^{\theta}\left(\omega\right)}{\pi_{t-1}^{\theta}\left(\Omega^\theta(\sigma)\right)}
\end{align*}
Therefore, $\frac{\pi_{t}^{\theta}\left(\omega\right)}{\pi_{t}^{\theta}\left(\Omega^\theta(\sigma)\right)}$ is a non-negative supermartingale for every $\omega\in\Omega^{\theta}\backslash \Omega^\theta(\sigma)$. It follows that $\frac{\pi_{t}^{\theta}\left(\Omega^{\theta}\backslash \Omega^\theta(\sigma)\right)}{\pi_{t}^{\theta}\left(\Omega^\theta(\sigma)\right)}$ is also non-negative supermartingale. By the Ville's maximal inequality for supermartingales, for any $\eta>0$, 
\[
\mathbb{P}_{D}\left(\frac{\pi_{t}^{\theta}\left(\Omega^{\theta}\backslash \Omega^\theta(\sigma)\right)}{\pi_{t}^{\theta}\left(\Omega^\theta(\sigma)\right)}\geq\eta\text{ for some }t\right)<\frac{1}{\eta}\frac{\pi_{0}^{\theta}\left(\Omega^{\theta}\backslash \Omega^\theta(\sigma)\right)}{\pi_{0}^{\theta}\left(\Omega^\theta(\sigma)\right)}.
\]
Since $\pi_{t}^{\theta}\left(\Omega^\theta(\sigma)\right)=1-\pi_{t}^{\theta}\left(\Omega^{\theta}\backslash \Omega^\theta(\sigma)\right)$, the above inequality implies that
\[
\mathbb{P}_{D}\left(\pi_{t}^{\theta}\left(\Omega^{\theta}\backslash \Omega^\theta(\sigma)\right)\geq\frac{\eta}{1+\eta}\text{ for some }t\right)<\frac{1}{\eta}\frac{\pi_{0}^{\theta}\left(\Omega^{\theta}\backslash \Omega^\theta(\sigma)\right)}{\pi_{0}^{\theta}\left(\Omega^\theta(\sigma)\right)}.
\]
Pick some $\epsilon\in\left(0,\tilde\epsilon\right)$ and $\pi_{0}^{\theta}\in B_{\epsilon}\left(\hat\pi\right)$, then $\pi_{0}^{\theta}\left(\Omega^\theta \backslash \Omega^\theta(\sigma)\right)<\epsilon$. Notice that the ratio $\frac{\pi_{t}^{\theta}\left(\omega\right)}{\pi_{t}^{\theta}\left(\omega^{\prime}\right)}$ remain unchanged throughout all periods for any $\omega,\omega^{\prime}\in \Omega^\theta(\sigma)$. Hence, if $\pi_{t}^{\theta}\not\in B_{\tilde\epsilon}\left(\hat\pi\right)$ for some $t\geq0$, then there exists $t$ such that $\pi_{t}^{\theta}\left(\Omega^{\theta}\backslash \Omega^\theta(\sigma)\right)\geq\pi_{0}^{\theta}\left(\Omega^{\theta}\backslash \Omega^\theta(\sigma)\right)+\tilde\epsilon-\epsilon$. Using the previous inequality,
\begin{align*}
 & \mathbb{P}_{D}\left(\pi_{t}^{\theta}\not\in B_{\tilde\epsilon}\left(\hat\pi\right)\text{ for some }t\geq0\right)\\
\leq \quad & \mathbb{P}_{D}\left(\pi_{t}^{\theta}\left(\Omega^{\theta}\backslash \Omega^\theta(\sigma)\right)\geq\pi_{0}^{\theta}\left(\Omega^{\theta}\backslash \Omega^\theta(\sigma)\right)+\tilde\epsilon-\epsilon\text{ for some }t\right) \\
<\quad & \left(\frac{1}{\pi_{0}^{\theta}\left(\Omega^{\theta}\backslash \Omega^\theta(\sigma)\right)+\tilde\epsilon-\epsilon}-1\right)\frac{\pi_{0}^{\theta}\left(\Omega^{\theta}\backslash \Omega^\theta(\sigma)\right)}{\pi_{0}^{\theta}\left(\Omega^\theta(\sigma)\right)} \\
<\quad & \left(\frac{1}{\tilde\epsilon-\epsilon}-1\right)\frac{\epsilon}{1-\epsilon}
\end{align*}
which converges to $0$ as $\epsilon$ approaches $0$.
This implies that for any $\gamma\in(0,1)$ we have $\mathbb{P}_{D}\left(\pi_{t}^{\theta}\in B_{\tilde\epsilon}\left(\hat\pi\right),\forall t\geq0\right)>\gamma$ when $\epsilon$ is sufficiently small. Notice that $\pi_{t}^{\theta}\in B_{\tilde\epsilon}\left(\hat\pi\right),\forall t\geq0$ in turn implies that $a_{t}\in\supp \left(\sigma\right),\forall t\geq0$, validating our assumption. 

\end{proof}

Say that a SCE or a BN-E $\sigma$ with supporting belief $\pi$ is \textit{uniformly quasi-strict} if $\supp \left(\sigma\right)=A_{m}^{\theta}\left(\pi\right)$
for every belief $\pi\in\Delta\Omega^{\theta}\left(\sigma\right)$. The following lemma implies that given any discount factor, a uniformly quasi-strict SCE is p-absorbing. 
\begin{lem}\label{lem:uniform quasi-strictness} Suppose the $\theta$-modeler has discount factor $\delta\in(0,1)$.
Suppose $\sigma$ is a uniformly quasi-strict SCE with supporting belief $\hat \pi$, then for any $\gamma\in(0,1)$, there exists $\epsilon>0$ such that starting from any prior $\pi_0^\theta\in B_\epsilon(\hat \pi)$, the probability that the $\theta$-modeler always plays actions in $\supp(\sigma)$ for all periods is strictly larger than $\gamma$.
\end{lem}

\begin{proof}
Since $\sigma$ is uniformly quasi-strict with supporting belief $\pi$, $\supp \left(\sigma\right)$ contains all myopically optimal actions against each degenerate belief $\delta_{\omega}$ concentrated on $\omega\in\supp \left(\pi\right)$. In addition, $\supp \left(\sigma\right)$ must be optimal against $\delta_{\omega}$ for an agent who maximizes discounted utility, because the dynamic programming problem described by \eqref{eq:dynamic-programming} reduces to a static maximization problem when the belief is degenerate. This implies that $\supp \left(\sigma\right)$ is also (dynamically) optimal against $\pi$. Further, since $A^{\theta}$ is upper hemicontinuous (by Lemma \ref{lem:Uhc}), there exists $\tilde\epsilon>0$ small enough such that $\supp \left(\sigma\right)=A^{\theta}\left(\tilde{\pi}\right)$ for all $\tilde{\pi}\in B_{\tilde\epsilon}\left(\pi\right)$. The rest of the proof is identical to the proof of Lemma \ref{lem:quasi-strictness}.
\end{proof}

\subsection{Proof of Theorem \ref{thm:concentrated-prior}(i)}
\paragraph{Necessity.} Suppose $\theta$ is locally robust at some full-support prior $\pi_0^\theta$. It follows from Theorem \ref{thm:Global-Robustness} and identifiability that there exists $\hat\omega\in\Omega^\theta$ such that the degenerate belief $\delta_\omega$ supports a p-absorbing SCE under $\theta$, i.e. $C^\theta\not=\emptyset$. 

\paragraph{Sufficiency.} Suppose $ C^\theta\not=\emptyset$. Take any $\hat\omega\in C^\theta$ and any full-support prior $\pi_0^\theta$.  Consider the probability measure $\mathbb P_S^{\theta,\hat\omega}$,  i.e. the probability measure over infinite histories $H$ induced by the switcher if the true DGP is as described by $\theta$ and $\hat\omega$. By identifiability and Lemma \ref{lem:kronecker-lemma}, the posterior $\pi_t^\theta$ converges to $\delta_{\hat\omega}$ almost surely under $\mathbb P_S^{\theta,\hat\omega}$. So for any $\mu>0$, we can find a positive measure of length-$T$ histories $\hat H_T$ where the posterior for model $\theta$ enters the $\mu$-neighborhood of $\delta_{\hat\omega}$, i.e. $\pi_T^\theta\in B_\mu(\delta_{\hat\omega})$. Let $\mu$ be small enough so that the posterior $\pi_T^\theta(\hat\omega)>1/\sqrt{\alpha}$. By absolute continuity (Assumption \ref{assu:application-2}), we know $\hat H_T$ is also realized with positive probability under the true measure $\mathbb P_S$. 

Next I show that when $\mathcal Y$ is discrete, we can choose $\epsilon$ to be sufficiently small such that for any $\theta'\in N_\epsilon(\theta)$ and prior $\pi_0^{\theta'}\in N_\epsilon(\pi_0^\theta;\theta,\theta')$, the Bayes factor $\lambda_t$ never exceeds $\sqrt{\alpha}$ before period $T$.\footnote{The proof for the case of a continuous $\mathcal Y$ can be found in  Appendix \ref{sec:Online-Appendix}.} For each $\omega\in\Omega^\theta$, with a slight abuse of notation, denote the set of ``nearby'' parameters within $\theta'$ by  $N_\epsilon(\omega;\theta')\coloneqq \{\omega'\in\Omega^{\theta'}:d(Q^{\theta,\omega},Q^{\theta',\omega'})\leq \epsilon\}$. Then we have $q^{\theta'}(y|a,\omega')\leq q^{\theta}(y|a,\omega)+\epsilon$ for all $y\in\mathcal Y$, $a\in\mathcal A$, and $\omega'\in N_\epsilon(\omega;\theta')$. Let $\epsilon$ be sufficiently small such that $N_\epsilon(\omega;\theta')$ is disjoint across $\Omega^\theta$. By construction we have $\pi_0^{\theta'}(N_\epsilon(\omega;\theta'))\leq \pi_0^\theta(\omega) + \epsilon$. Hence,
\begin{align*}
    \lambda_t = \frac{\ell_t(\theta')}{\ell_t(\theta)}= & \frac{\sum_{\omega\in\Omega^\theta}\sum_{\omega'\in  N_\epsilon(\omega;\theta')}\pi_{0}^{\theta}\left(\omega'\right)\prod_{\tau=0}^{t-1}q^{\theta'}\left(y_{\tau}|a_{\tau},\omega'\right)}{\sum_{\omega\in\Omega^{\theta}}\pi_{0}^{\theta}\left(\omega\right)\prod_{\tau=0}^{t-1}q^{\theta}\left(y_{\tau}|a_{\tau},\omega\right)} \\
    \leq & \frac{\sum_{\omega\in\Omega^\theta}(\pi_{0}^{\theta}\left(\omega\right)+\epsilon)\prod_{\tau=0}^{t-1}(q^{\theta}\left(y_{\tau}|a_{\tau},\omega\right)+\epsilon)}{\sum_{\omega\in\Omega^{\theta}}\pi_{0}^{\theta}\left(\omega\right)\prod_{\tau=0}^{t-1}q^{\theta}\left(y_{\tau}|a_{\tau},\omega\right)} \\
    = & \max_{\omega\in\Omega^{\theta}}  \left(1+\frac{\epsilon}{\pi_0^\theta(\omega)}\right)\prod_{\tau=0}^{t-1}\left(1+\frac{\epsilon}{q^{\theta}\left(y_{\tau}|a_{\tau},\omega\right)}\right).
\end{align*}
We can choose $\epsilon$ to be sufficiently small so that $\lambda_t$ does not exceed $\sqrt{\alpha}>1$ for $t=0,...,T$ regardless of the action and outcome history. 

Finally, note that for any $t>T$, we can write 
\begin{align*}
    \lambda_t =  \lambda_T \frac{\sum_{\omega'\in\Omega^{\theta'}} \prod_{\tau = T}^{t-1}\pi_\tau^{\theta'}(\omega') q^{\theta'}(y_\tau|a_\tau,\omega')}{\sum_{\omega\in\Omega^{\theta}} \prod_{\tau = T}^{t-1}\pi_\tau^{\theta}(\omega) q^{\theta}(y_\tau|a_\tau,\omega)}
    \coloneqq  \lambda_T \lambda_{T,t}.
\end{align*}
Recall that on histories $\hat H_T$ we have $\pi_T^{\theta}(\hat\omega)>1/\sqrt{\alpha}$, so we can use the same arguments as in the proof of Theorem \ref{thm:concentrated-prior}(ii) to show that $\mathbb P_S (\lambda_{T,t}\leq \sqrt{\alpha}, \forall t>T)>0$. Since on $\hat H_T$ we have no switches before period $T$ and $\epsilon$ is small enough such that $\lambda_T<\sqrt{\alpha}$, we have $\mathbb P_S (\lambda_t\leq \alpha, \forall t\geq 0)\geq \mathbb P_S(\hat H_T) \cdot \mathbb P_S(\lambda_{T,t}\leq \sqrt{\alpha}, \forall t>T)>0$. \hfill$\Box$

\subsection{Proof of Theorem \ref{thm:concentrated-prior}(ii) and (iii)}

Note that part (iii) immediately follows from part (ii). I now prove part (ii) in two steps. 

\paragraph{Necessity.} Suppose $\theta$ is globally robust at prior $\pi_0^\theta$. By Theorem \ref{thm:Global-Robustness}, we know that there must exist a p-absorbing SCE under $\theta$. By identifiability, any SCE can only be supported by a pure belief, and hence $ C^\theta\not=\emptyset$. Pick any prior $\pi_0^\theta$ such that $\pi_0^\theta( C^\theta)<1/\alpha$.

Let us construct a competing model $\theta'$ such that it contains the prediction of $ C^\theta$ and the true DGP. In particular, let $\Omega^{\theta'} =  C^\theta\cup\{\omega^\ast\}$ and let predictions $q^{\theta'}$ satisfy that for all $a\in\mathcal A$,
\begin{align*}
    q^{\theta'}(\cdot|a,\omega) =  
    \begin{cases}
     q^{\theta}(\cdot|a,\omega) & \text{ if }\omega\in C^\theta,\\
     q^\ast(\cdot|a) & \text{ if }\omega=\omega^\ast.
     \end{cases}
\end{align*}
In addition, pick some $\epsilon\in(0,1)$ and let the prior $\pi_0^{\theta'}$ be such that 
\begin{align*}
    \pi_0^{\theta'}(\omega) =  \begin{cases}
     (1-\epsilon) \frac{\pi_0^{\theta}(\omega)}{\pi_0^{\theta}( C^\theta)}  & \text{ if }\omega\in C^\theta,\\
     \epsilon & \text{ if }\omega=\omega^\ast.\end{cases}
\end{align*}

Since $\theta'$ is correctly specified, by Lemma \ref{lem:Correctly-full-characterization}, on the paths where $m_t$ eventually equals $\theta$, the agent eventually only play actions in the support of a SCE almost surely, and her posterior converges to a supporting belief of the SCE, i.e. $\pi_t^\theta( C^\theta)\xrightarrow{\text{a.s.}}1$. By construction 
$$ \ell_t(\theta') = (1-\epsilon) \sum_{\omega\in C^\theta}\frac{\pi_0^{\theta}(\omega)}{\pi_0^{\theta}( C^\theta)}\ell_t(\theta,\omega)+\epsilon \ell_t(\theta^\ast),$$ so we have
\begin{align*}
    \frac{\ell_t(\theta')}{\ell_t(\theta)} & = (1-\epsilon)\frac{\pi_t^\theta( C^\theta)}{\pi_0^\theta( C^\theta)} + \epsilon\frac{\ell_t(\theta^\ast)}{\ell_t(\theta)}.
\end{align*}
Since $\theta'$ is correctly specified, by Lemma \ref{lem:Correctly-full-characterization}, on paths where $m_t$ eventually equals $\theta$, the first term almost surely converges to $(1-\epsilon)\frac{1}{\pi_0^\theta( C^\theta)}$. Since $\pi_0^\theta( C^\theta)<1/\alpha$, there exists a small enough $\epsilon$ such that $\frac{\ell_t(\theta')}{\ell_t(\theta)}>\alpha$ for sufficiently large $t$, contradicting the assumption that $m_t$ eventually equals $\theta$.

\paragraph{Sufficiency.} Suppose $ C^\theta\not=\emptyset$ and $\pi_0^\theta( C^\theta)\geq 1/\alpha$. Pick any competing model $\theta'$ and a full-support prior $\pi_0^{\theta'}$. We will show that model $\theta$ persists against $\theta'$ at the given priors.  

Define a new probability measure $\hat{\mathbb P}$ over the action and outcome histories $H$ such that for any histories $\hat H\subset H$,
$$\hat{\mathbb P}\left(\hat H\right)=\sum_{\omega\in C^{\theta}}\frac{\pi_{0}^{\theta}\left(\omega\right)}{\pi_0^\theta( C^\theta)}\mathbb{P}_{S}^{\theta,\omega}\left(\hat H\right),$$ where $\mathbb{P}_{S}^{\theta,\omega}$ is the probability measure over histories induced by the agent switcher if the true DGP is identical to the DGP prescribed by $\theta$ and $\omega$. Define the following process,
$$\hat \lambda_t\coloneqq \frac{1}{\pi_0^\theta( C^\theta)}\frac{\ell_t(\theta')}{\sum_{\omega\in C^{\theta}}\frac{\pi_{0}^{\theta}\left(\omega\right)}{\pi_0^\theta( C^\theta)}\ell_t(\theta,\omega)}.$$ Then it is a martingale w.r.t. $\hat{\mathbb P}$
with $\mathbb E^{\hat{\mathbb P}} (\hat\lambda_0)=1/\pi_0^\theta( C^\theta)$.
By definition, $\hat\lambda_t \geq \lambda_t$, where the equality holds only if $\Omega^\theta =  C^\theta$. By Ville's maximum inequality, 
\begin{align*}
    \hat{\mathbb P}(\lambda_t \leq \alpha,\forall t)\geq \hat{\mathbb P}(\hat\lambda_t \leq \alpha,\forall t)> 1- \frac{1}{\pi_0^\theta( C^\theta)\alpha}\geq 0.
\end{align*}
This then implies that there exists $\hat\omega\in C^\theta$ such that 
\begin{align*}
    {\mathbb P^{\theta,\hat\omega}_S}(\lambda_t \leq \alpha,\forall t)> 0.
\end{align*}
Since $\theta$ has no traps, it is identifiable and all of its p-absorbing SCEs are quasi-strict. Identifiability implies that $\mathbb P^{\theta,\hat\omega}_S(\pi_t^\theta(\hat\omega))\xrightarrow{\text{a.s.}}1$. With quasi-strictness, by Lemma \ref{lem:quasi-strictness}, there exists $\epsilon>0$ such that the optimal actions must be in the support of a SCE when $\pi_t^\theta(\hat\omega)>1-\epsilon$. Taken together, the no-trap conditions imply that there exists $T>0$ such that with positive probability (measured by $\mathbb P^{\theta,\hat\omega}_S$), the agent plays only SCE actions after period $T$ and never switches. Denote the set of such histories by $\hat H$. Moreover, for any $\hat h\in\hat H$, denote the observable history for the first $T$ periods by $\hat h_{T-}$ and the history after the first $T$ periods by $\hat h_{T+}$. Since $T$ is finite, by absolute continuity (Assumption \ref{assu:Subjective-DGP}), for any $\hat h\in\hat H$, the history $\hat h_{T-}$ also occurs with positive probability under the true measure $\mathbb P_S$. Conditional on $\hat h_{T-}$, since the agent plays only SCE actions on $\hat H$ after the first $T$ periods, the two probability measures $\mathbb P^{\theta,\hat{\omega}}_S$ and $\mathbb P_S$ over $\hat H$ are identical to each other. Therefore,
\begin{align*}
    \mathbb P_S(\hat H_T) &= \sum_{\hat h\in\hat H_T}\mathbb P_S(\hat h_{T-}) \mathbb P_S(\hat h_{T+}|\hat h_{T-}) \\
    &= \sum_{\hat h\in\hat H_T}\mathbb P_S(\hat h_{T-}) \mathbb P^{\theta,\hat{\omega}}_S(\hat h_{T+}|\hat h_{T-}) \\
     & \geq  \min_{\tilde h\in\hat H_T}\frac{\mathbb P_S(\tilde h_{T-})}{\mathbb P^{\theta,\hat{\omega}}_S(\tilde h_{T-})} \mathbb P^{\theta,\hat{\omega}}_S(\hat H_T)>0.
\end{align*}
This means that with positive probability (under the true probability measure $\mathbb P_S$), the agent never switches to $\theta'$. Therefore, model $\theta$ persists against $\theta'$.

\subsection{Proof of Corollary \ref{cor:alpha equals 1}}
Note that the proof of Theorem \ref{thm:concentrated-prior} does not use the assumption that $\alpha>1$. Therefore, Corollary \ref{cor:alpha equals 1} is immediately implied by Theorem \ref{thm:concentrated-prior}.

\subsection{Proof of Theorem \ref{thm:Local-lobustness-w-convergence}}\label{subsec:Proof-thm-local-robustness-w-convergence}
I show a more generalized version of Theorem \ref{thm:Local-lobustness-w-convergence}: \smallskip

\emph{Suppose that the action frequency of a $\theta$-modeler almost surely converges under all full-support priors. If model $\theta\in\Theta^{\bar\theta}$ is $\bar\theta$-constrained locally robust, then it must admit a p-absorbing BN-E $\sigma$ at which $\theta$ is locally KL-minimizing within the $\bar\theta$-family.}\smallskip

Suppose no Berk-Nash equilibrium $\sigma$ under $\theta$ is locally KL-minimizing at $\sigma$ within the $\bar\theta$-family, but $\theta$ is $\bar\theta$-constrained locally robust within the $\epsilon$-neighborhood for some $\epsilon>0$. Define a function $K^{\bar\theta}:\Delta\mathcal{A}\times\Omega^{\bar\theta}\rightarrow\mathbb{R}$, where
\begin{equation}
K^{\bar\theta}\left(\sigma,\omega\right)\coloneqq\sum_{\mathcal{A}}\sigma\left(a\right)D_{KL}\left(q^{\ast}\left(\cdot|a\right)\parallel q\left(\cdot|a,\omega\right)\right).
\end{equation}
That is, $K^{\bar\theta}\left(\sigma,\omega\right)$ is the $\sigma$-weighted KL divergence between the prediction of $\omega$ and the true DGP.
Take any Berk-Nash equilibrium $\sigma\in\Delta\mathcal{A}$ under $\theta$. By assumption, there must exist some parameter $\omega^{\prime}\in\Omega^{\bar\theta}$ such that $\min_{\omega\in\Omega^{\theta}}\left\Vert \omega-\omega^{\prime}\right\Vert \leq\epsilon$ and 
\begin{equation}
\min_{\omega\in\Omega^{\theta}}K^{\bar\theta}\left(\sigma,\omega\right)>K^{\bar\theta}\left(\sigma,\omega^{\prime}\right).\label{eq:BNE-finite-cover}
\end{equation}
By continuity, there exists some open neighborhood of $\sigma$, denoted as $O_{\sigma}$, in which $\omega^{\prime}$ yields a strictly lower KL divergence than $\Omega^{\theta}$, i.e. $\forall\sigma^{\prime}\in O_{\sigma},$ we have \[ \min_{\omega\in\Omega^{\theta}}K^{\bar\theta}\left(\sigma^{\prime},\omega\right)>K^{\bar\theta}\left(\sigma^{\prime},\omega^{\prime}\right). \] We know from Lemma \ref{lem:compact-Berk-NE} that the set of Berk-Nash equilibria under $\theta$ is compact. Therefore, by the Heine-Borel theorem, there must exist a finite number of $\epsilon$-close parameters, collected by a set $R_{\epsilon}\subset B_\epsilon(\Omega^\theta)$, such that for any Berk-Nash equilibrium $\sigma$, we can find some parameter from the set $R_{\epsilon}$ such that the above inequality \eqref{eq:BNE-finite-cover} holds.

Consider a competing model $\theta'$ with an expanded parameter space $\Omega^{\theta'}=\Omega^{\theta}\cup R_{\epsilon}$, and some prior $\pi_{0}^{\theta'}$ that allocates a total probability of $\epsilon$ evenly to $R_{\epsilon}$. Formally, let 
\begin{align*}
\pi_{0}^{\theta'}\left(\omega\right) & =\left(1-\epsilon\right)\pi_{0}^{\theta}\left(\omega\right),\forall\omega\in\Omega^{\theta},\\
\pi_{0}^{\theta'}\left(\omega\right) & =\frac{\epsilon}{\left|R_{\epsilon}\right|},\forall\omega\in R_{\epsilon}.
\end{align*}
Consider all possible histories in which the switcher eventually adopts $\theta$. It must be that $\lim\sup_{t}\ell_{t\rightarrow\infty}(\theta')/\ell_{t}(\theta)\leq\alpha$ on those paths. Note that the switcher's action frequency a.s. converges to a Berk-Nash equilibrium by assumption. Consider the paths where this limit equilibrium is $\sigma$. Then we can find some $T>0$ and $\eta>0$ such that $\forall t>T$ , there exists $\omega^{\prime\prime}\in R_{\epsilon}$ such that $K^{\bar\theta}\left(\sigma_{t},\omega^{\prime\prime}\right)-K^{\bar\theta}\left(\sigma_{t},\omega\right)<-\eta,\forall\omega\in\Omega^{\theta}$. It then follows that
\begin{align*}
\lambda_{t}^{\theta'} & =\frac{\sum_{\omega^{\prime}\in\Omega^{\theta'}}\pi_{0}^{\theta'}\left(\omega^{\prime}\right)\prod_{\tau=0}^{t-1}q^{\theta'}\left(y_{\tau}|a_{\tau},\omega^{\prime}\right)}{\sum_{\omega\in\Omega^{\theta}}\pi_{0}^{\theta}\left(\omega\right)\prod_{\tau=0}^{t-1}q^{\theta}\left(y_{\tau}|a_{\tau},\omega\right)}\\
 & >\frac{\frac{\epsilon}{\left|R_{\epsilon}\right|}\prod_{\tau=0}^{t-1}q^{\theta'}\left(y_{\tau}|a_{\tau},\omega^{\prime\prime}\right)}{\sum_{\omega\in\Omega^{\theta}}\pi_{0}^{\theta}\left(\omega\right)\prod_{\tau=0}^{t-1}q^{\theta}\left(y_{\tau}|a_{\tau},\omega\right)}\\
 & =\frac{\epsilon}{\left|R_{\epsilon}\right|}\frac{1}{\sum_{\omega\in\Omega^{\theta}}\pi_{0}^{\theta}\left(\omega\right)\exp\left(t\left(K^{\bar\theta}\left(\sigma_{t},\omega^{\prime\prime}\right)-K^{\bar\theta}\left(\sigma_{t},\omega\right)\right)\right)}\\
 & >\frac{\epsilon}{\left|R_{\epsilon}\right|}\exp\left(t\eta\right)
\end{align*}
Therefore, for any $\alpha>0$, almost surely, $\ell_{t}^{\theta'}/\ell_{t}(\theta)$ exceeds $\alpha$ for infinitely many $t$, contradicting the assumption that $\lim\sup_{t\rightarrow\infty}\ell_{t}^{\theta'}/\ell_{t}(\theta)\leq\alpha$ on those paths. Therefore, $\theta$ does not persist against $\theta'$. Since the choice of $\epsilon$ is arbitrary, this implies that $\theta$ is not ${\bar\theta}$-constrained locally robust. \hfill $\Box$

\subsection{A Necessary Condition for Constrained Local Robustness}
The following theorem states a necessary condition for constrained local robustness in a special environment with binary actions. 
\begin{thm}
\label{thm:Local-dim2-necessary}Suppose $\left|\mathcal{A}\right|=2$. Then a model $\theta\in\Theta^{\bar\theta}$ is $\bar\theta$-constrained locally robust only if it admits a BN-E $\sigma$ at which $\theta$ is locally KL-minimizing within the $\bar\theta$-family. 
\end{thm}
The critical step in proving Theorem \ref{thm:Local-dim2-necessary} is to show that a dogmatic modeler's action frequency almost surely enters an arbitrarily small neighborhood of the set of Berk-Nash equilibria infinitely often. From here, we can use an analogous argument to the proof of Theorem \ref{thm:Local-lobustness-w-convergence}. I use Figure \ref{fig:example-dim2} to illustrate this first step. 
\begin{figure}
\begin{centering}
\includegraphics[width=0.4\textwidth]{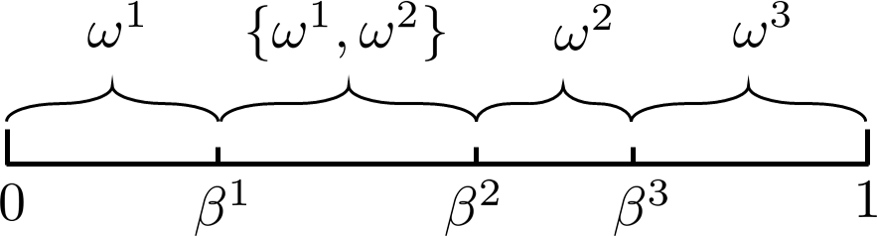}
\par\end{centering}
\caption{\label{fig:example-dim2}Example of a binary-action setting: Each point in the interval represents a mixed action's weight assigned to $a^{2}$; the parameter(s) placed above a segment of the interval are the minimizer(s) of $K^{\bar\theta}\left(\sigma,\omega\right)$ in $\Omega^{\theta}$ for all $\sigma$ in this segment.}
\end{figure}

When the action space is binary, we can write any mixed action as $\beta\cdot a^{1}+\left(1-\beta\right)\cdot a^{2}$, where $\beta\in\left[0,1\right]$. Therefore, the strategy space can be represented as the unit interval denoting the set of possible weights on $a^{2}$. To add more structure, suppose that the parameter space $\Omega^{\theta}$ has four elements, each of which is a KL-minimizer in $\Omega^{\theta}$ at some mixed strategies. Since the KL divergence is continuous in the probability of each action, it is straightforward to show that the set of mixed strategies at which a parameter is a KL-minimizer is compact and connected. For example, in Figure \ref{fig:example-dim2}, $\omega^{1}$ uniquely minimizes the KL divergence when evaluated at a mixed action when $\beta\in\left[0,\beta^{1}\right]$, while both $\omega^{1}$ and $\omega^{2}$ are minimizers when $\beta\in\left[\beta^{1},\beta^{2}\right]$. Restrict attention to the set of paths where the sequence of the action frequency $\left\{ \sigma_{t}\right\} _{t}$ is such that both $\omega^{1}$ and $\omega^{2}$ are KL-minimizers infinitely often but not $\omega^{3}$. Since the action space is binary, if $\sigma_{t}$ enters two non-connected regions on the unit interval infinitely often, it must also cross the region in between infinitely often.\footnote{This does not hold when $\left|\mathcal{A}\right|\geq3$ because there can be multiple paths connecting any two mixed actions. In fact, Example 2 in \citet{esponda2021asymptotic} describes a setting with $\left|\mathcal{A}\right|=3$, in which the dogmatic modeler's action frequency almost surely oscillates around the unique Berk-Nash equilibrium but remains bounded away from it.} This implies that $\sigma_{t}$ must enter $\left[\beta^{1},\beta^{2}\right]$ infinitely often. To generate this pattern, it must be that $a^{2}\in A_{m}^{\theta}\left(\delta_{\omega^{1}}\right)$ and $a^{1}\in A_{m}^{\theta}\left(\delta_{\omega^{2}}\right)$, because otherwise only one action will be played in the limit. Thus, there exists a mixed belief over $\omega^{1}$ and $\omega^{2}$ that makes the myopic agent indifferent between the actions. Since both $\omega^{1}$ and $\omega^{2}$ are KL-minimizers when $\beta\in\left[\beta_{1},\beta_{2}\right]$, any mixed action with $\beta\in\left[\beta_{1},\beta_{2}\right]$ is a BN-E, supported by the aforementioned mixed belief. Therefore, the agent's action frequency is almost surely arbitrarily close to the set of Berk-Nash equilibria infinitely often. The argument for other cases is analogous.

\begin{proof}

We only need to show that given any $\epsilon>0$, almost surely, a dogmatic modeler's action frequency $\sigma_{t}$ enters the $\epsilon$-neighborhood of some Berk-Nash equilibrium infinitely often from every full-support prior and policy. Then using a similar argument as in the proof of Theorem \ref{thm:Local-lobustness-w-convergence}, it can be shown that $\theta$ is not $\bar\theta$-constrained locally robust if there is no Berk-Nash equilibrium $\sigma$ such that $\theta$ is locally KL-minimizing at $\sigma$.

For convenience, let $\mathcal{A}=\left\{ a^{1},a^{2}\right\} $. First, consider the paths where $\sigma_{t}$ converges to some limit $\sigma$. denoted by $H^{1}$. Then Lemma \ref{lem:belief-when-action-freq-converges} tells us that $\pi_{t}^{\theta}\left(\Omega^{\theta}\left(\sigma\right)\right)$ converges to 1. Therefore, any action $a\not\in\cup_{\pi\in\Delta\Omega^{\theta}\left(\sigma\right)}A_{m}^{\theta}\left(\pi\right)$ cannot be in the support of $\sigma$. Hence, for each action $a$ in the support of $\sigma$, there exists some belief $\pi_{a}\in\Delta\Omega^{\theta}\left(\sigma\right)$ such that $a\in A_{m}^{\theta}\left(\pi_{a}\right)$. If $\supp \left(\sigma\right)$ is a singleton, then this immediately implies that $\sigma$ is a Berk-Nash equilibrium. If instead $\supp \left(\sigma\right)=\left\{ a^{1},a^{2}\right\} $, then by the hemi-continuity of $A_{m}^{\theta}$, there must exist some $\pi_{\sigma}\in\Delta\Omega^{\theta}\left(\sigma\right)$ such that $\left\{ a^{1},a^{2}\right\} =A_{m}^{\theta}\left(\pi_{a}\right)$, which again implies that $\sigma$ is a Berk-Nash equilibrium. Therefore, her action frequency $\sigma_{t}$ enters the $\epsilon$-neighborhood of some Berk-Nash equilibrium infinitely often for any $\epsilon>0$ almost surely on $H^{1}$.

Now consider paths where her action frequency oscillates forever, denoted by $H^{2}$. Let $\Omega_{\infty}^{\theta}$ be the set of all parameters in $\Omega^{\theta}$ that are KL-minimizers infinitely often, i.e. $\Omega_{\infty}^{\theta}=\left\{ \omega\in\Omega^{\theta}:\omega\in\Omega^{\theta}\left(\sigma_{t}\right)\text{ for infinitely many }t\text{ on }H^{2}\right\} $. Take any $\omega\in\Omega_{\infty}^{\theta}$. Suppose that $A_{m}^{\theta}\left(\delta_{\omega}\right)=\left\{ a^{1},a^{2}\right\} $, then each action frequency $\sigma_{\omega}$ that satisfies $\omega\in\Omega^{\theta}\left(\sigma_{\omega}\right)$ is a Berk-Nash equilibrium. By construction, this means $\sigma_{t}$ constitutes a Berk-Nash equilibrium infinitely often. 

Suppose instead that $\forall\omega\in\Omega_{\infty}^{\theta}$, we have that $A_{m}^{\theta}\left(\delta_{\omega}\right)$ is singleton. Since $\sigma_{t}$ oscillates, $\Omega_{\infty}^{\theta}$ cannot be a singleton. It must be that $A_{m}^{\theta}\left(\delta_{\omega}\right)=\left\{ a^{1}\right\} $ for some $\omega\in\Omega_{\infty}^{\theta}$ and or $A_{m}^{\theta}\left(\delta_{\omega^{\prime}}\right)=\left\{ a^{2}\right\} $ for some other $\omega^{\prime}\in\Omega_{\infty}^{\theta}$. Given any $\omega$ and $\omega^{\prime}\in\Omega_{\infty}^{\theta}$, say they are \textit{related} if there exists some mixed action $\sigma$ such that $\omega,\omega^{\prime}\in\Omega^{\theta}\left(\sigma\right)$. I now show that there must exist such a pair of related parameters such that $A_{m}^{\theta}\left(\delta_{\omega}\right)=\left\{ a^{1}\right\} $ and $A_{m}^{\theta}\left(\delta_{\omega^{\prime}}\right)=\left\{ a^{2}\right\} $.

First of all, every parameter in $\Omega_{\infty}^{\theta}$ must be related to some other parameter in $\Omega_{\infty}^{\theta}$. Suppose not for the sake of a contradiction. Then there exists some ``isolated'' parameter $\omega^{\ast}\in\Omega_{\infty}^{\theta}$ in the following sense: let $C_{\omega}=\left\{ \beta\in\left[0,1\right]:\omega\in\Omega^{\theta}\left(\beta a^{1}+(1-\beta)a^{2}\right)\right\} $, then there exists some positive constant $\gamma$ such that $B_{\gamma}\left(C_{\omega^{\ast}}\right)\cap\left(\cup_{\omega\in\Omega_{\infty}^{\theta}\backslash\left\{ \omega^{\ast}\right\} }C_{\omega}\right)=\emptyset$. However, since $\omega^{\ast}$ is a KL-minimizer infinitely often, it happens infinitely often that $\sigma_{t}\in C_{\omega^{\ast}}$. It implies that some KL-minimizer at $\sigma\in B_{\gamma}\left(C_{\omega^{\ast}}\right)\backslash C_{\omega^{\ast}}$ should also be a KL-minimizer at $\sigma_{t}$ infinitely often yet not included by $\Omega_{\infty}^{\theta}$, contradicting the definition of $\Omega_{\infty}^{\theta}$. By the same logic, there cannot be two cliques in $\Omega_{\infty}^{\theta}$ such that every parameter in the first clique is unrelated to every parameter in the second clique.

Hence, if every pair of related parameters in $\Omega_{\infty}^{\theta}$ induce the same optimal action, then $A_{m}^{\theta}\left(\delta_{\omega}\right)=\left\{ a^{1}\right\} $ or $\left\{ a^{2}\right\} $ for all $\omega\in\Omega_{\infty}^{\theta}$, which we know is not true. Therefore, there exists a related pair $\omega,\omega^{\prime}\in\Omega_{\infty}^{\theta}$ such that $A_{m}^{\theta}\left(\delta_{\omega}\right)=\left\{ a^{1}\right\} $ and $A_{m}^{\theta}\left(\delta_{\omega^{\prime}}\right)=\left\{ a^{2}\right\} $. Therefore, each mixed action in $C_{\omega}\cap C_{\omega^{\prime}}$ is a Berk-Nash equilibrium. Notice that each $C_{\omega}$ is compact and convex. Since $\sigma_{t}$ enters both $C_{\omega}$ and $C_{\omega^{\prime}}$ infinitely many times, it must be that $\sigma_{t}$ enters the $\epsilon$-neighborhood of $C_{\omega}\cap C_{\omega^{\prime}}$ infinitely often for any $\epsilon>0$. The proof is now complete.
\end{proof}

\subsection{\label{subsec:Constrained-LR-proof}Proof of Theorem \ref{thm:Local-robustness-sufficient}}

\paragraph{Part (i).}Suppose $\sigma$ is a pure p-absorbing BN-E with $\theta$ being locally dominant at $\sigma$ within the $\bar\theta$-family, and $\sigma$ assigns probability 1 to $\hat a\in\mathcal{A}$. Then there exists a full-support prior $\pi_{0}^{\theta}$ such that a $\theta$-modeler eventually only plays $\hat a$ with positive probability. It follows from Lemma \ref{lem:belief-when-action-freq-converges} that $\pi_{t}^{\theta}\left(\Omega^{\theta}\left(\sigma\right)\right)\stackrel{a.s.}{\rightarrow}1$. For any $\zeta,\gamma\in(0,1)$, using a similar argument as in Lemma \ref{lem:p-absorbingness}, we know that there exists a full-support prior ${\pi}_{0}^{\theta}$ with ${\pi}_{0}^{\theta}\left(\Omega^{\theta}\left(\sigma\right)\right)>\gamma$ such that with probability higher than $\zeta$, a $\theta$-modeler plays $a_{t}=\hat a$  for all $t\geq0$. 


Since $\Omega^\theta$ is finite, there exists $\epsilon_1>0$ such that whenever $\epsilon<\epsilon_1$, the $\epsilon$-neighborhoods of any $\omega\in\Omega^\theta$ and $\tilde\omega\in\Omega^\theta$ do not coincide. It follows from Lemma 2 in \cite{frick2021belief} that the local dominant condition implies that all KL-minimizers in $\Omega^\theta(\sigma)$ correspond to the same distribution under $\sigma$. Moreover, there exists $\epsilon_2>0$ s.t. whenever $\epsilon<\epsilon_2$, we have 
$$\mathbb E \left(\frac{q(\cdot|a,\omega^\prime)}{q(\cdot|a,\omega)}\right)^{d}< 1$$
for all $\omega\in\Omega^\theta(\sigma)$, all $\omega^\prime \in\Omega^{\bar\theta}\cap B_\epsilon(\Omega^\theta(\sigma))\setminus \Omega^\theta(\sigma)$, and all $a\in\supp(\sigma)$.

Let $\epsilon<\min\{\epsilon_1,\epsilon_2\}$. Pick any competing model $\theta'\in N^{\overline\theta}_\epsilon(\theta)$ with belief $\pi^{\theta'}_0\in N^{\overline\theta} _\epsilon(\pi^\theta_0)$. Then there exists a correspondence $\iota:\Omega^\theta\rightrightarrows\Omega^{\theta'}$ s.t. $\iota(\omega)=\Omega^{\theta'}\cap B_\epsilon(\omega)\setminus\Omega^\theta(\sigma)$ and $\cup_{\Omega^\theta} \iota(\omega) \equiv \Omega^{\theta'}\setminus\Omega^\theta(\sigma)$. Pick any $\hat\omega\in\Omega^\theta (\sigma)$. Since $q$ is uniformly continuous in $\omega$, there exists $\epsilon_3\leq\epsilon_2$ such that if $\epsilon<\epsilon_3$, then at all actions $a\in\supp(\sigma)$, we have
$$\mathbb E \left(\frac{\sum_{\omega^\prime\in \iota(\omega)}f(\omega^\prime)q(\cdot|a,\omega^\prime)}{q(\cdot|a,\hat\omega)}\right)^{d}< 1$$ 
for all $\omega\in\Omega^\theta\setminus\Omega^\theta(\sigma)$ and all probability distributions $f$ over $\iota(\omega)$, and
 $$\mathbb E \left(\frac{\sum_{\omega^\prime\in \cup_{\Omega^\theta(\sigma)}\iota(\omega)}f(\omega^\prime)q(\cdot|a,\omega^\prime)}{q(\cdot|a,\hat\omega)}\right)^{d}< 1$$ 
for all probability distributions $f$ over $\cup_{\Omega^\theta(\sigma)}\iota(\omega)$.

Denote the likelihood of outcomes under the predictions of $q$ and $\omega$ by $\ell_t(q,\omega)$, i.e. $\ell_t(q,\omega)\coloneqq \prod_{\tau=0}^{t-1} q(y_\tau|a_\tau,\omega)$. Let $\xi_t$ be an indicator function  such that $\xi_t=1$ if $a_t=\hat a$ and $\xi_t=0$ otherwise. Note that $\xi_t$ is adapted to the filtration of $H_{t}$. Define a data-generating process $\hat q$ as follows,
$$\hat q (\cdot|a_t,\hat\omega)\coloneqq \xi_t q (\cdot|a_t,\hat\omega)+(1-\xi_t) q^\ast (\cdot|a_t).$$ For any $\omega\in\Omega^\theta\setminus\Omega^\theta(\sigma)$ and any distribution $f$ over $\iota(\omega)$, we have
\begin{align*}
    & \mathbb E^{\mathbb P_D} \left( \left(\frac{\sum_{\omega^\prime\in \iota(\omega)}f(\omega^\prime)\ell_t(q,\omega^\prime)}{\ell_t(\hat q,\hat\omega)}\right)^{d}|h_{t} \right) \\
    = &
    \left(\frac{\sum_{\omega^\prime\in \iota(\omega)}f(\omega^\prime)\ell_{t-1}(q,\omega^\prime)}{\ell_{t-1}(\hat q,\hat\omega)}\right)^{d} \mathbb E^{\mathbb P_D} \left( \left(\frac{\sum_{\omega^\prime\in \iota(\omega)}\hat f(\omega^\prime)q(y_t|a_t,\omega^\prime)}{\hat q(y_t|a_t,\hat\omega)}\right)^{d}|h_{t} \right)\\
    < & \left(\frac{\sum_{\omega^\prime\in \iota(\omega)}f(\omega^\prime)\ell_{t-1}(q,\omega^\prime)}{\ell_{t-1}(\hat q,\hat\omega)}\right)^{d},
\end{align*}
where $\hat f(\omega^\prime) = f(\omega^\prime) \ell_{t-1}(q,\omega^\prime)/\left( \sum_{\omega^{\prime\prime}\in \iota(\omega)}f(\omega^{\prime\prime})\ell_{t-1}(q,\omega'')\right)$. Hence, by Ville's maximal inequality for supermartingales, we know that for any $\kappa>1$, 
\begin{align*}
    \mathbb P_D \left(\frac{\sum_{\omega^\prime\in \iota(\omega)}f(\omega^\prime)\ell_t(q,\omega^\prime)}{\ell_t(\hat q,\hat\omega)}\leq\kappa,\forall t \right)\geq 1-\frac{1}{\kappa^{d}}.
\end{align*}
Similarly, for all probability distributions $f$ over $\cup_{\Omega^\theta(\sigma)}\iota(\omega)$ and any $\eta\in (1,\alpha)$, we have that 
\begin{align*}
    \mathbb P_D \left(\frac{\sum_{\omega^\prime\in \cup_{\Omega^\theta(\sigma)}\iota(\omega)} f(\omega^\prime)\ell_t(q,\omega^\prime)}{\ell_t(\hat q,\hat\omega)}\leq \eta,\forall t \right)\geq 1-\frac{1}{\eta^{d}}.
\end{align*}
Let $M\coloneqq |\Omega^\theta\setminus\Omega^\theta(\sigma)|$. When $\kappa$ is sufficiently large, we have
\begin{align} 
    & \mathbb P_D\left(\frac{\sum_{\omega'\in\Omega^{\theta'}} \pi^{\theta'}(\omega') \ell_t(q,\omega')}{\ell_t(\hat q,\hat\omega)}\leq \pi^{\theta'}(\Omega^\theta(\sigma))+\pi^{\theta'}(\iota(\Omega^\theta(\sigma)))\eta +\pi^{\theta'}(\iota(\Omega^\theta\setminus\Omega^\theta(\sigma))) \kappa ,\forall t\geq 0\right) \nonumber \\ 
    \geq & 1-\frac{1}{\eta^{d}} + M\cdot \left(1-\frac{1}{\kappa^{d}}\right)-M =1-\frac{1}{\eta^{d}}-\frac{M}{\kappa^{d}}>0\label{eq:thm4-lkr}
\end{align} where the first inequality follows from the inequality $P(A_1 ... A_n)\geq P(A_1)+...+P(A_n)-(n-1)$ for random events $A_1,...,A_n$.
Note that when $\pi^{\theta'}\in N^{\overline \theta}_\epsilon (\pi^\theta)$ and $\pi^\theta(\Omega^\theta(\sigma))>\gamma$, we have
$$\pi^{\theta'}(\Omega^\theta(\sigma))+\pi^{\theta'}(\iota(\Omega^\theta(\sigma)))\eta +\pi^{\theta'}(\iota(\Omega^\theta\setminus\Omega^\theta(\sigma))) \kappa\leq \eta +(1-\gamma+\epsilon) \kappa.$$
Using the observation we derived in the first paragraph, we know that there exists a full-support prior $\pi^{\theta}_0$ under which with probability larger than $\frac{1}{\eta^{d}}+\frac{M}{\kappa^{d}}$, the dogmatic modeler's behavior satisfies that $a_{t}=\hat a$ and ${\pi}_{t}^{\theta}\left(\Omega^{\theta}\left(\sigma\right)\right)>\gamma,\forall t\geq0$. Note that when $a_t=\hat a$ for all $t\geq 0$, we have $\hat q(\cdot|a_t,\hat\omega)\equiv q(\cdot|a_t,\hat\omega)$ for all $t\geq 0$. Using the inequality $P(AB)\geq P(A)+P(B)-1$ for random events $A$ and $B$, we know that when $\epsilon<\min\{\epsilon_1,\epsilon_2,\epsilon_3\}$, there exists a full-support  prior $\pi^\theta_0$ such that the following event happens with positive probability:
\begin{align*}
    & a_{t}=\hat a \text{ and }{\pi}_{t}^{\theta}\left(\Omega^{\theta}\left(\sigma\right)\right)>1-\epsilon,\forall t\geq0 \\
    &\frac{\sum_{\omega'\in\Omega^{\theta'}} \pi^{\theta'}_0(\omega') \ell_t(q,\omega')}{\ell_t( q,\hat\omega)}\leq \eta+(1-\gamma+\epsilon)\kappa,\forall t\geq0
\end{align*}
When $\epsilon$ is small enough and $\gamma$ is sufficiently close to $1$, conditional on the above event, we have
\begin{align*}
    \frac{\ell_t(\theta')}{\ell_t(\theta)} < \frac{\sum_{\omega'\in\Omega^{\theta'}} \pi^{\theta'}_0(\omega') \ell_t(q,\omega')}{\pi^\theta_0 (\Omega^\theta(\sigma)) \ell_t( q,\hat\omega)}<\frac{\eta+(1-\gamma+\epsilon)\kappa}{\gamma}<\alpha.
\end{align*}
Hence, conditional on this event, the switcher never switches to the competing model $\theta'$. It follows that the switcher adopts $\theta$ forever with positive probability w.r.t. $\mathbb P_S$.

\paragraph{Part (ii).} Suppose in addition that model $\theta$ has no traps. Denote by $\hat\omega$ the consistent parameter associated with the pure p-absorbing BNE at which $\theta$ is locally dominant within the $\overline\theta$-family. A similar argument as in the proof of Theorem \ref{thm:concentrated-prior} implies that for all $\gamma\in(0,1)$, starting from any full-support prior $\pi_0^\theta$, there exists $\epsilon_4$ and a positive number $T$ such that when $\epsilon<\epsilon_4$ and the competing model $\theta'$ is $\epsilon$-close to $\theta$, the Bayes factor never exceeds $\sqrt{\alpha}$ before period $T$ and $\pi_T^\theta(\hat\omega)>\gamma$. 

On the other hand, we have the following lemma, which mirrors Lemma \ref{lem:quasi-strictness} but applies to a strict BN-E. The proof uses similar argument and thus is omitted.
\begin{lem}
Suppose $\hat a$ is a strict BN-E with supporting belief $\delta_{\hat\omega}$, then for any $\zeta\in (0,1)$, there exists $\gamma\in(0,1)$ such that starting from any prior $\pi_0^\theta\in B_{1-\gamma}(\delta_{\hat\omega})$, the probability that the $\theta$-modeler always actions in $\supp(\sigma)$ for all periods is strictly larger than $\zeta$.
\end{lem}
Following the same steps in the proof of Part (i), we can show that there exists a small enough $\gamma$, whose value depends only on the characteristics of model $\theta$ and $\epsilon$ but not the specific competing model, such that with any belief $\pi_T^\theta\in B_{1-\gamma}(\delta_{\hat\omega})$, the switcher adopts model $\theta$ for all periods after $T$ with positive probability. Therefore, model $\theta$ is $\overline\theta$-constrained locally robust at all priors.

\subsection{Proof of Theorem \ref{thm:multiple-competing-models}}
To show that Theorem \ref{thm:Global-Robustness} continues to hold when $\alpha>K$, it suffices to show that a model $\theta$ is globally robust at some prior by the new definition if $\theta$ admits a p-absorbing SCE. Without loss of generality, take any $\Theta'=\{\theta^1,...,\theta^K\}\subseteq\Theta$ and define for each $k\in\{1,...,K\}$ a process $\{S_t^k\}_t$ as follows,

$$S_t^k = \frac{\sum_{\omega^\prime\in\Omega^{\theta^k}}\pi_{0}^{\theta^{k}}\left(\omega^{\prime}\right) \prod_{\tau=0}^t q^{\theta^k} (y_\tau|a_\tau,\omega^\prime)}{\prod_{\tau=0}^t q^\ast (y_\tau|a_\tau)}.$$ Then for any $\eta\in(1,\alpha)$, we have 
$$\mathbb P_D(S_t^k\leq \eta,\forall t\geq 0)\geq 1- \frac{\mathbb E^{\mathbb P_D} S_0^k}{\eta}= 1-\frac{1}{\eta}.$$ Hence, when $\eta$ is sufficiently close to $\alpha$,
\begin{align*}
    &\mathbb P_D(S_t^k\leq \eta,\forall t\geq 0,\forall k\in\{1,...,K\})  \\
    \geq & 1-\sum_{k=1}^K P_B(S_t^k> \eta\text{ for some } t\geq 0)\\
    \geq & 1-\frac{K}{\eta}>0. 
\end{align*}
The rest of the argument is identical to the proof in Section \ref{proof:thm:Global-Robustness}. It follows that Theorem \ref{thm:Global-Robustness} also continues to hold when $\alpha>K$. 

For $\bar\theta$-constrained local robustness, note that Inequality \eqref{eq:thm4-lkr}  implies 
\begin{align*} 
    & \mathbb P_D\Bigg(\frac{\sum_{\omega^k\in\Omega^{\theta^k}} \pi^{\theta^k}(\omega^k) \ell_t(q,\omega^k)}{\ell_t(\hat q,\omega^\ast)}\leq \pi^{\theta^k}(\Omega^\theta(\sigma))+\pi^{\theta^k}(\iota(\Omega^\theta(\sigma)))\eta \\
    &\hspace{13em} +\pi^{\theta^k}(\iota(\Omega^\theta\setminus\Omega^\theta(\sigma))) \kappa ,\forall t\geq 0,\forall k\in\{1,...,K\}\Bigg)  \\ 
    \geq & 1-\sum_{k=1}^K \mathbb P_D\Bigg(\frac{\sum_{\omega^k\in\Omega^{\theta^k}} \pi^{\theta^k}(\omega^k) \ell_t(q,\omega^k)}{\ell_t(\hat q,\omega^\ast)}> \pi^{\theta^k}(\Omega^\theta(\sigma))+\pi^{\theta^k}(\iota(\Omega^\theta(\sigma)))\eta \\
    &\hspace{13em} +\pi^{\theta^k}(\iota(\Omega^\theta\setminus\Omega^\theta(\sigma))) \kappa\text{ for some }t\geq 0\Bigg) \\
    \geq & 1-K\left(\frac{1}{\eta^{d}}+\frac{M}{\kappa^{d}}\right).
\end{align*}
If $K<\alpha^{d}$, then the term above is strictly positive when $\eta$ is sufficiently close to $\alpha$ and $\kappa$ is sufficiently large. The rest of the proof is analogous to the proof in Section \ref{subsec:Constrained-LR-proof}.
\hfill $\Box$

\subsection{Proof of Proposition \ref{prop: media-bias-2}}

It suffices to show that the agent makes a switch to $\hat\theta$ with positive probability. It then follows from Proposition \ref{prop: media-bias-1} that $\hat\theta$ is eventually adopted forever with positive probability. 

Define a new probability measure $\hat{\mathbb P}$ over the action and outcome histories $H$ such that for any histories $\hat H\subset H$,
$$\hat{\mathbb P}\left(\hat H\right)=\pi_0^{\hat\theta}(\omega^L)\mathbb{P}_{S}^{\hat\theta,\omega^L}\left(\hat H\right)+\pi_0^{\hat\theta}(\omega^R)\mathbb{P}_{S}^{\hat\theta,\omega^R}\left(\hat H\right),$$ where $\mathbb{P}_{S}^{\hat\theta,\omega}$ is the probability measure over histories induced by the agent switcher if the true DGP is identical to the DGP prescribed by $\hat\theta$ and $\omega$. Then 
$\ell_t(\theta,\omega^M)/\ell_t(\hat\theta)$ is a martingale w.r.t. $\hat{\mathbb P}$
with an expectation of $1$. Hence, for any $\eta>1$, the probability that $\ell_t(\theta,\omega^M)/\ell_t(\hat\theta)\leq \eta$ for all $t$ is positive (measured by $\hat{\mathbb P}$). Since $a^M$ is the only SCE under model $\theta$, by Lemma \ref{lem:Correctly-full-characterization} the agent almost surely eventually play $a^M$ on the paths where the model choice eventually equals $\theta$. If so, the agent's posterior $\pi_t^\theta$ almost surely converges to $\delta_{\omega^M}$. In summary, on paths where $m_t$ eventually equals $\theta$, it happens with positive probability (measured by $\hat{\mathbb P}$) that $\ell_t(\theta,\omega^M)/\ell_t(\hat\theta)\leq \eta$ for all $t$ and $\pi_t^\theta\xrightarrow{\text{a.s.}} \delta_{\omega^M}$. This then implies that for any $\epsilon>0$, we can construct a finite sequence of outcome realizations $(y_0,...,y_{t-1})$ such that $\ell_t(\theta,\omega^M)/\ell_t(\hat\theta)\leq \eta$ for all $t\leq T$ and $\pi_T^\theta\in B_\epsilon(\delta_{\omega^M})$. Moreover, since $T$ is finite, this sequence of outcomes are also realized with positive probability under the true measure $\mathbb P_S$. Notice that 
\begin{align*}
    \frac{\ell_T(\hat\theta)}{\ell_T(\theta)}   = \pi_T^{\theta}(\omega^M)\frac{\ell_T(\hat\theta)}{\pi_0^\theta(\omega^M)\ell_t(\theta,\omega^M)} 
     \geq (1-\epsilon)\frac{\eta}{\pi_0^\theta(\omega^M)},
\end{align*}
where the right-hand side is strictly larger than $\alpha$ when $\pi_0^\theta(\omega^M)<1/\alpha$ if $\epsilon$ is close enough to $0$ and $\eta$ is close enough to 1. Therefore, the agent makes a switch from $\theta$ to $\hat\theta$ with positive probability.

\subsection{Proof of Proposition \ref{prop:application-2}}

Suppose the agent's action space contains $K$ elements, $a^1<a^2<...<a^K$. Define function $h:[\underline\omega,\overline\omega]\rightarrow[\underline\omega,\overline\omega]$, such that $h(\omega)$ returns the KL-minimizer evaluated at the largest myopically optimal action against the degenerate belief $\delta_\omega$ i.e. $h(\omega)$ minimizes $D_{KL} \left(q^\ast(\cdot|\overline a(\omega))\parallel q(\cdot|\overline a(\omega),\hat b,\omega)\right)$ where $\overline a(\omega)=\max A^\theta(\delta_\omega)$. By Assumption \ref{assu:application-2}, there exists an increasing sequence of intervals $\{(\omega_k,\omega_{k+1})\}^K_{k=0}$ such that $\omega_0=\underline \omega$, $\omega_K=\overline\omega$, $a^k$ is the unique myopically optimal action over $(\omega_{k-1},\omega_k)$ and both $a^{k-1}$ and $a_k$ are myopically optimal at $\omega_{k-1}$. Function $h$ is flat within each interval. If there exists a pure BN-E under model $\theta$, then it must be supported by a degenerate belief at $\omega$ such that $h(\omega)=\omega$. By Assumption \ref{assu:application-2}, any pure BN-E must also be self-confirming, and any mixed BN-E cannot be self-confirming. 

Suppose $\hat b>b^\ast$, then $h$ jumps up discontinuously at all cutoffs $\{\omega_k\}_{1\leq k\leq K-1}$. Suppose there exists no solution to $h(\omega)=\omega$. Then since $h(\underline\omega)\geq\underline\omega$ and $h(\overline\omega)\leq\overline\omega$, we know that there must exist $\hat k$ such that $h(\omega)>\omega$ for all $\omega\in(\omega_{k^\ast-1},\omega_{k^\ast})$ and $h(\omega^\prime)<\omega^\prime$ for all $\omega^\prime\in(\omega_{k^\ast},\omega_{k^\ast+1})$. But this contradicts the observation that $h$ jumps up at $\omega_{k^\ast}$. It also immediately follows that there exists a solution $\hat \omega$ to $h(\hat \omega)=\hat \omega$ such that  $h(\omega^\prime)>\omega^\prime$ for $\omega^\prime<\hat \omega$ and   $h(\omega^\dprime)<\omega^\dprime$ for $\omega^\dprime<\hat \omega$. Let $\hat a$ be the unique myopically optimal action at $\delta_{\hat\omega}$. Then $\hat a$ is a pure self-confirming equilibrium, supported by the generate belief at $\hat \omega$. By assumption \ref{assu:application-2}, $\omega\in\Omega^\theta$, and thus $\hat a$ is also a self-confirming equilibrium under $\theta$. Note that $\hat a$ is uniformly strict. By Corollary \ref{corr:Suff-robustness}, model $\theta$ is globally robust. 

Now suppose the agent is underconfident, then $h$ jumps down discontinuously at the cutoffs $\{\omega_k\}_{1\leq k\leq K-1}$. Hence, there exists at most one solution to $h(\omega)=\omega$. Suppose there exists a SCE $\sigma^\dagger$ when the agent believes his ability is given by $\tilde b$. Then by the upper-hemicontinuity of $A^\theta$, when $\hat b$ is lower than but sufficiently close to $\tilde b$, there exists some $\hat \omega>\omega^\ast$ such that $g(a^\dagger,\hat b,\hat\omega)=g(a^\dagger,b^\ast,\omega^\ast)$, where $a^\dagger=\max\supp(\sigma^\dagger)$ and is the unique myopically optimal action against $\delta_{\hat\omega}$. It follows that $a^\dagger$ is a uniformly strict SCE under $\theta$. Since there always exists a SCE when the agent is correctly specified, i.e. $\tilde b=b^\ast$, we infer that model $\theta$ is globally robust when $b^\ast-\hat b$ is sufficiently small. 

Suppose instead that there is no solution to $h(\omega)=\omega$ when the agent's self-perception is given by $\hat b$. If so, there exists no SCE under model $\theta$. By Theorem \ref{thm:Global-Robustness}, $\theta$ is not globally robust. By continuity, $h(\omega)=h(\omega)$ also does not admit any solution at $\tilde b$ if it is sufficiently close to $\tilde b$. Therefore, there exists an open neighborhood around $\hat b$ such that model $\theta$ is not globally robust. I now show that $\theta$ is not $\bar\theta$-constrained locally robust. Suppose $\theta$ admits a mixed BN-E $\hat \sigma$, supported by a potentially mixed belief $\hat\pi^\theta\in\Delta\Omega^\theta$. Suppose $\hat\sigma$ takes support over $a^k$ and $a^{k+1}$ (note that both $\hat\sigma$ and $\hat\pi^\theta$ have at most two elements in their support). Then for any $\hat\omega\in\supp(\pi^\theta)$, we have 
$$\hat\omega\in\arg\min_{\Omega^\theta} \sigma(a^k) D_{KL}\left(q^\ast(\cdot|a^k)\parallel q^\theta (\cdot|a^k,\omega^\prime)\right)+\sigma(a^{k+1}) D_{KL}\left(q^\ast(\cdot|a^{k+1})\parallel q^\theta (\cdot|a^{k+1},\omega^\prime)\right).$$
Let $\omega^k$ denote the KL-minimizer in $[\underline\omega,\overline\omega]$ at $a^k$ and $\omega^{k+1}$ denote the KL-minimizer in $[\underline\omega,\overline\omega]$ at $a^{k+1}$. Since $g(a,\hat b,\omega)-g(a,b^\ast,\omega^\ast)$ is strictly increasing in $a$ when $\hat b<b^\ast$ and $\omega>\omega^\ast$, we know that $\omega^k>\omega^{k+1}$. Since $\omega^k,\omega^{k+1}\in\Omega^\theta$ by assumption, we have $\supp(\hat\pi^\theta)\subseteq [\omega^{k+1},\omega^k]$. Hence, for all $\hat\omega\in\supp(\hat\pi^\theta)$,
\begin{align*}
    g(a^k,\hat b,\hat\omega)-g(a^k,b^\ast,\omega^\ast)&\leq0 \\
    g(a^{k+1},\hat b,\hat\omega)-g(a^{k+1},b^\ast,\omega^\ast)&\geq0,
\end{align*}
with at least one inequality being strict. Suppose the second inequality is strict. Pick $\tilde b\in(\hat b,b^\ast]$ and define $\tilde \omega$ by $g(a^k,\hat b,\hat\omega)=g(a^k,\tilde b,\tilde \omega)$. Similarly, since $g(a,\hat b,\omega)-g(a,\tilde b,\tilde \omega)$ is strictly increasing in $a$, we know that $g(a^{k+1},\hat b,\hat\omega)>g(a^{k+1},\tilde b,\tilde \omega)$; analogously $g(a^{k+1},\hat b,\hat\omega)-g(a^{k+1},b^\ast,\omega^\ast)\geq0$. Therefore, the $\hat\sigma$-weighted KL divergence is smaller at $(\tilde b,\tilde\omega)$ than at $(\hat b,\hat\omega)$. When $\tilde b$ is sufficiently close to $\hat b$, the parameter pair $(\tilde b,\tilde\omega)$ is also close to $(\hat b,\hat\omega)$. Since the agent's action frequency converges, by Theorem \ref{thm:Local-lobustness-w-convergence}, model $\theta$ is not $\bar\theta$-constrained locally robust at $\hat b$. Moreover, for any $\tilde b\in(\hat b,b^\ast]$, there exists a competing model $\theta'$ with $(\tilde b,\tilde \omega)\in\Omega^{\bar\theta}$ such that $\theta$ does not persist against $\theta'$. 

Similarly, model $\theta$ is not $\bar\theta$-constrained locally robust if we perturb the value of $\hat b$. Combining this with the previous observation, we could find a sequence of intervals such that model $\theta$ is either globally robust or not $\bar\theta$-constrained locally robust, each occuring within disjoint intervals. \hfill $\Box$

\section{Supplemental Appendix\label{sec:Online-Appendix}}

\subsection{Examples Omitted from Section \ref{sec:Robustness}}

\begin{example}[A p-absorbing mixed SCE] \label{exa:p-absorbing-mixed-SCE}Consider a dogmatic modeler's problem, where there are two actions $\mathcal A=\{1,2\}$ and three parameters $\Omega^\theta=\{1,1.5,2\}$ inside the parameter space of model $\theta$. The agent's payoff is simply the outcome $y_t$, with the true DGP being the normal distribution $N(0.25,1)$ for all actions. Model $\theta$ is misspecified, predicting that $y_t\sim N((\omega-a_t)^2,1)$. Note that every mixed action is a self-confirming equilibrium, with the supporting belief assigning probability $1$ to the parameter value of $1.5$. Here, every mixed SCE is p-absorbing since its support contains every action that can be played by the agent. But her action sequence may never converge. To see that, notice that a belief that assigns larger probability to $\omega=1$ than $\omega=2$ leads to action $a=2$, but such play in turn induces her to attach lower probability to $\omega=1$ than $\omega=2$ and leads to action $a=1$. Nevertheless, Corollary \ref{corr:Suff-robustness} tells us that the aforementioned SCE is indeed p-absorbing.

\end{example}

\begin{example}[A self-confirming equilibrium that fails to be p-absorbing]\label{exa:not-p-absorbing-SCE} Consider a dogmatic modeler's problem, where there are two actions $\mathcal A=\{1,3\}$ and three parameters $\Omega^\theta=\{1,2,3\}$ inside the parameter space of model $\theta$. The agent's payoff is the absolute value of the outcome, $|y_t|$, with the true DGP of $y_t$ given by a normal distribution $N(1,1)$ for all actions. Consider a misspecified model $\theta$ that predicts $y_t\sim N(\omega-a_t,1)$. Note that $\theta$ admits a single self-confirming equilibrium in which the agent plays $a^\ast=1$ with probability $1$, supported by a belief that assigns probability $1$ to $\omega^\ast=2$. However, this SCE is not p-absorbing. To see that, notice that the agent is indifferent between the two actions when the parameter takes the value of $2$. When the agent keeps playing $a=1$, the parameters $1$ and $3$ fit the data equally well on average, so their log-posterior ratio is a random walk which a.s. crosses $1$ infinitely often. However, the high action $a=3$ is strictly optimal against any belief that assigns a higher probability to $\omega=1$ than $\omega=3$. Hence, the high action must be played infinitely often almost surely.
\end{example}

\subsection{Micro-Foundation for Application \ref{app:media bias}}

In this subsection I specify the payoff structure for the news consumption problem in Application \ref{app:media bias}, which provides a micro-foundation for Assumption \ref{assu:media bias}. 

To do this, we first extend the learning framework introduced in Section \ref{sec:Framework} to allow for an unobserved payoff that may depend on an unknown state. That is, besides the observable payoff jointly determined by the action and the random outcome $u(a_t,y_t)$, there may exist an unobserved payoff $\tilde u(a_t,\omega)$ that depends on the action and a fundamental state $\omega\in\Omega$. Under any subjective model $\theta$, the agent maximizes the sum of the observed and the unobserved payoff given her belief over the fundamental state and possibly other parameters. This maximization gives rise to an optimal-action correspondence $A^\theta_m:\Delta\Omega^\theta\rightrightarrows\mathcal A$, which we can use to define a self-confirming equilibrium. All results in Section \ref{sec:Robustness} remain unchanged.

Subscribing to media outlets provide entertainment value. Media outlets produce higher quality news reports if the story is aligned with their political leaning. If the agent subscribes to media $a^L$, she earns an emotional utility of $1$ iff she receives a $l$ story; similarly, if she subscribes to media $a^R$, she earns an emotional utility of $1$ iff she receives a $r$ story. If she subscribes to the neutral media $a^M$, she earns a constant emotional payoff of $0.65$.

Subscribing to media outlets also provide valuable information. In additional to subscribing to a media outlet $a_t$, the agent takes an outside action $v_t\in\{v^L,v^M,v^R\}$ upon receiving the story $y_t$. The agent earns a payoff of $1$ if she takes $v^L$ in state $\omega^L$ and $v^R$ in state $\omega^R$, but in state $\omega^M$ she earns a constant payoff of $0.5$ by taking any action. Note that it is optimal for the agent to follow the story she receives in each period.

In Table \ref{table:media-stories-payoffs}, I summarize the expected total payoffs associated with each action under model $\theta$ and model $\theta'$. It is then straightforward to verify Assumption \ref{assu:media bias}.

\begin{table}[t]\label{table:media-stories-payoffs}
    \centering

\begin{tabular}{c|c c c}
     $E^\theta(\text{payoff}|a,\omega)$ & $\omega^L$ & $\omega^M$ & $\omega^R$  \\ \hline
     $a^L$ & 1.4 & 1.1 & 1 \\
     $a^M$ & 1.25 & 1.15 & 1.25 \\
     $a^R$ & 1 & 1.1 & 1.4
\end{tabular}
\quad\quad
\begin{tabular}{c|c  c}
     $E^{\hat \theta}(\text{payoff}|a,\omega)$ & $\omega^L$ & $\omega^R$  \\ \hline
     $a^L$ & 1.2 & 1 \\
     $a^M$ & 1.15 & 1.15 \\
     $a^R$ & 1 & 1.2
\end{tabular}


\caption{Expected payoffs under model $\theta$ (left) and expected payoffs under model $\theta'$ (right).}
\end{table}

\subsection{Examples Omitted from Section \ref{sec:extension}}
I provide two examples below to substantiate the observation in Footnote \ref{fn:persistence-against-each-model}. Example \ref{exa:persistence-single-model} presents a scenario in which $\theta$ persists against $\theta^1$ and $\theta^2$ separately but does not persist against $\{\theta^1,\theta^2\}$, while Example \ref{exa:persistence-two-models} shows an opposite scenario.
\begin{example}
\label{exa:persistence-single-model}Let $x^{1}$ and $x^{2}$ be two i.i.d. normally distributed variables, both with mean $0$ and variance $1$. Suppose $x^{3}$ and $x^{4}$ are also i.i.d. normally distributed but with mean $1$ and variance $1$. Suppose the agent can play one of two actions in each period, $\mathcal{A}=\left\{ 1,2\right\} $ and uses subjective models to learn about the mean of each element in $(x^1,x^2,x^3,x^4)$. Her flow payoff is given by $a\cdot\left(x^{4}-x^{3}\right)$. Hence, she would like to play $a=2$ if $\overline{x}^{4}>\overline{x}^{3}$ and play $a=1$ if $\overline{x}^{3}>\overline{x}^{4}$. However, $x^{1}$ and $x^{3}$ are only observable when $a=1$, while $x^{2}$ and $x^{4}$ are only observable when $a=2$. That is, the outcome $y$ is given by $(x^1,x^3)$ when $a=1$ and given by $(x^2,x^4)$ when $a=2$. She entertains an initial model $\theta$ and two competing models, $\left\{ \theta^{1},\theta^{2}\right\} $, each of which is equipped with a binary parameter space. The predictions of each model are summarized by the following table. The predicted means are independent of the actions taken. 
\[
\begin{array}{ccc}
\theta & \omega^{1} & \omega^{2}\\
\left(\overline{x}^{1},\overline{x}^{2},\overline{x}^{3},\overline{x}^{4}\right) & \left(1,1,1,0\right) & \left(1,1,0,1\right)
\end{array}
\]

\[
\begin{array}{ccc}
\theta^{1} & \omega^{1\prime} & \omega^{2\prime}\\
\left(\overline{x}^{1},\overline{x}^{2},\overline{x}^{3},\overline{x}^{4}\right) & \left(1,0,1,0\right) & \left(1,0,0,1\right)
\end{array}
\]

\[
\begin{array}{ccc}
\theta^{2} & \omega^{1\prime\prime} & \omega^{2\prime\prime}\\
\left(\overline{x}^{1},\overline{x}^{2},\overline{x}^{3},\overline{x}^{4}\right) & \left(0,1,1,0\right) & \left(0,1,0,1\right)
\end{array}
\]

Notice that there are two strict (and thus p-absorbing) Berk-Nash equilibria under $\theta$: (1) $a=1$ is played w.p. 1, supported by the belief that assigns probability 1 to $\omega^{1}$; (2) $a=2$ is played w.p. 1 , supported by the belief that assigns probability 1 to $\omega^{2}$. First observe that $\theta$ persists against $\theta^{1}$ at a prior $\pi_{0}^{\theta}$ that assigns sufficiently high belief to $\omega^{1}$. This follows from the fact that the likelihood ratio between $\theta$ and $\theta^{1}$ is always $1$ when $a=1$ is played, and that the equilibrium is p-absorbing. Analogously, $\theta$ persists against $\theta^{2}$ at a prior $\pi_{0}^{\theta}$ that assigns sufficiently high belief to $\omega^{2}$. However, notice that $\theta$ does not persist against $\left\{ \theta^{1},\theta^{2}\right\} $ at any priors and policies, because regardless of the actions taken by the agent, at least one of $\theta^{1}$ and $\theta^{2}$ would fit the data strictly better than $\theta$, prompting the agent to adopt $\theta^{1}$ and $\theta^{2}$ infinitely often.
\end{example}

\begin{example}
\label{exa:persistence-two-models} Let $y$ be a normally distributed variable with mean 0 and variance 1, whose distribution is independent of actions. The agent can play one of two actions in each period, $\mathcal{A}=\{1,2\}$ and uses subjective models to learn about the mean of $y$. Her flow payoff is given by $a\cdot y$. She entertains an initial model $\theta$ and two competing models, $\{\theta^1,\theta^2\}$. Model $\theta^1$ has a single parameter and perfectly matches the true DGP, while models $\theta$ and $\theta^2$ both have a binary parameter space. The predictions about $\overline{y}$ of each model are summarized by the following table. 

\begin{center}
    
\begin{tabular}{c|cc}
$\theta$ & $\omega^{1}$ & $\omega^{2}$\\ \hline
 $a^1$ & -1 & 1 \\
 $a^2$ & -2 & 1
\end{tabular}
\quad
\begin{tabular}{c|cc}
$\theta^1$ & $\omega^{\prime}$ & \\\hline
$a^1,a^2$ & -1 &   \\
\end{tabular}
\quad
\begin{tabular}{c|cc}
$\theta^2$ & $\omega^{\prime\prime}$ &\\\hline
$a^1,a^2$ & 2\\
\end{tabular}

\end{center}

Suppose the agent's prior satisfies that $\pi_0^\theta(\omega^1)=1-\pi_0^\theta(\omega^0)=\frac{0.5}{\alpha}<\frac{1}{\alpha}$. First consider what happens when the agent has only one competing model, $\theta^1$. By the Law of Large Numbers, the likelihood ratio between $\theta^1$ and $\theta$ eventually exceeds $\alpha$ almost surely because 
\begin{align*}
    \frac{\ell_t(\theta^1)}{\ell_t(\theta)} & =\frac{\prod_{\tau=0}^{t-1}q^{\theta^{1}}\left(y_{\tau}|a_{\tau},\omega^{\prime}\right)}{\prod_{\tau=0}^{t-1}q^{\theta}\left(y_{\tau}|a_{\tau},\omega^{1}\right)\pi_{0}^{\theta}\left(\omega^{1}\right)+\prod_{\tau=0}^{t-1}q^{\theta}\left(y_{\tau}|a_{\tau},\omega^{2}\right)\pi_{0}^{\theta}\left(\omega^{2}\right)}\\ 
    & = \frac{\prod_{\tau=0}^{t-1}q^{\ast}\left(y_{\tau}|a_{\tau}\right)}{\prod_{\tau=0}^{t-1} \boldsymbol{1}_{(a_\tau=a^1)} q^{\ast}\left(y_{\tau}|a_{\tau}\right)\pi_{0}^{\theta}\left(\omega^{1}\right) + \xi (h_t)} \\
    & \geq \frac{\prod_{\tau=0}^{t-1}q^{\ast}\left(y_{\tau}|a_{\tau}\right)}{\prod_{\tau=0}^{t-1} \frac{1}{\alpha} q^{\ast}\left(y_{\tau}|a_{\tau}\right) + \xi (h_t)}
\end{align*}
where $\frac{\xi (h_t)}{\prod_{\tau=0}^{t-1}q^{\ast}\left(y_{\tau}|a_{\tau}\right)}$ converges to $0$ almost surely. Threrefore, $\theta$ does not persist against $\theta^1$ under prior $\pi^\theta_0$. 

However, model $\theta$ persists against $\Theta'\coloneqq\{\theta^1,\theta^2\}$ at prior $\pi^\theta_0$. First notice that for any $a_0\in\mathcal{A}$, there exists some $y_0$ sufficiently large such that $$\ell_1(\theta^2) > \alpha \cdot \max \{ \ell_1(\theta),\ell_1({\theta^1})\}$$ and thus the agent switches to $\theta^2$ in the beginning of period $1$. As a result, the agent plays $a_1=a^2$ in period $1$ since it is the strictly dominant strategy under $\theta^2$. But then we could find some sufficiently small $y_1$ such that the following two inequalities hold: 
$$\ell_2({\theta}) >  \alpha \cdot \max \{ \ell_2(\theta^1),\ell_2(\theta^2)\},$$
$$\pi^\theta_2 (\omega^1)=\frac{\pi^\theta_0 (\omega^1) q^\theta (y_0|a_0,\omega^1) q^\theta (y_1|a_1,\omega^1)}{\sum_{\omega\in\{\omega^1,\omega^2\}}\pi^\theta_0 (\omega) q^\theta (y_0|a_0,\omega) q^\theta (y_1|a_1,\omega)}>\max\{\frac{1}{\alpha},c\},$$
where $c$ is chosen such that while adopting $\theta^1$, the agent finds $a^1$ payoff-maximizing if her belief assigns a probability higher than $c$ to $\omega^1$, i.e. $\pi^\theta (\omega^1)>c$. The first inequality implies that the agent switches back to $\theta$ in the beginning of period $2$. The second inequality, together with the observation that the pure strategy $a^1$ is a uniformly strict self-confirming equilibrium supported by the belief that assigns probability $1$ to $\omega^1$, ensures that with positive probability, the agent plays action $a^1$ forever, provided that her decisions are made based on model $\theta$. But notice that on those paths, the agent indeed no longer switches to other models after period $1$ because for $t>2$,
\begin{align*}
    \frac{\ell_t(\theta^1)}{\ell_t(\theta)} <\frac{\ell_2(\theta^1)}{\ell_2(\theta)} \frac{\prod_{\tau=2}^{t-1}q^{\ast}\left(y_{\tau}|a_{\tau}\right)}{\prod_{\tau=2}^{t-1}  q^{\ast}\left(y_{\tau}|a_{\tau}\right)\pi_{2}^{\theta}\left(\omega^{1}\right)} <1<\alpha. 
\end{align*}

Since outcomes $y_0$ and $y_1$ that satisfy the aformentioned properties are drawn with positive probability, we conclude that $\theta$ persists against $\Theta^c\coloneqq\{\theta^1,\theta^2\}$ at prior $\pi^\theta_0$.
\end{example}

\end{document}